\documentclass{svjour3}
\usepackage{amsmath,amsfonts,amssymb}
\usepackage{mathtools}
\usepackage{relsize,lscape}    
\newcommand*{\id}{{\mathrm{id}}}

\newcommand*{\la}{{\langle}}                                                  
\newcommand*{\ra}{{\rangle}}                                                  
                                                        
\newcommand*{\R}{{\mathbb R}}

\newcommand*{\CC}{{\mathbb{C}}}

\newcommand*{\Vb}{{\mathbb V}}

\newcommand*{\Jf}{{\mathfrak J}}

\newcommand*{\pa}{{\partial}}

\newcommand*{\Cb}{{\mathcal C}}
\newcommand*{\rd}{{\mathrm d}}

\newcommand*{\mfd}{{\mathfrak d}}

\newcommand*{\bs}{\boldsymbol}
\newcommand*{\lb}{[\![}
\newcommand*{\rb}{]\!]}
\newcommand{\smlsym}[1]{\mathsmaller{\mathsmaller{#1}}}
\newcommand*{\ot}{\smlsym\otimes}
\newcommand*{\ob}{\hat{\smlsym\otimes}}
\smartqed
\journalname{}
\allowdisplaybreaks

\begin{document}
\title{The algebraic structure of the non-commutative nonlinear Schr\"odinger and modified Korteweg--de Vries hierarchy}
\titlerunning{Nonlinear Schr\"odinger hierarchies}
\author{Gordon Blower$^1$ \and Simon~J.A.~Malham$^2$}    
\authorrunning{Blower and Malham}

\institute{$^1$ Department of Mathematics and Statistics, Lancaster University, Lancaster LA1 4YF, UK\\
$^2$ Maxwell Institute for Mathematical Sciences,        
and School of Mathematical and Computer Sciences,   
Heriot-Watt University, Edinburgh EH14 4AS, UK\\
\email{G.Blower@lancaster.ac.uk, S.J.A.Malham@hw.ac.uk}}
\date{11th March 2023}           
\maketitle

\begin{abstract}
  We prove that each member of the non-commutative nonlinear Schr\"odinger and modified Korteweg--de Vries hierarchy is a Fredholm Grassmannian flow,
  and for the given linear dispersion relation and corresponding equivalencing group of Fredholm transformations,
  is unique in the class of odd-polynomial partial differential fields.
  Thus each member is linearisable and integrable in the sense that time-evolving solutions can be generated by solving a linear Fredholm Marchenko equation,
  with the scattering data solving the corresponding linear dispersion equation.
  At each order, each member matches the corresponding non-commutative Lax hierarchy field which thus represent odd-polynomial partial differential fields.
  We also show that the cubic form for the non-commutative sine--Gordon equation corresponds
  to the first negative order case in the hierarchy, and establish the rest of the negative order non-commutative hierarchy. 
  To achieve this, we construct an abstract combinatorial algebra, the P\"oppe skew-algebra, that underlies the hierarchy.
  This algebra is the non-commutative polynomial algebra over the real line generated by compositions,
  endowed with the P\"oppe product---the product rule for Hankel operators pioneered by Ch. P\"oppe for classical integrable systems.
  Establishing the hierarchy members at non-negative orders, involves proving the existence of a `P\"oppe polynomial' expansion for basic compositions
  in terms of `linear signature expansions' representing the derivatives of the underlying non-commutative field.
  The problem boils down to solving a linear algebraic equation for the polynomial expansion coefficients, at each order.
\end{abstract}

\section{Introduction}\label{sec:intro}
For the non-commutative nonlinear Schr\"odinger and modified Korteweg--de Vries hierarchy, we prove that each member is both,
a Fredholm Grassmannian flow, and for the given linear dispersion relation and corresponding equivalencing group of Fredholm transformations,
is unique in the class of odd-polynomial partial differential fields.
That each member represents a Fredholm Grassmannian flow means they are linearisable in the sense that solutions
can be generated by solving a linear Fredholm Marchenko equation whose scattering data is the solution to
the corresponding linear dispersion relation. We also show that each member of the Lax hierarchy generates an odd-polynomial partial differential field, 
and so by uniqueness, at each non-negative order, the Fredholm Grassmannian flow and Lax hierarchy member are one and the same.
We also show that the cubic form for the non-commutative sine--Gordon equation corresponds
to the first negative order case in the Lax hierarchy, and establish the rest of the negative order non-commutative hierarchy. 
Our approach is inspired by the pioneering work of Ch.\/ P\"oppe. In a sequence of papers, P\"oppe~\cite{P-SG,P-KdV,P-KP},
P\"oppe and Sattinger~\cite{PS-KP}, and Bauhardt and P\"oppe~\cite{BP-ZS} developed the fundamental product rule
for additive Hankel operators and semi-additive operators, in order to establish the integrability and specific solution forms
for classical integrable systems. These included for example, the scalar sine-Gordon equation and Kadomtsev--Petviashvili hierarchy. 
P\"oppe's approach has recently been substantially developed and extended. In particular, Doikou, Malham and Stylianidis~\cite{DMS}
streamlined and extended P\"oppe's approach to the non-commutative Korteweg--de Vries and nonlinear Schr\"odinger equations,
and Malham~\cite{Malham:quinticNLS} extended the approach to the quartic-order, quintic-degree non-commutative nonlinear Schr\"odinger equation. 
Also see Doikou, Malham, Stylianidis and Wiese~\cite{DMSW:AGFintegrable,DMSW:AGFcoagulation}.
Further, Malham~\cite{Malham:KdVhierarchy} developed a simpler form of the P\"oppe algebra constructed herein to
prove that each member of the non-commutative potential Korteweg--de Vries hierarchy is unique in the class of
polynomial partial differential fields and represents a Fredholm Grassmannian flow. 
Blower and Newsham~\cite{BN} developed a systems perspective to P\"oppe's approach, constructing
tau-functions and families of solutions to the Kadomtsev--Petviashvili equation, while
Blower and Doust~\cite{BD} extend this approach to the sinh-Gordon equation.

Let us now briefly outline our approach herein. Consider a non-commutative nonlinear dispersive partial differential equation
for $g=g(x;t)$ of order $n$ of the form,
\begin{equation*}
\pa_tg=-(\mathrm{i}\mathcal I)^{n-1}\pi_n\big(g,\pa g,\pa^2g,\ldots,\pa^{n-2}g,\pa^ng\big),
\end{equation*}
where $\pa=\pa_x$ and $x\in\R$. Here $\pi_n=\pi_n(\cdot)$ is a polynomial of its arguments---a polynomial partial differential field.
It is linear in $\pa^ng$. The diagonal matrix $\mathcal I$ simply has top left block `$-\id$' and bottom right block is `$\id$'.
As we see presently, $g=\lb G\rb$ is the square-integrable kernel of a Hilbert--Schmidt operator $G$. Herein we use the notation $\lb G\rb$
to denote the kernel of an operator $G$.
Suppose $P$ is a Hilbert--Schmidt Hankel operator on the negative real axis satisfying the linear dispersive equation,
\begin{equation*}
\pa_tP=-(\mathrm{i}\mathcal I)^{n-1}\pa^nP. 
\end{equation*}
Note if there are any linear terms in $\pi_n$, we should augment this equation for $P$ with such linear terms on the right. 
In particular we suppose the square-integrable kernel of $P$ has the form $p=p(y+z+x;t)$ for $y,z\in(-\infty,0]$, while $x\in\R$ represents an additive parameter.
The matrix-valued kernel $p$ satisfies the same linear dispersive partial differential equation as that above for $P$; it represents the scattering data. 
The Marchenko equation at the operator level, here has the Fredholm form,
\begin{equation*}
2\mathrm{i}P=G\,(\id+P^2),
\end{equation*}
for the unknown operator $G$. Provided $U\coloneqq(\id+P^2)^{-1}$ exists as a Fredholm operator,
then the solution Hilbert--Schmidt operator $G=2\mathrm{i}PU$ to this Marchenko equation,
parametrises a Fredholm Grassmannian flow of subspaces spanned by linear dispersive solutions $p$.
See, for example, Doikou \textit{et al.}~\cite{DMSW:AGFintegrable}. P\"oppe's insight was to recognise
the crucial role the Hankel properties of $P$ played in making the connection between the solution $G=\mathrm{i}PU$ to the Marchenko equation,
and that its kernel $\lb G\rb$ satisfies a specific nonlinear dispersive partial differential equation of the form shown above.
The connection is made via the P\"oppe kernel product rule,
\begin{equation*}
\lb F\pa_x(HH^\prime)F^\prime\rb(y,z;x,t)=\lb FH\rb(y,0;x,t)\lb H^\prime F^\prime\rb(0,y;x,t),
\end{equation*}
where $H$ and $H^\prime$ are Hankel operators as described above, and $F$ and $F^\prime$ are any two Hilbert--Schmidt operators.
This rule indicates at the fundamental level, that there is a connection between the matrix products of kernels of operators
of the form $G=2\mathrm{i}PU$ and/or their derivatives (on the right), and kernels of monomials involving operator compositions
of similar objects, but with one order higher derivative (on the left).
Using that $\id+P^2=(\id-\mathrm{i}P)(\id+\mathrm{i}P)$ and setting $V\coloneqq(\id-\mathrm{i}P)^{-1}$,
we see that,
\begin{equation*}
G=2V(\mathrm{i}P)V^\dag\equiv V-V^\dag.
\end{equation*}
Further, we observe that $\pa V=V\pa(\mathrm{i}P)V$, and if we use sub-indicies to denote partial derivatives `$\pa$', we find,
$G_1=V(\mathrm{i}P)_1V-V^\dag(\mathrm{i}P)_1^\dag V^\dag$. In particular, if for any
Hilbert--Schmidt operator $F$ we set $[F]\coloneqq\lb F-F^\dag\rb$, then we observe that the kernels in both these cases are given by, 
\begin{equation*}
\lb G\rb=[V]\qquad\text{and}\qquad\lb G\rb_1=[V(\mathrm{i}P)_1V].
\end{equation*}
It is now easy to imagine that the $n$th partial derivative of $\lb G\rb$ has the form,
\begin{equation*}
  \lb G\rb_n=\sum\chi\bigl(a_1\cdots a_n\bigr)\,\bigl[V(\mathrm{i}P)_{a_1}V\cdots V(\mathrm{i}P)_{a_k}V\bigr],
\end{equation*}
where the sum is over all compositions $a_1a_2\cdots a_k$ of $n$. Naturally $\pa_t\lb G\rb$
is given by $[V\pa_t(\mathrm{i}P)V]$ where $\pa_t(\mathrm{i}P)$ can be expressed in terms $\pa^n(\mathrm{i}P)$ 
using the linear dispersion equation for $P$. Hence our goal is to express $\pa_t\lb G\rb$ in terms
of a polynomial $\pi_n=\pi_n\big(\lb G\rb,\lb G\rb_1,\ldots,\lb G\rb_{n-2},\lb G\rb_n\big)$, linear in $\lb G\rb_n$.
The monomials in the polynomial $\pi_n$ consist of factors of the form $\lb G\rb$, $\lb G\rb_1$, \ldots, $\lb G\rb_{n-2}$,
each expressible as a linear combination of basis elements $[V(\mathrm{i}P)_{a_1}V\cdots V(\mathrm{i}P)_{a_k}V]$
parameterised by compositions as shown above, with the product involved being the P\"oppe kernel product.
If we extend the basis elements to include basis elements of the form $[V(\mathrm{i}P)_{a_1}V\cdots V(\mathrm{i}P)_{a_k}V]$
where any of the $V$ factors shown my be replaced by $V^\dag$, then the operator partial fractions formulae $V=\id+(\mathrm{i}P)V=\id+(\mathrm{i}P)V$
imply that the P\"oppe product generates a closed algebra on such basis elements (see Lemma~\ref{lemma:Poppeproductidentities}). 
The playing field is thus set. It is the algebra of such basis elements equipped with the P\"oppe product.
In fact we use an abstract version of this algebra by stripping the basis elements of their $P$ and $V$ labels 
and focusing on the compositions $a_1a_2\cdots a_k$ and a binary encoding, $\bs0$ and $\bs0^\dag$, of the intervening $V$ or $V^\dag$ factors. 
The P\"oppe product essentially only acts on these components and thus transport our \emph{playing field} to the closed algebra of
basis elements $[\bs0 a_1\bs0 a_2\bs 0\cdots\bs 0 a_k\bs0]$, where any of the $\bs0$'s may be replaced by $\bs0^\dag$.
The $\lb G\rb_n=[V]_n$ are linear combinations of such abstract basis elements (as shown shown above) and we label such
specific linear combinations by $[\bs n]$. The quantity $\pa_t\lb G\rb=\pa_t[V]$ can ultimately be expressed in terms of
$[\bs0 n\bs0]$ or $[\bs0 n\bs 0^\dag]$, respectively depending on whether $n$ is odd or even.
The \emph{game} is to determine if $[\bs0 n\bs0]$ or $[\bs0 n\bs 0^\dag]$ can be expressed
in terms of a polynomial $\pi_n=\pi_n([\bs0],[\bs 1],\ldots,[\bs{n-2}],[\bs n])$, linear in $[\bs n]$.
We express $\pi_n$ as a linear combination of monomials with factors chosen from  
$[\bs0]$, $[\bs 1]$, \ldots, $[\bs{n-2}]$ and a linear term $[\bs n]$. The monomials present in $\pi_n$
are restricted to those that can generate $[\bs0 n\bs0]$ or $[\bs0 n\bs 0^\dag]$ or basis elements involving
compositions $a_1a_2\cdots a_k$ of $n$ via the P\"oppe product. An unknown complex coefficient is associated with 
each monomial in the linear combination. We use the linear expansions for 
$[\bs0]$, $[\bs 1]$, \ldots, $[\bs{n-2}]$ and linear term $[\bs n]$ in terms of the basis elements,
and compute the P\"oppe products of all the expansion basis elements in each of the factors of the monomial.
The result is a large linear combination of basis elements, each with a factor which is a linear combination
of the unknown coefficients. We equate this to $[\bs0 n\bs0]$ or $[\bs0 n\bs 0^\dag]$, depending on whether $n$ is
odd or even, and equate all the coefficients of the basis elements present.
This generates a large linear algebraic system of equations for the unknown coefficients.
Though over-determined, it can be solved for a unique set of coefficients, see Theorem~\ref{thm:main}.

Our Main result thus establishes that for each non-negative integer $n$, there exists a unique
polynomial $\pi_n=\pi_n([\bs0],[\bs 1],\ldots,[\bs{n-2}],[\bs n])$ such that $[\bs0 n\bs0]=\pi_n$,
when $n$ is odd, or $[\bs0 n\bs0^\dag]$, when $n$ is even. The non-commutative Lax hierarchy in our context,
can be iteratively generated, to obtain the next equation in the hierarchy from the previous one by applying
a specific operator to $\pi_n$. Given the existence of a match up to $\pi_n$, we show that applying this specific iterative operator
to $\pi_n$ generates $[\bs0(n+1)\bs0^\dag]$ when $n$ is odd, and generates $[\bs0(n+1)\bs0]$ when $n$ is even.
Since, from our main Theorem~\ref{thm:main}, we know there is a unique polynomial expansion for
$[\bs0 n\bs0]$ when $n$ is odd or $[\bs0 n\bs 0^\dag]$ when $n$ is even for any order, there is a unique
polynomial expansion for $[\bs0(n+1)\bs0^\dag]$ or $[\bs0(n+1)\bs0]$. The uniqueness property means that
the polynomials $\pi_n$ we establish at each order, must match the Lax hierarchy members.
We also investigate pursuing the specific operator in the opposite direction to generate non-commutative Lax hierarchy
members at all negative orders---see Tracy and Widom~\cite{TW} for the scalar case.
In particular we show that the first negative order case $n=-1$, corresponds to the cubic form
of the non-commutative sine-Gordon equation.

The solution to any of the non-commutative hierarchy equations is generated by solving the
linear dispersion equation for the matrix kernel $p$ or equivalently the Hilbert--Schmidt Hankel operator $P$, and then
solving the linear Fredholm Marchenko equation $2\mathrm{i}P=G\,(\id+P^2)$ for $G$. The hierarchy member solution is $\lb G\rb$.
Each member of the hierarchy is thus linearisable as the solution is generated via solving a linear dispersion equation
and a linear Fredholm equation. However this procedure also identifies the graph of $G$ as a Fredholm Grassmannian flow,
represented in a specific coordinate patch parametrised by Hilbert--Schmidt operators.
Such flows are explored in detail in Doikou \textit{et al.} \cite{DMSW:AGFintegrable}.
We can think of the Fredholm Grassmannian as all collections of graphs of compatible linear Hilbert--Schmidt maps.
Indeed briefly, suppose we set $\mathbb V\coloneqq L^2((-\infty,0];\mathbb C^m)$ for some $m\in\mathbb N$, for example.
Consider the pair of operators,
\begin{equation*}
  \begin{pmatrix} \id-Q\\ 2\mathrm{i}P\end{pmatrix},
\end{equation*}
both on $\mathbb V$. Here we suppose $\id-Q$ is a Fredholm operator on $\Vb$, with $Q$ a Hilbert--Schmidt operator,
and $P$ is a Hilbert--Schmidt operator on $\Vb$. Assuming that the regularised determinant $\mathrm{det}_2(\id-Q)\neq0$,
this pair of operators defines a Fredholm Grassmannian flow in a given coordinate patch as follows. We can think of
this pair of operators as spanning a subspace of $\mathbb H\coloneqq\mathbb V\times\mathbb V$ that is isomorphic to $\Vb$.
The transformation $(\id-Q)^{-1}$ of this subspace generates,
\begin{equation*}
  \begin{pmatrix} \id\\ G\end{pmatrix},
\end{equation*}
where $G=2\mathrm{i}P(\id-Q)^{-1}$. We can think of the Hilbert--Schmidt operator $G$ as parametrising all such subspaces of $\mathbb H$,
that can be projected onto the canonical subspace represented by the pair of operators $(\id,O)$.
This is one coordinate patch of the Fredholm Grassmannian of such subspaces of $\mathbb H$.
Note that if we set $Q=-P^2$, then $G$ represents the solution to the Marchenko equation.
Recall, in our application we assume $P$ satisfies the dispersion equation $\pa_tP=-(\mathrm{i}\mathcal I)^{n-1}\pa^nP$. 
Further we suppose that $P$ is a Hankel operator. This property is a natural far-field symmetry for the dispersive field in the sense
that it is a natural symmetry arising as the result of the scattering, of an incident wave from one far field, into the opposite far field.
See for example the construction of the Marchenko equation in Drazin and Johnson~\cite{DJ} or Appendix~B in Doikou \textit{et al.} \cite{DMSW:AGFintegrable}.
That the Marchenko equation solution $G$ parameterises a class of subspaces of $\mathbb H$ characterised by solutions of a dispersive field $P$, 
generates the following perspective. We can think of the Fredholm Grassmannian, in the coordinate patch represented by the particular pair $(\id, G)$,
as parametrising the time-evolving \emph{envelope} of dispersive field solutions, i.e.\/ the time-evolving subspace represented by the pair $(\id,G)$.
In principle we could consider $\lb G\rb=\lb G\rb(y,z;x,t)$ or in particular $\lb G\rb(0,0;x,t)$ as an observable.  

The Marchenko equation, and its role in inverse scattering and linearisation, has been fundamental in classical integrable systems
from the very early stages. See for example Dyson~\cite{Dyson}, Miura~\cite{Miura}, Zakharov and Shabat~\cite{ZS,ZS2},
Ablowitz \textit{et al.} \cite{ARSII}, Fokas and Ablowitz~\cite{FA}, Mumford~\cite{Mumford},
P\"oppe~\cite{P-SG,P-KdV,P-KP}, P\"oppe and Sattinger~\cite{PS-KP}, Bauhardt and P\"oppe~\cite{BP-ZS} and Nijhoff \textit{et al.} \cite{NQVDLCII,NQVDLCI}.
There has also been a resurgence of interest in such linearisation approaches, see for example 
Fokas and Pelloni~\cite{FP}, McKean~\cite{McKean}, Fu~\cite{Fu} and Fu and Nijhoff~\cite{FNI}. 
It was Sato~\cite{SatoI,SatoII} and Segal and Wilson~\cite{SW} who pioneered the connection between Fredholm Grassmannians and
integrable systems. Recently there has also been a resurgence in this direction as well, see for example, Mulase~\cite{Mulase},
Dupr\'e \textit{et al.} \cite{DGP2006,DGP2007,DGP2013}, Kasman~\cite{Kasman1995,Kasman1998}, Hamanaka and Toda~\cite{HT}, Cafasso~\cite{Cafasso},
Cafasso and Wu~\cite{CW} and Arthamonov \textit{et al.} \cite{AHH} (also see Beck \textit{et al.} \cite{BM,BDMS1,BDMS2}
and Doikoi \textit{et al.} \cite{DMSW:AGFintegrable,DMSW:AGFcoagulation}).
Related to this is the well-studied connection between the Korteweg--de Vries hierarchy, the intersection theory of Deligne--Mumford moduli space
and the a string equation in two-dimensional gravity; see for example Witten~\cite{Witten1990,Witten1991} and Cafasso and Wu~\cite{CW}.
Some of the earliest work on non-commutative integrable systems includes
Fordy and Kulish~\cite{FK}, Nijhoff \textit{et al.} \cite{NQVDLCI}, Ablowitz \textit{et al.} \cite{APT}, 
Ercolani and McKean~\cite{EM} and Aden and Carl~\cite{AC1995}. Again there has been more recent interest in such systems and their solutions, such as 
Treves~\cite{TI,TII}, Hamanaka and Toda~\cite{HT}, Degasperis and Lombardo~\cite{DL2}, Dimakis and M\"uller--Hoissen~\cite{DM-H2010},
Carillo and Schiebold~\cite{CSI,CSIa,CSII} whose results are particularly relevant to those herein, Sooman~\cite{Sooman},
Pelinovsky and Stepanyants~\cite{Pelinovsky}, Buryak and Rossi~\cite{BuryakRossi}, Doikou \textit{et al.} \cite{DMS},
Stylianidis~\cite{Stylianidis}, Adamopoulou and Papamikos~\cite{AP}, Malham~\cite{Malham:quinticNLS}, G\"urses and Pekcan~\cite{GP2022} and Ma~\cite{Ma}.
The role of Hankel operators in integrable systems first explored by P\"oppe, has recently re-emerged as an active and fruitful research direction.
In particular, relevant to our results herein are Blower and Newsham~\cite{BN}, Blower and Doust~\cite{BD}, Grudsky and Rybkin~\cite{GRI,GRII,GRIII},
Grellier and Gerard~\cite{Gerard} and Gerard and Pushnitski~\cite{GP}.
The combinatorial algebraic approach we consider herein was introduced in Malham~\cite{Malham:KdVhierarchy} for the simpler
non-commutative potential Korteweg--de Vries equation; also see Doikoi \textit{et al.} \cite{DMSW:AGFintegrable}.
Dimakis and M\"uller--Hoissen~\cite{DM-H2008,DM-H2009} consider integrable systems in the context of bidifferential graded algebras,
while in Dimakis and M\"uller--Hoissen~\cite{DM-H2005}, they consider connections to shuffle and Rota--Baxter algebras. 
See Reutenauer~\cite{Reutenauer}, Malham and Wiese~\cite{MW} and Ebrahimi--Fard \textit{et al.} \cite{EFMKLMW}
for more details on shuffle algebras and references for Rota--Baxter algebras.

To summarise, our achievements herein are as follows. In terms of algebras, we: 
\begin{enumerate}
\item[(i)] Introduce and develop new abstract non-commutative algebras. These are the algebra of
  non-negative integer monomial forms $\mathbb C\la\mathbb Z_{\bs0}\ra$ described above, equipped with a quasi-Leibniz type product
  based on the P\"oppe product, and its skew-form subalgebra $\mathbb C[\mathbb Z_{\bs0}]$. They are instrumental
  to the results (ii)--(iv) just below.  
\end{enumerate}
For the non-commutative nonlinear Schr\"odinger and modified Korteweg--de Vries hierarchy, we:
\begin{enumerate}
\item[(ii)] Provide a constructive proof that at each non-negative order, there exists a unique hierarchy member
  in the class of odd-polynomial partial differential fields. The proof simultaneously establishes that the
  solution flow of each member is a Fredholm Grassmannian, and therefore linearisable in the sense outlined in detail above;
\item[(iii)] Give a simple proof of the non-commutative Lax hierarchy in this context. In addition, 
  we immediately establish that at each non-negative order, the unique hierarchy member in (ii), and the corresponding Lax hierarchy member, coincide; 
\item[(iv)] Establish that the first negative order non-commutative Lax hierarchy member is the cubic form
  of the non-commutative sine-Gordon equation, and further, demonstrate how to generate the rest of the negative order non-commutative hierarchy.
\end{enumerate}

Our paper is organised as follows. In Section~\ref{sec:Hankelandlinear} we introduce the P\"oppe product for Hankel operators and motivate the solution form we propose
for the non-commutative nonlinear Schr\"odinger and modified Korteweg--de Vries hierarchy, based on the associated Marchenko equation. 
We introduce the P\"oppe kernel monomial algebra in Section~\ref{sec:Poppealg} with the P\"oppe product,
and in particular its isomorphic abstract form as well as the skew-P\"oppe subalgebra we use for the
proofs of our main results. We present a sequence of simple examples in Section~\ref{sec:ncNLS} 
illustrating the use of the abstract P\"oppe algebra to generate the order zero through to order four members of the non-commutative hierarchy.
In Section~\ref{sec:Laxhierarchy} we establish the non-commutative Lax hierarchy using the P\"oppe algebra. 
Herein, we also generate the cubic form of the non-commutative sine-Gordon equation as the first negative order case,
and indicate how to generate the rest of the negative order cases. We begin Section~\ref{sec:hierarchycoding} with
the illuminating example of the quintic non-commutative modified Korteweg--de Vries equation, before stating, and then proving, our main results. 
Finally, in Section~\ref{sec:conclu} we present some further conclusions and applications.

\section{Hankel operators, the P\"oppe product and the Marchenko equation}\label{sec:Hankelandlinear}
In this section we introduce the concepts and results that underlie our formulation.
Herein, we: introduce the necessary Hilbert--Schmidt and Hankel operators we use and define the P\"oppe product; 
motivate the solution ansatz we use throughout; formulate the base linear dispersion equation;
elucidate the well-posedness results for the Marchenko equation we require and establish the connection
between the P\"oppe product and finite rank operators.

\subsection{Hankel operators and the P\"oppe product}\label{subsec:Poppeprod}
To begin, let us fix some notation. Let $\Vb$ be the Hilbert space of square-integrable, complex matrix-valued functions on $(-\infty,0]$,
i.e.\/ $\Vb\coloneqq L^2((-\infty,0];\CC^{m})$ for some $m\in\mathbb N$.
Further, we denote by $\mathfrak J_2(\Vb)$ the space of Hilbert--Schmidt operators on $\Vb$,
i.e.\/ bounded operators whose sum of the squares of their singular values is finite.  
For any given operator $F=F(x,t)\in\mathfrak J_2(\Vb)$
there exists a unique square-integrable kernel $f=f(y,z;x,t)$ with
$f\in L^2((-\infty,0]^{\times 2};\CC^{m\times m})$
such that for any $\phi\in\Vb$ we have
\begin{equation*}
(F\phi)(y;x,t)=\int_{-\infty}^0f(y,z;x,t)\phi(z)\,\rd z.
\end{equation*}
Conversely, any such function $f\in L^2((-\infty,0]^{\times 2};\CC^{m\times m})$
defines an operator $F=F(x,t)$ in $\mathfrak J_2(\Vb)$ with (for each $x,t$):
\begin{equation*}
\|F\|_{\mathfrak J_2(\Vb)}=\|f\|_{L^2((-\infty,0]^{\times 2};\CC^{m\times m})}. 
\end{equation*}
See for example Simon~\cite[p.~23]{Simon:Traces}.
\begin{definition}[Kernel bracket]\label{def:Kernelbracket}
For any Hilbert--Schmidt operator $F=F(x,t)$, which depends on the parameters $x\in\R$ and $t\geqslant0$,
we use the \emph{kernel bracket} notation $\lb F\rb$ to refer to the kernel $f=f(y,z;x,t)$ of $F$:
\begin{equation*}
\lb F\rb(y,z;x,t)\coloneqq f(y,z;x,t). 
\end{equation*}
In general, since $f$ is square-integrable, it only exists almost everywhere on $(-\infty,0]^{\times 2}$.
However below, the operators we consider have continuous kernels and so $f$ makes sense pointwise.
In such cases, we can in particular set $y=z=0$, for which we use the notation $\lb F\rb_{0,0}(x,t)\coloneqq f(0,0;x,t)$.
\end{definition}
Recall that the \emph{trace} of any trace-class operator $F$ on $(-\infty,0]$ is given by,
\begin{equation*}
\mathrm{tr}\, F\coloneqq\int_{-\infty}^0 f(z,z)\,\rd z.
\end{equation*}

By a Hankel operator, which may depend on a parameter $x$, we mean the following.
\begin{definition}[Hankel operator with parameters]\label{def:Hankel}
We say a Hilbert--Schmidt operator $H\in\mathfrak J_2(\Vb)$ 
with corresponding square-integrable kernel $h$ is \emph{Hankel} or \emph{additive} with parameter $x\in\R$
if its action, for any square-integrable function $\phi\in\Vb$, 
is given by (here $y\in(-\infty,0]$),
\begin{equation*}
(H\phi)(y;x)\coloneqq\int_{-\infty}^0h(y+z+x)\phi(z)\,\rd z.
\end{equation*}
\end{definition}
P\"oppe~\cite{P-SG,P-KdV} recognised the fundamental role played by such Hankel operators
in classical integrable systems. The kernel of the derivative with respect to the additive parameter $x$
of the operator product of an arbitrary pair of Hankel operators can be expressed as the matrix product
of their respective kernels as follows. See P\"oppe~\cite{P-SG,P-KdV}, as well as Doikou \textit{et al.} \cite{DMS}
and Malham~\cite{Malham:KdVhierarchy}. We include a proof for completeness.
\begin{lemma}[P\"oppe product]\label{lemma:origPoppeproduct}
  Assume $H$ and $H^\prime$ are Hankel Hilbert--Schmidt operators with parameter $x$ and $F$ and $F^\prime$
  are Hilbert--Schmidt operators. Further assume the kernels of $F$ and $F^\prime$ are continuous, whilst
  the kernels of $H$ and $H^\prime$ are continuously differentiable. Then the following P\"oppe product
  rule holds,
  \begin{equation*}
    \bigl[\hspace{-0.1cm}\bigl[F\pa_x(HH^\prime)F^\prime\bigr]\hspace{-0.1cm}\bigr](y,z;x)
    =\lb FH\rb(y,0;x)\lb H^\prime F^\prime\rb(0,z;x).
  \end{equation*}
\end{lemma}
\begin{proof}
We use the fundamental theorem of calculus and Hankel properties of $H$ and $H'$.
Let $f$, $h$, $h'$ and $f'$ denote the integral kernels of $F$, $H$, $H'$ and $F'$ respectively.
By direct computation $\lb F\pa_x(HH^\prime)F^\prime\rb(y,z;x)$ equals
\begin{align*}
&\int_{\R_-^3}
f(y,\xi_1;x)\pa_x\bigl(h(\xi_1+\xi_2+x)h^\prime(\xi_2+\xi_3+x)\bigr)
f^\prime(\xi_3,z;x)\,\rd \xi_3\,\rd \xi_2\,\rd \xi_1\\
&=\int_{\R_-^3}
f(y,\xi_1;x)\pa_{\xi_2}\bigl(h(\xi_1+\xi_2+x)h^\prime(\xi_2+\xi_3+x)\bigr)
f^\prime(\xi_3,z;x)\,\rd \xi_3\,\rd \xi_2\,\rd \xi_1\\
&=\int_{\R_-^2}
f(y,\xi_1;x)h(\xi_1+x)h^\prime(\xi_3+x)f^\prime(\xi_3,z;x)\,\rd \xi_3\,\rd \xi_1\\
&=\int_{\R_-}f(y,\xi_1;x)h(\xi_1+x)\,
\rd \xi_1\cdot\int_{\R_-}h^\prime(\xi_3+x)f^\prime(\xi_3,z;x)\,\rd \xi_3\\
&=\bigl(\lb FH\rb(y,0;x)\bigr)\bigl(\lb H'F'\rb(0,z;x)\bigr),
\end{align*}
which corresponds to the result stated. 
\qed
\end{proof}
\begin{remark}\label{rmk:implicitnotation}
We implicitly interpret kernel products written in the form $\lb\,\cdot\,\rb\lb\,\cdot\,\rb\cdots\lb\,\cdot\,\rb\lb\,\cdot\,\rb$
as $\lb\,\cdot\,\rb(y,0;x)\lb\,\cdot\,\rb(0,0;x)\cdots\lb\,\cdot\,\rb(0,0;x)\lb\,\cdot\,\rb(0,z;x)$.
\end{remark}

\subsection{Solution ansatz motivation} 
Let us formally motivate the solution form we use for the non-commmutative nonlinear Schr\"odinger and modified Korteweg--de Vries hierarchy we study herein.
This comes from the sine-Gordon equation.  With $x\in\R$ and $t\geqslant0$, we assume the sine--Gordon equation has the form,
$\pa_t\pa u=\sin u$, where $u=u(x,t)$ and $\pa\coloneqq\pa_x$. In the scalar case, when $u\in\R$, it is well-known
that there exists a solution of the form,
\begin{equation*}
u=2\mathrm{i}\,\mathrm{tr}\,\log\biggl(\frac{\id-\mathrm{i}P}{\id+\mathrm{i}P}\biggr),
\end{equation*}
or the equivalent form $u=4\,\mathrm{arctan}\, P$.
Here $P=P(x,t)$ is a Hankel Hilbert--Schmidt operator with an integral kernel $p=p(x,t)$ which satisfies the
linearised form of the sine--Gordon equation (see for example P\"oppe~\cite[Cor.~3.2]{P-SG}), $\pa_t\pa p=p$.
For scalar valued kernels, we have the following.
Suppose that $H=H(x)$ is a Hankel operator dependent on the parameter $x\in\R$. Then for any $n\in\mathbb N$, we have,
\begin{equation*}
\pa\,\mathrm{tr}\, H^n=\frac{n}{2}\lb H^n\rb_{0,0}. 
\end{equation*}
Further, suppose that $\Theta=\Theta(H)$ is a power series function of the Hankel operator $H$, with
scalar-valued coefficients $c_n$, of the form $\Theta(H)=\sum_{n\geqslant1}c_n\,H^n$. Then we have,
\begin{equation*}
\pa\,\mathrm{tr}\,\Theta(H)=\tfrac12\lb H\,D\Theta(H)\rb_{0,0},
\end{equation*}
where $D\Theta=D\Theta(H)$ is the series $D\Theta(H)=\sum_{n\geqslant1}nc_n\,H^{n-1}$.
\begin{remark}
  This is equivalent to the result embodied in equation (3.26) in P\"oppe~\cite{P-SG}.
  We give a proof in Proposition~\ref{prop:bracketidentities} in Section~\ref{sec:tracefiniterank} below.
  Also see Blower and Doust~\cite{BD}.
\end{remark}
\begin{example}[Logarithm of the Cayley transform]\label{ex:logCayley}
The solution ansatz for the scalar sine-Gordon equation above involves the logarithm of
Cayley transform, i.e.\/ the form,
\begin{equation*}
\Theta(P)=\log\biggl(\frac{\id-\mathrm{i}P}{\id+\mathrm{i}P}\biggr).
\end{equation*}
By direct computation we observe that,
\begin{equation*}
D\Theta(P)=-\frac{\mathrm{i}\cdot\id}{\id-\mathrm{i}P}-\frac{\mathrm{i}\cdot\id}{\id+\mathrm{i}P}
\quad\Rightarrow\quad P\,D\Theta(P)=-\frac{2\mathrm{i}P}{(\id-\mathrm{i}P)(\id+\mathrm{i}P)}.
\end{equation*}
Then, using the trace and kernel bracket result above, we have,
$\pa\, \mathrm{tr}\, \Theta(P)=-\lb(\id-\mathrm{i}P)^{-1}(\mathrm{i}P)(\id+\mathrm{i}P)^{-1}\rb_{0,0}$,
though the order of the factors shown on the right us not important.
\end{example}
Recall the solution form to the scalar sine--Gordon equation given above, $u=2\mathrm{i}\,\mathrm{tr}\,\Theta(P)$, where $\Theta=\Theta(P)$ is
the logarithm of the Cayley transform given in Example~\ref{ex:logCayley}. Let $\pa^{-1}\coloneqq\pa^{-1}_x$ denote the primative operator, 
$\bigl(\pa^{-1}\phi\bigr)(x)\coloneqq\int_{-\infty}^x\phi(\xi)\,\rd\xi$.
Then, given the final result in Example~\ref{ex:logCayley}, we can express the solution to the scalar sine--Gordon equation in the form,
\begin{equation*}
u=-2\mathrm{i}\,\pa^{-1}\bigl[\hspace{-0.1cm}\bigl[(\id-\mathrm{i}P)^{-1}(\mathrm{i}P)(\id+\mathrm{i}P)^{-1}\bigr]\hspace{-0.1cm}\bigr]_{0,0}.
\end{equation*}
This form of the solution for the scalar sine--Gordon equation motivates
the solution form we seek for the non-commutative nonlinear Schr\"odinger and modified Korteweg--de Vries hierarchy, which we utilise
in the following sections. The sine--Gordon equation is just a special case, the order `$-1$' case, in that hierarchy. 
\begin{example}[Non-commutative sine-Gordon cubic-form equation]\label{ex:SGform}
  If $u$ satisfies the scalar sine--Gordon equation $\pa_t\pa u=\sin u$, and $u=-2\mathrm{i}\pa^{-1}g$,
  then $g=g(x,t)$ satisfies the following sine--Gordon cubic-form equation,
  \begin{equation*}
    \pa_t\pa g=g+g\,\pa^{-1}(\pa_t g^2)+\pa^{-1}(\pa_t g^2)\,g.
  \end{equation*}
  To see this, define the operator $\Gamma$ by, $\bigl(\Gamma\phi\bigr)(x)\coloneqq\int_{-\infty}^x\gamma(\xi)\phi(\xi)\,\rd\xi$,
  where $\gamma=-2\mathrm{i}g$.
  Using that for any $n\in\mathbb N$ we have, $(\Gamma\circ1)^n\equiv n!\,(\Gamma^n\circ1)$, then we observe that in fact,
  $\sin u=(\id+\Gamma^2)^{-1}\circ\Gamma\circ1$.
  In other words $\gamma$ satisfies the integral equation $(\id+\Gamma^2)\circ\pa_t\gamma=\Gamma\circ1$ or equivalently 
  satisfies $\pa_t\gamma+\pa^{-1}\bigl(\gamma\pa^{-1}(\gamma\pa_t\gamma)\bigr)=\pa^{-1}\gamma$. Noting that $g=\gamma/(-2\mathrm{i})$ 
  and symmetrically splitting the nonlinear term, generates the sine--Gordon cubic-form equation above.
  The cubic-form equation above is often interpreted to be the non-commutative sine--Gordon equation in,
  for example, Schiebold~\cite[Prop.~6.2]{SchieboldncSG}.
\end{example}

\subsection{The linear dispersion equation}\label{subsec:lineardispersion}
Consider the following coupled linear system of equations for the Hilbert--Schmidt operators $P_\alpha$, $P_\beta$, $G_\alpha$ and $G_\beta$,
\begin{equation*}
\begin{aligned}
\pa_tP_\alpha&=\mu_n\pa^n P_\alpha,\\
\mathrm{i}P_\alpha&=G_\alpha(\id+P_\beta P_\alpha),
\end{aligned}
~\qquad\text{and}~\qquad
\begin{aligned}
\pa_tP_\beta&=(-1)^{n-1}\mu_n\pa^n P_\beta,\\
\mathrm{i}P_\beta&=G_\beta(\id+P_\alpha P_\beta).
\end{aligned}
\end{equation*}
for some order $n\in\mathbb Z$, where the parameter $\mu_n\in\mathbb C$.
In order for the partial differential equations for $P_\alpha$ and $P_\beta$ shown
to be dispersive, we necessarily require that $\mu_n$ is pure imaginary when $n$ is even
and real when $n$ is odd. We suppose that the matrix-valued kernel of $P_\beta$ has the same shape
as the transpose of the matrix-valued kernel of $P_\alpha$. 
The matrix-valued kernels of $G_\alpha$ and $G_\beta$ naturally match those of $P_\alpha$ and $P_\beta$, respectively.
If we set,
\begin{equation*}
  P\coloneqq\begin{pmatrix} O & P_\beta \\ P_\alpha & O \end{pmatrix},\qquad
  G\coloneqq\begin{pmatrix} O & G_\beta \\ G_\alpha & O \end{pmatrix},
  \qquad\text{and}\qquad 
  \mathcal I\coloneqq\begin{pmatrix} -\id & O \\ O & \id \end{pmatrix},
\end{equation*}
then the system of linear equations above can be expressed in the form,
\begin{align*}
\pa_t P&=-\mu_n(\mathrm{i}\mathcal I)^{n-1}\pa^n P,\\
\mathrm{i}P&=G(\id+P^2),
\end{align*}
where now the parameter $\mu_n\in\mathbb R$. This form is given, eg., in Schiebold~\cite[p.679--80]{SchieboldncSG}.
Now consider the following second order cubic nonlinear equation for the kernel $\lb G\rb$,
\begin{equation*}
\pa_t\lb G\rb(y,z;x,t)=-\mu_2\mathrm{i}\mathcal I\bigl(\pa_x^2\lb G\rb(y,z;x,t)-2\,\lb G\rb(y,0;x,t)\lb G\rb(0,0;x,t)\lb G\rb(0,z;x,t)\bigr).
\end{equation*}
Written in terms of the kernels $\lb G_\alpha\rb$ and $\lb G_\beta\rb$ with $y=z=0$, we observe,
\begin{align*}
  \mathrm{i}\pa_t\lb G_\alpha\rb&=\mu_2\pa_x^2\lb G_\alpha\rb-2\mu_2\,\lb G_\alpha\rb\,\lb G_\beta\rb\,\lb G_\alpha\rb,\\
  \mathrm{i}\pa_t\lb G_\beta\rb&=-\mu_2\pa_x^2\lb G_\beta\rb+2\mu_2\,\lb G_\beta\rb\,\lb G_\alpha\rb\,\lb G_\beta\rb.
\end{align*}
There are several different consistent choices we can make for $P_\alpha$ and $P_\beta$, as follows.
For example, suppose we set $P_\beta=P_\alpha^\dag$, the adjoint operator to $P_\alpha$ with respect to the $L^2$ inner product.
Then if $G=\mathrm{i}PU$ with $U\coloneqq(\id+P^2)^{-1}$, as defined above, at the block level it transpires $G_\beta=G_\alpha^\dag$. 
In this case the kernel $\lb G_\beta\rb(0,0;x,t)$ is the complex conjugate transpose of the kernel $\lb G_\alpha\rb(0,0;x,t)$.
And thus, assuming the kernel $\lb G\rb$ generated from $G=PU$ satisfies the equation above,
the equation for the block $\lb G_\alpha\rb(0,0;x,t)$ collapses to the non-commutative nonlinear Schr\"odinger equation.
Further note, for the choice $P_\beta=P_\alpha^\dag$, the operator $P$ is Hermitian with respect to the  $L^2$ inner product,
i.e.\/ $P^\dag=P$. 
\begin{remark}[Reverse and shifted space-time nonlocal equations]
The system of linear equations for $P$ above allows us to incorporate, and thus deduce corresponding results,
for the nonlocal versions of the non-commutative nonlinear Schr\"odinger hierarchy. These include the
reverse time, reverse space-time and space-time shifted nonlocal versions of these equations outlined in 
Ablowitz and Musslimani~\cite{AMusslimani,AMshift}, Fokas~\cite{Fokas}, Grahovski, Mohammed and Susanto~\cite{GMS}
and G\"urses and Pekcan~\cite{GP2018,GP2019b,GP2020,GP2022}.
This fact is outlined in detail in Example~4 and Remark~17 in Doikou \textit{et al.} \cite{DMSW:AGFintegrable}.
\end{remark}

\subsection{The Marchenko equation}
Consider an operator $P\in\mathfrak J_N(\mathbb V)$, where $N=1$ or $N=2$. For the moment $P$ is not necessarily a Hankel operator,
and $\mathbb V$ is an arbitrary seperable Hilbert space. The space $\mathfrak J_1(\Vb)$ denotes the set of trace-class (nuclear) operators.
Crucial to the P\"oppe algebra we introduce in Section~\ref{sec:Poppealg} are both, the Marchenko equation,
\begin{equation*}
P=G\,(\id-Q),
\end{equation*}
for the operator $G$, and the P\"oppe product in Lemma~\ref{lemma:origPoppeproduct}. In our application, we set $Q\coloneqq-P^2$.
The following abstract result is proved in Doikou \textit{et al.\/} \cite[Lemma~1]{DMSW:AGFintegrable}.
\begin{lemma}[Existence and Uniqueness; Doikou \textit{et al.\/} \cite{DMSW:AGFintegrable}]\label{lemma:EU}
  Assume $Q_0\in\mathfrak J_2$ and for some $T>0$ we know that $Q\in C^\infty\bigl([0,T];\Jf_2\bigr)$ with $Q(0)=Q_0$ and
  $P\in C^\infty\bigl([0,T];\Jf_N\bigr)$, where $N$ is $1$ or $2$. Further assume, $\mathrm{det}_2(\id-Q_0)\neq0$.
  Then there exists a $T^\prime>0$ with $T^\prime\leqslant T$ such that for $t\in[0,T^\prime]$ we have $\mathrm{det}_2(\id-Q(t))\neq0$ and 
  there exists a unique solution $G\in C^\infty\bigl([0,T^\prime];\Jf_N\bigr)$ to the linear equation $P=G(\id-Q)$.
\end{lemma}
Now suppose $\Vb\coloneqq L^2((-\infty,0];\CC^{m})$ for some $m\in\mathbb N$.
For any function $w\geqslant0$, we denote the weighted $L^2$-norm of any complex matrix-valued function $f$ on $(-\infty,0]$ by,
\begin{equation*}
   \|f\|_{L^2_w}^2\coloneqq\int_{-\infty}^0\mathrm{tr}\,\bigl(f^\dag(x)f(x)\bigr)\,w(x)\,\rd x.
\end{equation*}
Doikou \textit{et al.\/} \cite[Lemma~3]{DMSW:AGFintegrable} also establish, if $p(\cdot;t)\in L^2_w((-\infty,0])$
with $w\colon y\mapsto(1-y)^2$, the Hankel operator $P=P(t)$ generated by $p$ is such that $P(t)\in\mathfrak J_2(\Vb)$.
Hence we \emph{assume} the solutions $p=p(y;t)$ to the linear dispersive system $\pa_t p=-\mu_n(\mathrm{i}\mathcal I)^{n-1}\pa_y^n p$
lie in $L^2_w((-\infty,0])$. We then take $P=P(x,t)$ to be the Hankel operator with kernel $p=p(y+z+x;t)$,
where $y,x\in(-\infty,0]^{\times2}$, with parameter $x\in\R$. Statements for $p=p(\cdot;t)$ on $(-\infty,0]$
translate, for each $x\in\R$, to statements for $p=p(\cdot+x;t)$ on $(-\infty,x]$. This is important, as we wish
to include natural solutions $p=p(y;t)$ to the linear dispersion equation that are unbounded as $y\to\infty$.
Examples of such solutions are exponential-form solutions that generate soliton solutions to the corresponding non-commutative
integrable nonlinear partial differential equation. Explicitly, the Marchenko equation we consider herein takes the form,
\begin{equation*}
p(y+z+x;t)=g(y,z;x,t)-\int_{-\infty}^0 g(y,\xi;x,t)q(\xi,z;x,t)\,\rd\xi,
\end{equation*}
where $q$ is the kernel of $Q\coloneqq-P^2$. With this in hand, we have the following result, adapted from 
Doikou \textit{et al.\/} \cite[Lemma~6]{DMSW:AGFintegrable}.
\begin{lemma}[Existence and Uniqueness: Marchenko equation]\label{lemma:EUlinearintegralsystem}
  Assume the smooth initial data $p_0=p_0(\cdot)$ for $p=p(\cdot;t)$ is such that $\mathrm{det}_2(\id-Q_0)\neq0$,
  where $Q_0\coloneqq-P_0^2$ and $P_0$ is the Hankel operator generated by $p_0$. Further assume there exists a $T>0$
  such that there is a solution,
  \begin{equation*}
    p\in C^{\infty}\bigl([0,T];L^2_w((-\infty,0];\CC^{m\times m})\cap C^\infty((-\infty,0];\CC^{m\times m})\bigr),
  \end{equation*}
  to the linear dispersion equation $\pa_t p=-\mu_n(\mathrm{i}\mathcal I)^{n-1}\pa_y^n p$, where $w\colon y\mapsto(1-y)^2$.
  Then there exists a $T^\prime>0$ with $T^\prime\leqslant T$ such that for $t\in[0,T^\prime]$ we know:
  (i) The Hankel operator $P=P(x,t)$ with parameter $x\in\R$ generated by $p$ is Hilbert--Schmidt valued on $\Vb$;
  (ii) The determinant $\mathrm{det}_2(\id-Q(x,t))\neq0$ where $Q(x,t)\coloneqq-P^2(x,t)$, and hence 
  (iii) There is a unique Hilbert--Schmidt valued solution $G=G(x,t)$ with $G\in C^\infty([0,T^\prime];\mathfrak J_2(\Vb))$
  to the linear Fredholm equation $P=G(\id-Q)$.
\end{lemma}

\subsection{Trace formulae and finite-rank operators}\label{sec:tracefiniterank}
We now establish that at the core of the P\"oppe product is in fact a finite rank operator. 
For this section only, for convenience, we assume the domain of support of the functions under consideration is $[0,\infty)$
as opposed to $(-\infty,0]$. A reflection transformation translates between the two. 
For a bounded integral operator $K\colon L^2(0,\infty)\to L^2(0,\infty)$ with a continuous kernel $k=k(y,z)$,
we write $\lb K\rb_{0,0}=k(0,0)$ for the kernel bracket.
\begin{definition}[Shift operator]
  We define the shift operator $S_\eta\colon L^2(0,\infty)\to L^2(0,\infty)$ by $S_\eta f(x)=f(x-\eta)\,\mathrm{ind}_{(0,\infty)}(x-\eta)$,
  where $\mathrm{ind}_{(0,\infty)}$ is the indicator function on $(0,\infty)$.
\end{definition}
It is well-known that $S_\eta$ is a linear isometry and $(S_\eta)_{\eta>0}$ is a strongly continuous contraction semigroup.
Further, the adjoint $S_\eta^\dag$ is a linear contraction, and $(S_\eta^\dag)_{\eta>0}$ is a strongly continuous contraction semigroup on $L^2$.
\begin{proposition}
  Set $\sigma_\eta(K)\coloneqq S_\eta^\dag KS_\eta$ for $K\in\mathfrak J(L^2(0,\infty))$. Then we have:
  \begin{enumerate}
  \item[(i)] $\sigma_\eta(K)\in\mathfrak J(L^2(0,\infty))$ for all $K\in\mathfrak J(L^2(0,\infty))$ with
    $\|\sigma_\eta(K)\|\leqslant\|K\|$, $K\mapsto\sigma_\eta(K)$ is linear and $\sigma_{\eta+\xi}=\sigma_\eta(\sigma_\xi(K))$;
  \item[(ii)] $K=K^\dag$ implies $(\sigma_\eta(K))^\dag=\sigma_\eta(K)$ and $K\geqslant0$ implies $\sigma_\eta(K)\geqslant0$;
  \item[(iii)] $(S_\eta)_{\eta>0}$ gives a strongly continuous contraction semigroup on the von Neumann--Schatten ideal $\mathfrak J_p$
  for $1\leqslant p<\infty$ and on the space of compact operators. Also $\sigma_\eta(K)\to K$ as $\eta\to0$ for such $K$;
\item[(iv)] Let $\delta$ be the generator of the semigroup in (iii), so $\sigma_\eta=\exp(t\,\delta)$.
  Then for continuously differentiable kernels $k=k(y,z)$ we have, $\delta k(y,z)=(\pa_y+\pa_z)k(y,z)$;
  \item[(v)] Suppose that $K$ has a continuous kernel $k=k(y,z)$, and that $K$ is self-adjoint, non-negative and trace class.
    Then, $\mathrm{det}(\id+\sigma_\eta(K))=\mathrm{det}(\id+K\mathrm{Pr}_{(\eta,\infty)})$ and,
  \begin{equation*}
       \lb\sigma_\eta(K)\rb_{0,0}=k(\eta,\eta)=-\frac{\rd}{\rd\eta}\mathrm{tr}\,\sigma_\eta(K);
  \end{equation*}    
  \item[(vi)] P\"oppe's bracket operation satisfies, $\mathrm{tr}\,\delta K=-\lb K\rb_{0,0}$.
\end{enumerate}
\end{proposition}
\begin{proof}
  (i) Follows since $(S_\eta)$ is a contraction semigroup, while (ii) is straightforward. 
  (iii) The Schatten class gives an operator ideal, so we have $\| S_\eta^\dag KS_\eta\|_{\mathfrak J_p}\leqslant\| K\|_{\mathfrak J_p}$
  since $\|S_\eta\|_{\mathfrak J}=1$. In view of this, we only need to check continuity in the relevant norm. 
  For the Hilbert--Schmidt operators $\mathfrak J_2$, we let $\sigma_\eta(K)(y,z)$ be the kernel of $\sigma_\eta(K)$ as an integral operator.
  Then we have $\sigma_\eta(K)(y,z)=K(y+\eta,z+\eta)$, so by the Hilbert--Schmidt theorem,
  \begin{equation*}
    \|\sigma_\eta(K)-K\|^2_{\mathfrak J_2}=\int_0^\infty\int_0^\infty\| k(y+\eta,z+\eta)-k(y,z)\|^2\, \rd y\rd z,
  \end{equation*}
  which converges to $0$ as $\eta\to 0^+$. For the trace class operators $\mathfrak J_1$,
  we observe that the space of trace class operators may be identified with the projective tensor product
  $\mathfrak J_1=L^2\hat\otimes L^2$, so we have a nuclear expansion,
  \begin{equation*}
  k(y,z)=\sum_{j=1}^\infty f_j(y)g_j(z),
  \end{equation*}
  where $\sum_{j=1}^\infty \|f_j\|_{L^2}\|g_j\|_{L^2}=\|K\|_{\mathfrak J_1}$. Then we have
  \begin{equation*}
  \sigma_\eta(K)(y,z)-k(y,z)=\sum_{j=1}^\infty\bigl(f_j(y+\eta)-f_j(y)\bigr)g_j(z+\eta)+\sum_{j=1}^\infty f_j(y)\bigl(g_j(z+\eta)-g_j(y)\bigr),
  \end{equation*}
  and so,
  \begin{equation*}
  \|\sigma_\eta(K)-K\|_{\mathfrak J_1}\leqslant
  \sum_{j=1}^\infty \|S_\eta^\dag(f_j)-f_j\|_{L^2}\|g_j\|_{L^2}+\sum_{j=1}^\infty\|f_j\|_{L^2}\|S_\eta(g_j)-g_j\|_{L^2},
  \end{equation*}
  where the right-hand side converges to $0$ as $\eta\to 0^+$ by dominated convergence. 

  For $1<p<\infty $, we observe that the finite-rank operators give a dense linear subspace of $\mathfrak J_p$,
  so we can argue as with the trace class operators. Likewise, the finite-rank operators give a dense linear subspace of the space of compact operators.
  Hence $(\sigma_\eta)_{\eta>0}$ gives as strongly continuous contraction semigroup on these spaces. 
  Now $S_\eta^\dag f\to 0$ as $\eta\to\infty$ for all $f\in L^2(0,\infty)$,
  so for all finite rank operators $F$, we have $\sigma_\eta(F)\to 0$ as $\eta\to\infty$.
  Then for $K\in\mathfrak J_p$ and $\varepsilon>0$ there exists a finite rank $F$ such that $\|K-F\|_{\mathfrak J_p}<\varepsilon$,
  so $\|\sigma_\eta(K)\|_{\mathfrak J_p}\leqslant\|K-F\|_{\mathfrak J_p}+\|\sigma_\eta(F)\|_{\mathfrak J_p}$ is less than $2\varepsilon$
  for all sufficiently large $\eta$. (Note, we do not assert that $(\sigma_\eta)_{\eta>0}$ is strongly continuous on $\mathfrak J$ itself.)
  
  (iv) From the definition of generator, we have,
  \begin{equation*}
  \delta k(y,z)=\left.\frac{\rd}{\rd\eta}\right|_{\eta=0}\sigma_\eta(K)(y,z)=\left.\frac{\rd}{\rd\eta}\right|_{\eta=0}k(y+\eta,z+\eta)
  =\bigl(\pa_y+\pa_z\bigr)k(y,z).
  \end{equation*}
    
  (v) We have, $\det(\id+\sigma_\eta(K))=\det(\id+S_\eta^\dag KS_\eta) =\det(\id+KS_\eta S_\eta^\dag)=\det(\id+K\mathrm{Pr}_{(\eta, \infty )})$.
  By Mercer's formula we have, $\mathrm{tr}\,\sigma_\eta(K)=\int_0^\infty K(y+\eta,y+\eta)\,\rd y$,
  and we can differentiate this formula using the fundamental theorem of calculus.

  (vi) The result (v) may be formulated in terms of the generator without explicit mention of the semigroup.
  Let ${\mathcal D}^1\coloneqq\{\phi\in L^2((0,\infty);\mathbb C)\colon \phi'\in L^2((0,\infty);\mathbb C)\}$, and recall from
  Hille and Phillips~\cite[p.~535]{HP} that ${\mathcal D}^1$ is the domain of the generator of $(S_\eta^\dag)_{\eta>0}$.
  By Plancherel's theorem,
  we have ${\mathcal D}^1\subset L^\infty$, so ${\mathcal D}^1$ is an algebra under pointwise multiplication of functions;
  hence there is a map $\mu\colon {\mathcal D}^1\otimes{\mathcal D}^1\to {\mathcal D}^1$ given by $\phi(y)\psi(z)\mapsto \phi(y)\psi(y)$.
  There is also a natural inclusion ${\mathcal D}^1\otimes {\mathcal D}^1\to L^2\otimes L^2=\mathfrak J_1(L^2)$, and the trace satisfies
  $\mathrm{tr}\,(K)=\int_0^\infty \mu (K)(x)\,\rd x$.
  We have $\delta\colon{\mathcal D}^1\otimes{\mathcal D}^1\to L^2\otimes L^2\colon\phi\otimes\psi\mapsto\phi'\otimes\psi+\phi\otimes\psi'$;
  hence for $k=\sum_{j=1}^\infty \phi_j(y)\psi_j(z)$ we have,
  \begin{equation*}
   \mathrm{tr}\,\delta K=\sum_{j=1}^\infty \int_0^\infty \bigl(\phi_j(y)\psi_j(y)+\phi_j(y)\psi_j'(y)\bigr)\,\rd y
    =-\sum_{j=1}^\infty \phi_j(0)\psi_j(0),
  \end{equation*}
  so we have the required expression for P\"oppe's bracket operation.\qed
\end{proof}

Suppose that $\phi\in L^2((0,\infty );\mathbb M_{m\times m}(\mathbb C))$.
Then we introduce $\phi_{(x)}(\eta)\coloneqq\phi(2x+\eta)$ and the Hankel operator
$\Gamma_{\phi_{(x)}}:L^2((0,\infty);\mathbb C^{m\times 1})\to L^2((0,\infty);\mathbb C^{m\times 1})$ by,
\begin{equation*}
 \Gamma_{\phi_{(x)}} h(y)=\int_0^\infty \phi (y+z+2x)f(z)\, \rd z
\end{equation*}
for $f\in L^2((0, \infty );\mathbb C^{m\times 1})$.
Suppose, $\int_0^\infty t\|\phi(y)\|^2\,\rd y<\infty$ and $\int_0^\infty t\|\psi(y)\|^2\,\rd y<\infty$.
Then $\Gamma_\phi$ is a Hilbert--Schmidt operator, and $\Gamma_\phi\Gamma_\psi$ is trace class with,
\begin{equation*}
\mathrm{tr}\,\bigl(\Gamma_\phi\Gamma_\psi\bigr)=\int_0^\infty \int_0^\infty \phi(y+z)\psi(y+z)\,\rd dy\rd z=\int_0^\infty y\phi(y)\psi(y)\,\rd y.
\end{equation*}
A bounded linear operator $\Gamma$ on $L^2(0, \infty )$ is Hankel if and only if $S_\eta^\dag\Gamma=\Gamma S_\eta$ for all $\eta>0$.
This may be interpreted as $\pa_x\Gamma=-\Gamma\pa_y$ when we consider operators on $C_c^\infty (0, \infty )$.
In the context of Hankel products, this leads to the following. 
\begin{proposition}[Bracket identities for Hankel operators]\label{prop:bracketidentities}
  We have the following:
  \begin{enumerate}
  \item[(i)] Let $\Gamma_\phi$ be the Hankel operator with kernel $\phi(y+z)$.
   Then $\sigma_\eta(\Gamma_\phi)$ has kernel $\phi (y+z+2\eta)$, so $\sigma_\eta(\Gamma_\phi)=\Gamma_{\phi_{(\eta)}}$;
 \item[(ii)] Let $\phi,\psi\in\mathbb M_{m\times m}(C_c^\infty (0,\infty))$ be functions as above.
   Then $\sigma_{2\eta}(\Gamma_\phi\Gamma_\psi)=\Gamma_{\phi_{(\eta)}}\Gamma_{\psi_{(\eta)}}$, and,
   \begin{equation*}
    \pa_\eta\bigl(\Gamma_{\phi_{(\eta)}}\Gamma_{\psi_{(\eta)}}\bigr),
   \end{equation*}
   is a bounded linear operator of finite rank with rank less than or equal to $m^2$;
 \item[(iii)] Suppose $m=1$, and $\phi,\psi\in C_c^\infty(0,\infty)$. Then $\delta (\Gamma_\phi\Gamma_\psi )$ has rank one;
 \item[(iv)] Conversely, suppose that $K$ is in the domain of $\delta$ and $\delta(K)$ has finite rank. Then $K=\Gamma_\Phi^\top\Gamma_\Psi$;
 \item[(v)] Let $\Gamma_x=\Gamma_{\phi_{(x)}}$ and $\Gamma_x'=\pa_x\Gamma_x$.
   Let $\Theta$ be holomorphic on an open neighbourhood of the spectrum of $\Gamma_x$ for all real $x$. Then we have,
   \begin{equation*}
   \lb \Gamma_x\Theta'(\Gamma_x)\rb_{0,0}=-\pa_x\mathrm{tr}\,\bigl(\Theta(\Gamma_x)\bigr).
  \end{equation*}
  \end{enumerate}
\end{proposition}
\begin{proof}
  Item by item we observe the following.
  (i) This is straightforward, and explains the notation.
  (ii) The Hankel operators $\Gamma_{\phi_{(\eta)}}$ and $\Gamma_{\psi_{(\eta)}}$ are Hilbert--Schmidt, so their product is trace class.
  Then we differentiate the kernel and obtain,
  \begin{equation*}
  \pa_\eta\int_0^\infty \phi(y+\xi+2\eta)\psi(\xi+z+2\eta)\,\rd\xi=-2\phi(y+2\eta)\psi(z+2\eta),
  \end{equation*}
  which gives an element of the vector space $\mathbb M_{m\times m}(\mathbb C)$ of dimension $m^2$.
  (iii) We have for $m=1$, $\delta(\Gamma_\phi\Gamma_\psi)(y,z)=-\phi(y)\psi(z)$.
  (iv) By hypothesis, there exist $\phi_j,\psi_j\in L^2(0,\infty)$ for $j=1,\ldots,n$ such that $\delta K(y,z)=-\sum_{j=1}^n \phi_j(y)\psi_j(z)$.
  Then we introduce the vector functions $\Phi=(\phi_1,\ldots, \phi_n)^{\mathrm{T}}$ and $\Psi=(\psi_1,\ldots, \psi_n)^{\mathrm{T}}$
  so that $\Gamma_\Phi^\top\Gamma_\Psi =\sum_{j=1}^n \Gamma_{\phi_j}\Gamma_{\psi_j}$; and also $\delta(\Gamma_\Phi^\top\Gamma_\Psi)=-\sum_{j=1}^n \phi_j(y)\psi_j(z)$.
  We consider $W=K-\Gamma_\Phi^\top\Gamma_\Psi$ which belongs to the domain of $\delta$ with $\delta(W)=0$, hence, 
  $\sigma_\eta(W)=W+\int_0^\eta\sigma_\xi(\delta W)\,\rd\xi=W$,
  where $\sigma_\eta(W)\to 0$ as $t\to\infty$; thus $W=0$ and $K=\Gamma_\Phi^\top\Gamma_\Psi $ is a Hankel product.
  (v) For even powers we have,
  \begin{equation*}
  2\pa_x\mathrm{tr}\,\Gamma_x^{2k}=2\sum_{j=0}^{k-1}\mathrm{tr}\,\bigl(\Gamma_x^{2j}\pa_x(\Gamma_x^2)\Gamma_x^{2(k-j-1)}\bigr)
  =2k\,\mathrm{tr}\,\bigl(\Gamma_x^{2k-2}\pa_x\Gamma_x^2\bigr),
  \end{equation*}
  where the final operator has finite rank. For odd powers, since,
  \begin{align*}
    2\pa_x\mathrm{tr}\,\bigl(\Gamma_x^{2k+1}\bigr)
  =&\;\mathrm{tr}\,\bigl(\Gamma_x'\Gamma_2^{2k}+\Gamma_x\Gamma_x'\Gamma_x^{2k-1}+\cdots+\Gamma_x^{2k}\Gamma_x'\bigr)\\
  &\;+\mathrm{tr}\,\bigl(\Gamma_x'\Gamma_2^{2k}+\Gamma_x\Gamma_x'\Gamma_x^{2k-1}+\cdots+\Gamma_x^{2k}\Gamma_x'\bigr),
  \end{align*}
  and then we move the terms in the second list one step to the left, except for the first, which we move to the end. Thus we obtain,
  \begin{multline*}
  \mathrm{tr}\,\bigl((\Gamma_x'\Gamma_x+\Gamma_x\Gamma_x')\Gamma_x^{2k-1}
  +\Gamma_x (\Gamma_x'\Gamma_x+\Gamma_x\Gamma_x')\Gamma_x^{2k-2}+\cdots+\Gamma_x^{2k-1}(\Gamma_x'\Gamma_x+\Gamma_x\Gamma_x')\bigr)\\
  =(2k+1)\,\mathrm{tr}\,\bigl(\Gamma_x^{2k-1}\pa_x\Gamma_x^2\bigr),
  \end{multline*}
  where again the final operator has finite rank.

  Suppose that the spectrum of $\Gamma_x$ is contained in $D(0,r)$ for some $r>0$. Then for all $s$ such that $\vert s\vert >r$,
  we have a convergent power series $(s-\gamma )^{-1}=\sum_{j=0}^\infty \gamma^j/s^{j+1}$ for all $\gamma$ in the spectrum of $\Gamma_x$ and,
  \begin{equation*}
  2\,\pa_x\mathrm{tr}\,\bigl((s\cdot\id-\Gamma_x)^{-1}\bigr)=\sum_{j=1}^\infty \frac{2}{s^{j+1}}\pa_x\mathrm{tr}\,\bigl(\Gamma^j_x\bigr).
  \end{equation*}
  Then for the first term in the series, we have,
  \begin{equation*}
  2\,\mathrm{tr}\,(\Gamma_x')=2\,\pa_x\int_0^\infty \phi(2y+2x)\,\rd y=-2\phi(2x)=-2\lb\Gamma_x\rb_{0,0},
  \end{equation*}
  and for the remaining terms in the series, $\mathrm{tr}\,(\Gamma_x^{j-2}\pa_x\Gamma_x^2)$ equals,
  \begin{equation*}
    -2\int_0^\infty\cdots\int_0^\infty \phi_{(x)}(y_0+y_1)\phi_{(x)}(y_1+y_2)\dots \phi_{(x)}(y_{j-1}+y_0)dy_0\dots dy_{j-1},
  \end{equation*}
  which equals $-2\lb\Gamma_x^j\rb_{0,0}$.
  The result $-\pa_x\mathrm{tr}\,\bigl((s\cdot\id-\Gamma_x)^{-1}\bigr)=\lb\Gamma_x(s\cdot\id-\Gamma_x)^{-2}\rb_{0,0}$
  holds for all $s$ such that $|s|>r$ follows when we multiply through by $-1/2$ and sum over $j$.
  By analytic continuation, we have the same identity for all $s$ in the unbounded component of the complement of the spectrum of $\Gamma_x$ in the complex plane.

  Now let $\Theta$ be holomorphic on an open neighbourhood of the spectrum of $\Gamma_x$ for all $x>0$.
  Note that $\|\Gamma_x\|\to 0$ as $x\to\infty$, so this uniformity is a mild restriction.
  Then there exists a contour $C$ that winds round the spectrum of $\Gamma_x$ once in the positive sense, so by Cauchy's integral formula
  $\Theta (\Gamma_x)=(2\pi i)^{-1}\int_C (s\cdot\id-\Gamma_x)^{-1}\Theta(\zeta)\,\rd\zeta$, hence, 
  \begin{align*}
    -\pa_x\mathrm{tr}\,\bigl(\Theta(\Gamma_x)\bigr)&=-\frac{1}{2\pi i}\int_C \pa_x\mathrm{tr}\,\bigl((s\cdot\id-\Gamma_x)^{-1}\bigr)\Theta(\zeta)\,\rd\zeta\\
    &=\frac{1}{2\pi i}\int_C\lb\Gamma_x (s\cdot\id-\Gamma_x)^{-2}\rb_{0,0}\Theta(\zeta)\,\rd\zeta,
  \end{align*}
  which integrates to $\lb\Gamma_x\Theta'(\Gamma_x)\rb_{0,0}$. The proof is complete. \qed
\end{proof}
\begin{example}
  For real-valued $\phi$, and for $\Theta(\zeta)=\zeta^2$, the basic formulae are,
  \begin{equation*} 
    \|\Gamma_x\|^2_{\mathfrak J_2}=\mathrm{tr}\,(\Gamma_x^2)=\int_0^\infty \!\!y\phi_{(x)}(y)^2\,\rd y
    \quad\text{and}\quad
    \lb\Gamma_x^2\rb_{0,0}=-\tfrac12\pa_x\mathrm{tr}\,(\Gamma_x^2)=\!\int_0^\infty\!\!\phi_{(x)}(y)^2\,\rd y.
  \end{equation*}    
\end{example}

\section{P\"oppe algebra}\label{sec:Poppealg}
We prescribe the kernel algebra generated by the quantities $\lb V\rb$ and $\lb V^\dag\rb$ and their derivatives,
based on the P\"oppe product in Lemma~\ref{lemma:origPoppeproduct}, as well as a subalgebra generated by the quantity $\lb V-V^\dag\rb$ and its derivatives.
We also outline abstract versions of these algebras to aid computations.
We nominate the abstract algebra as the \emph{P\"oppe algebra} and the corresponding subalgebra the \emph{skew-P\"oppe algebra}.
The P\"oppe algebra outlined herein represents a generalisation of the P\"oppe algebra used in Malham~\cite{Malham:KdVhierarchy}
to derive the non-commutative Korteweg--de Vries hierarchy.
We begin with some preliminary identities. Given a Hankel Hilbert--Schmidt operator $P=P(x,t)$ on $\Vb$,
depending on the parameters $x\in\R$ and $t\geqslant0$, we set, 
\begin{equation*}
V\coloneqq (\id-\mathrm{i}P)^{-1}.
\end{equation*}
Recall that $P$ is self adjoint, so that $P^\dag=P$ and thus $(\mathrm{i}P)^\dag=-\mathrm{i}P$.
\begin{lemma}[Operator identities]\label{lemma:alg}
Given a Hankel operator $P=P(x,t)$ which is Hilbert--Schmidt valued, and the definition $V\coloneqq(\id-\mathrm{i}P)^{-1}$, with
$V^\dag=(\id+\mathrm{i}P)^{-1}$ the adjoint operator to $V$, we observe that,
\begin{equation*}
V\equiv\id+(\mathrm{i}P)V\equiv\id+V(\mathrm{i}P)
\quad\text{and}\quad V^\dag\equiv\id+(\mathrm{i}P)^\dag V^\dag\equiv\id+V^\dag(\mathrm{i}P)^\dag,
\end{equation*}
and further that, $V-V^\dag\equiv2\,V(\mathrm{i}P)V^\dag$.
\end{lemma}
\begin{proof}
All these identities follow directly from the definitions of $V$ and $V^\dag$ and partial fraction identities.
\qed
\end{proof}
\begin{definition}[Fredholm Grassmannian flow]\label{def:FGflow}
Given a Hankel Hilbert--Schmidt operator $P=P(x,t)$ on $\Vb$, depending on the parameters
$x\in\R$ and $t\geqslant0$, we define the operator $G$ by,
\begin{equation*}
G\coloneqq V-V^\dag. 
\end{equation*}
\end{definition}
\begin{remark}\label{rmk:order}
Using Lemma~\ref{lemma:alg}, we can write $V-V^\dag=2\,V(\mathrm{i}P)V^\dag=2\,(\mathrm{i}P)U$,
where $U\coloneqq\bigl(\id+P^2\bigr)^{-1}$. This is possible because $V$ and $V^\dag$
can be expressed as power series in $P$ with scalar coefficients.
Hence the order of the operators in $V(\mathrm{i}P)V^\dag$ does not matter.
Thus, as outlined in detail in Doikou \textit{et al.\/} \cite[Sec.~2]{DMSW:AGFintegrable},
the flow of $G$ represents a Fredholm Grassmannian flow. 
\end{remark}
It is now helpful to define the \emph{signature character}, given previously
in Malham~\cite{Malham:KdVhierarchy} and Doikou \textit{et al.\/} \cite{DMSW:AGFintegrable}.
Let $\mathbb N^\ast$ denote the free monoid of words on $\mathbb N$, i.e.\/ the set
of all possible words of the form $a_1a_2\cdots a_k$ we can construct from letters $a_1, a_2, \ldots, a_k\in\mathbb N$.
\begin{definition}[Signature character]\label{def:signaturecharacter}
  Suppose $a_1a_2\cdots a_n\in\mathbb N^\ast$. The \emph{signature character} $\chi\colon\mathbb N^\ast\to\mathbb Q$
  of any such word is given by the product of Leibniz coefficients,
  \begin{equation*}
    \chi\colon a_1a_2\cdots a_n\mapsto\prod_{k=1}^n\begin{pmatrix} a_k+\cdots+a_n\\a_k\end{pmatrix}.
  \end{equation*}
\end{definition}
Let $\mathcal C(n)$ denote the set of all compositions of $n\in\mathbb N$.
The following result is equivalent to that in Malham~\cite[Lemma~2]{Malham:KdVhierarchy}
and Doikou \textit{et al.\/} \cite[Lemma~8]{DMSW:AGFintegrable}, where detailed proofs can be found.
For any integer $k$, we set $(\mathrm{i}P)_k\coloneqq\pa^k(\mathrm{i}P)$, $V_k\coloneqq\pa^kV$ and $V^\dag_k\coloneqq\pa^kV^\dag$.
For example, if $k=2$, then $V_2=\pa^2V$, while if $k=-1$, then $V_{-1}=\pa^{-1}V$. 
\begin{lemma}[Kernel signature expansion]\label{lemma:calc}
Given a Hankel operator $P=P(x,t)$ and that $V\coloneqq(\id-\mathrm{i}P)^{-1}$ with
$V^\dag=(\id+\mathrm{i}P)^{-1}$, we observe that $\pa V\equiv V(\mathrm{i}P)_1V$ and
$\pa V^\dag\equiv V^\dag(\mathrm{i}P)_1^\dag V^\dag$. With the sum over all compositions $a_1\cdots a_k\in\mathcal C(n)$, we have,
\begin{equation*}
  V_n=\sum\chi\bigl(a_1\cdots a_n\bigr)\,V(\mathrm{i}P)_{a_1}V\cdots V(\mathrm{i}P)_{a_k}V,
\end{equation*}
with the corresponding generalisation for $V_n^\dag$. In particular we have,
\begin{equation*}
  V_n-V_n^\dag=\sum\chi\bigl(a_1\cdots a_n\bigr)\,
  \bigl(V(\mathrm{i}P)_{a_1}V\cdots V(\mathrm{i}P)_{a_k}V-V^\dag(\mathrm{i}P)^\dag_{a_1}V^\dag\cdots V^\dag(\mathrm{i}P)_{a_k}^\dag V^\dag\bigr).
\end{equation*}
\end{lemma}

We now construct the algebra generated by $\lb V-V^\dag\rb$ and its derivatives. For any Hilbert--Schmidt operator $W$, we set,
\begin{equation*}
     [W]\coloneqq\lb W-W^\dag\rb\qquad\text{and}\qquad \{W\}\coloneqq\lb W+W^\dag\rb.
\end{equation*}
In other words the bracket `$[\,\cdot\,]$' generates the kernel of the difference, between its operator argument
and corresponding adjoint. It is the kernel of the skew-symmetric part of its operator argument. It is not a commutator.
Thus if $v=v(y,z;x,t)$ is the matrix-valued kernel corresponding to $V$ which depends on the parameters $x$ and $t$,
then $[V]=v(y,z;x,t)-v^\dag(z,y;x,t)$, where now $v^\dag$ is the complex conjugate transpose of the matrix $v$.
Analogously, $\{\,\cdot\,\}$ generates the kernel of the symmetric part of its operator argument.
From Lemma~\ref{lemma:calc}, we observe that,
\begin{equation*}
  [V_n]=\sum\chi\bigl(a_1\cdots a_n\bigr)\,\bigl[V(\mathrm{i}P)_{a_1}V\cdots V(\mathrm{i}P)_{a_k}V\bigr].
\end{equation*}
The kernel monomial algebra is generated by the monomials $\bigl[V(\mathrm{i}P)_{a_1}V\cdots V(\mathrm{i}P)_{a_k}V\bigr]$,
including monomials of this form with one or more of the $V$'s shown being replaced by $V^\dag$.
The P\"oppe product from Lemma~\ref{lemma:origPoppeproduct} generates closed-form identities for products of such monomials. 
In particular, we have the following.
\begin{lemma}[P\"oppe kernel product identities]\label{lemma:Poppeproductidentities}
  For arbitrary Hilbert--Schmidt operators $F$ and $F^\prime$ and a Hankel Hilbert--Schmidt operator $P$ with
  parameter $x$ and a smooth kernel, we have the following,
\begin{align*}
  \lb F(\mathrm{i}P)_aV\rb\,\lb V(\mathrm{i}P)_bF^\prime\rb
  =&\;\lb F(\mathrm{i}P)_{a+1}V(\mathrm{i}P)_bF^\prime\rb+\lb F(\mathrm{i}P)_aV(\mathrm{i}P)_{b+1}F^\prime\rb\\
  &\;+2\,\lb F(\mathrm{i}P)_aV(\mathrm{i}P)_1V(\mathrm{i}P)_bF^\prime\rb,\\
  \lb F(\mathrm{i}P_a)V^\dag\rb\,\lb V^\dag(\mathrm{i}P_b)F^\prime\rb
  =&\;\lb F(\mathrm{i}P)_{a+1}V^\dag(\mathrm{i}P)_bF^\prime\rb+\lb F(\mathrm{i}P)_aV^\dag(\mathrm{i}P)_{b+1}F^\prime\rb\\
  &\;+2\,\lb F(\mathrm{i}P)_aV^\dag(\mathrm{i}P)_1V^\dag(\mathrm{i}P)_bF^\prime\rb,\\
  \lb F(\mathrm{i}P)_aV^\dag\rb\,\lb V(\mathrm{i}P)_bF^\prime\rb
  =&\;\lb F(\mathrm{i}P)_{a+1}V(\mathrm{i}P)_bF^\prime\rb+\lb F(\mathrm{i}P)_aV^\dag(\mathrm{i}P)_{b+1}F^\prime\rb,\\
  \lb F(\mathrm{i}P)_aV\rb\,\lb V^\dag(\mathrm{i}P)_bF^\prime\rb
  =&\;\lb F(\mathrm{i}P)_{a+1}V^\dag(\mathrm{i}P)_bF^\prime\rb+\lb F(\mathrm{i}P)_aV(\mathrm{i}P)_{b+1}F^\prime\rb.
\end{align*}
\end{lemma}
\begin{proof}
The results stated are established straightforwardly. 
For example, using the identities in Lemma~\ref{lemma:alg}
and the basic P\"oppe product rule, we observe, 
\begin{align*}
\lb F(\mathrm{i}P)_aV\rb\,\lb V(\mathrm{i}P)_b F^\prime\rb
=&\;\lb F(\mathrm{i}P)_a+F(\mathrm{i}P)_aV(\mathrm{i}P)\rb
\lb(\mathrm{i}P)_bF^\prime+(\mathrm{i}P)V(\mathrm{i}P)_b F^\prime\rb\\
=&\;\lb F(\mathrm{i}P)_{a+1}(\mathrm{i}P)_b F^\prime\rb
+\lb F(\mathrm{i}P)_a(\mathrm{i}P)_{b+1} F^\prime\rb\\
&\;+\lb F(\mathrm{i}P)_{a+1}(\mathrm{i}P)V(\mathrm{i}P)_b F^\prime\rb
+\lb F(\mathrm{i}P)_{a}(\mathrm{i}P)_1V(\mathrm{i}P)_b F^\prime\rb\\
&\;+\lb F(\mathrm{i}P)_{a}V(\mathrm{i}P)_1(\mathrm{i}P)_b F^\prime\rb
+\lb F(\mathrm{i}P)_{a}V(\mathrm{i}P)(\mathrm{i}P)_{b+1} F^\prime\rb\\
&\;+\lb F(\mathrm{i}P)_aV(\mathrm{i}P)_1(\mathrm{i}P)V(\mathrm{i}P)_b F^\prime\rb\\
&\;+\lb F(\mathrm{i}P)_aV(\mathrm{i}P)(\mathrm{i}P)_1V(\mathrm{i}P)_b F^\prime\rb.
\end{align*}
Combining terms using the identities in Lemma~\ref{lemma:alg} generates the first result claimed. And so forth.
\qed
\end{proof}
\begin{remark}[Algebra of kernel monomials: abstract encoding]\label{rmk:monomialalgebra}
  As mentioned, the set of all kernel monomials of the form $\lb V(\mathrm{i}P)_{a_1}V(\mathrm{i}P)_{a_2}V\cdots V(\mathrm{i}P)_{a_k}V\rb$,
  where any of the $V$'s shown may be replaced by $V^\dag$, with the P\"oppe kernel product defined in Lemma~\ref{lemma:Poppeproductidentities},
  form a closed algebra of such monomials. We assume here that all the derivatives of the $P$ operator exist and are Hilbert--Schmidt valued.
  At this stage it is useful to consider an abstract encoding of this kernel monomial algebra, equipped with the P\"oppe kernel product.
  The abstract algebra is constructed by simply stripping the `$\mathrm{i}P$' and `$V$' labels from the kernel monomials, and respectively,
  replacing them by the composition components $a_1a_2\cdots a_k$, together with a binary encoding of whether an intervening operator is a $V$
  or $V^\dag$, i.e.\/ we replace,
  \begin{equation*}
  \lb V(\mathrm{i}P)_{a_1}V(\mathrm{i}P)_{a_2}V\cdots V(\mathrm{i}P)_{a_k}V\rb\to\bs0 a_1\bs0 a_2\bs0\cdots\bs0 a_k\bs0,
  \end{equation*}
  where any of the $\bs0$'s shown, corresponding to the $V$ operator, may be replaced by $\bs0^\dag$ in the corresponding position
  that a $V^\dag$ operator is present in the monomial on the left.
  In essence, the P\"oppe kernel product defined in Lemma~\ref{lemma:Poppeproductidentities} involves operations on these stripped down
  components only, i.e.\/ operations on the forms $\bs0 a_1\bs0 a_2\bs0\cdots\bs0 a_k\bs0$, where again some $\bs0$'s may be replaced by $\bs0^\dag$. 
  Since we mirror the P\"oppe product in the abstract setting in Definition~\ref{def:abstractPoppeproduct},
  we know that the kernel monomial algebra and our abstract algebra encoding just below, are isomorphic.
\end{remark}
Let us now introduce our abstract encoding for the algebra of operator kernel monomials
equipped with the P\"oppe product in Lemma~\ref{lemma:origPoppeproduct}, just mentioned.
Given a word $w=a_1a_2\cdots a_k$ generated using letters $a_1$, $a_2$, $\ldots$, $a_k$ from $\mathbb Z$,
and a word $\bs\varphi=\bs\theta_1\bs\theta_2\cdots\bs\theta_{k+1}$ generated using the letters $\bs\theta_1$, $\bs\theta_2$, $\ldots$, $\bs\theta_{k+1}$
chosen from the binary set $\{\bs 0,\bs 0^\dag\}$, let $w\times\bs\vartheta$ denote the corresponding word,
\begin{equation*}
w\times\bs\varphi=\bs\theta_1 a_1\bs\theta_2 a_2\bs\theta_3\cdots\bs\theta_k a_k\bs\theta_{k+1},
\end{equation*}
in the free monoid $(\mathbb Z_{\bs 0})^\ast$ where $\mathbb Z_{\bs 0}\coloneqq\mathbb Z\cup\{\bs 0,\bs 0^\dag\}$.
For such words, there is a single letter from the binary set $\{\bs 0,\bs 0^\dag\}$ sandwiched between each of the letters from $\mathbb Z$,
as well as one at each end. Let $\mathbb C\langle \mathbb Z_{\bs 0}\rangle$ denote the non-commutative polynomial algebra over $\mathbb C$
generated by words from $(\mathbb Z_{\bs 0})^\ast$, endowed with the following P\"oppe product.
\begin{definition}[P\"oppe product]\label{def:abstractPoppeproduct}
  Consider four words from $(\mathbb Z_{\bs 0})^\ast$ of the form $ua\bs 0$, $ua\bs 0^\dag$, $\bs 0bv$ and $\bs0^\dag bv$,
  where $u$ and $v$ are any subwords from $(\mathbb Z_{\bs 0})^\ast$ and $a,b\in\mathbb Z$. We define the \emph{P\"oppe product}
  from $\mathbb C\langle \mathbb Z_{\bs 0}\rangle\times\mathbb C\langle \mathbb Z_{\bs 0}\rangle$ to
  $\mathbb C\langle \mathbb Z_{\bs 0}\rangle$ of these words to be,
  \begin{align*}
    (ua\bs 0)(\bs 0 bv)&=u(a+1)\bs0 bv+ua\bs 0(b+1)v+2\cdot ua\bs 01\bs 0bv,\\
    (ua\bs 0^\dag)(\bs 0^\dag bv)&=u(a+1)\bs0^\dag bv+ua\bs 0^\dag(b+1)v+2\cdot ua\bs 0^\dag1\bs 0^\dag bv,\\
    (ua\bs 0^\dag)(\bs 0 bv)&=u(a+1)\bs0 bv+ua\bs 0^\dag(b+1)v\\
    (ua\bs 0)(\bs 0^\dag bv)&=u(a+1)\bs0^\dag bv+ua\bs 0(b+1)v.   
  \end{align*}
\end{definition}
Let $\nu$ denote the empty word in $\mathbb C\langle \mathbb Z_{\bs 0}\rangle$.
Then for any word $w\times\bs\varphi\in\mathbb C\langle \mathbb Z_{\bs 0}\rangle$
we have $\nu\,(w\times\bs\varphi)=(w\times\bs\varphi)\,\nu=w\times\bs\varphi$.
Let $\mathcal C\coloneqq\cup_{n\geqslant0}\mathcal C(n)$ denote the set of all compositions.
\begin{definition}[Signature expansion]\label{def:signatureexpansion}
For any $n\in\mathbb N\cup\{0\}$, we define the following linear \emph{signature expansions} $\bs n\in\mathbb C\langle\mathbb Z_{\bs 0}\rangle$,
\begin{equation*}
\bs n\coloneqq\sum_{a_1a_2\cdots a_k\in\mathcal C(n)}\chi(a_1a_2\cdots a_k)\,\cdot\bs 0a_1\bs 0a_2\bs 0\cdots\bs 0a_k\bs 0,
\end{equation*}
where the sum is over all possible compositions $a_1a_2\cdots a_k$, with $k\geqslant1$ parts, of $n$.
\end{definition}
For example, we note that $\bs 1=\chi(1)\cdot\bs 01\bs 0$ and $\bs 2=\chi(2)\cdot \bs 02\bs 0+\chi(11)\cdot\bs 01\bs 01\bs 0$.
Further note, for the case $n=0$, the signature expansion simply corresponds to the letter $\bs 0$ from the binary set $\{\bs 0,\bs 0^\dag\}$.
Equivalently we can write the relation in Definition~\ref{def:signatureexpansion} for the case $n=0$ as $\bs 0=\chi(0)\cdot\bs 0$.
Naturally by convention, we take $\chi(0)=1$.
Let us also remark on the following basic identities in $\mathbb C\langle \mathbb Z_{\bs 0}\rangle$, which follow from Lemma~\ref{lemma:alg}.
\begin{lemma}[Algebraic identities]\label{lemma:basicidentities}
  We have the following basic relations in $\mathbb C\langle \mathbb Z_{\bs 0}\rangle$,
  \begin{equation*}
    \bs 0\equiv\nu+\bs 00\equiv\nu+0\bs 0, \quad \bs 0^\dag\equiv\nu+\bs 0^\dag0^\dag\equiv\nu+0^\dag\bs 0^\dag
    \quad\text{and}\quad \bs 0-\bs 0^\dag=2\cdot\bs 00\bs 0^\dag.
  \end{equation*}
\end{lemma}
\begin{definition}\label{rmk:convention}
  Given any word $w\times\bs\varphi\in\mathbb C\langle \mathbb Z_{\bs 0}\rangle$, say 
  $w\times\bs\varphi=\bs\theta_1 a_1\bs\theta_2 a_2\bs\theta_3\cdots\bs\theta_k a_k\bs\theta_{k+1}$,
  the letters $a_i^\dag$ denote the letters `$-a_i$' from $\mathbb Z$, i.e.\/ $a_i^\dag=-a_i$. Further we set,
  \begin{equation*}
  \bigl(\bs\theta_1 a_1\bs\theta_2 a_2\bs\theta_3\cdots\bs\theta_k a_k\bs\theta_{k+1}\bigr)^\dag\coloneqq
  \bs\theta_1^\dag a_1^\dag \bs\theta_2^\dag  a_2^\dag \bs\theta_3^\dag \cdots\bs\theta_k^\dag  a_k^\dag \bs\theta_{k+1}^\dag ,
  \end{equation*}
  i.e.\/ we replace all the letters $a_i$ in $w$ by their counterparts $a_i^\dag=-a_i$ and all the letters in $\bs\theta$ by their counterparts.
  In the latter instance this means we change all the $\bs 0$'s to $\bs 0^\dag$, and vice-versa.
  Note we \emph{do not} reverse the order of the terms in $w\times\bs\varphi$. This means, for example, that
  we can interpret $(w\times\bs\varphi)^\dag=w^\dag\times\bs\varphi^\dag$, and since $w^\dag=(-1)^{|w|}w$,
  where $|w|$ is the length of $w$, then $(w\times\bs\varphi)^\dag\in\mathbb C\langle \mathbb Z_{\bs 0}\rangle$.
\end{definition}
Consider the following skew-symmetric and symmetric forms on $\mathbb C\langle \mathbb Z_{\bs 0}\rangle$.
\begin{definition}[Skew-symmetric and symmetric forms]\label{def:skewsymmparts}
Given any word $w\times\bs\varphi\in\mathbb C\langle \mathbb Z_{\bs 0}\rangle$,
we define its skew-symmetric and symmetric forms in $\mathbb C\langle \mathbb Z_{\bs 0}\rangle$,
respectively, by 
\begin{equation*}
  [w\times\bs\varphi]\coloneqq w\times\bs\varphi-(w\times\bs\varphi)^\dag
  \qquad\text{and}\qquad \{w\times\bs\varphi\}\coloneqq w\times\bs\varphi+(w\times\bs\varphi)^\dag .
\end{equation*}
Naturally we have $\bigl[(w\times\bs\varphi)^\dag\bigr]=-[w\times\bs\varphi]$ and $\bigl\{(w\times\bs\varphi)^\dag\bigr\}=\{w\times\bs\varphi\}$.
\end{definition}
The following product rules based on the P\"oppe product in Definition~\ref{def:abstractPoppeproduct},
are useful for our computations in all subsequent sections.
\begin{lemma}[Skew and symmetric P\"oppe products]\label{lemma:skewandsymmPoppeproducts}
  Consider the elements $[ua\bs 0]$, $\{ua\bs 0\}$, $[\bs 0bv]$ and $\{\bs0 bv\}$ from $\mathbb C\langle \mathbb Z_{\bs 0}\rangle$,
  where $u$ and $v$ are any subwords from $\mathbb C\langle \mathbb Z_{\bs 0}\rangle$ and $a,b\in\mathbb Z$.
  We have the following \emph{P\"oppe products} in $\mathbb C\langle \mathbb Z_{\bs 0}\rangle$ between these elements,
  \begin{align*}
    \{ua\bs 0\}\,[\bs 0 bv]&=\bigl[u(a+1)[\bs 0 bv]\bigr]+\bigl[ua\bs 0[(b+1)v]\bigr]+2\cdot [ua\bs 01\bs 0bv],\\
    [ua\bs 0]\,\{\bs 0 bv\}&=\bigl[u(a+1)\{\bs 0 bv\}\bigr]+\bigl[ua\bs 0\{(b+1)v\}\bigr]+2\cdot [ua\bs 01\bs 0 bv],\\
    [ua\bs 0]\,[\bs 0 bv]&=\bigl\{u(a+1)[\bs 0 bv]\bigr\}+\bigl\{ua\bs 0[(b+1)v]\bigr\}+2\cdot \{ua\bs 01\bs 0bv\},\\
    \{ua\bs 0\}\,\{\bs 0 bv\}&=\bigl\{u(a+1)\{\bs 0 bv\}\bigr\}+\bigl\{ua\bs 0\{(b+1)v\}\bigr\}+2\cdot \{ua\bs 01\bs 0bv\}.    
  \end{align*}
  These products also hold when $[ua\bs0]=[\bs0]$, which case term on the right involving `$(a+1)$' is absent.
  Likewise, these products also hold when $[\bs0bv]=[\bs0]$, in which case the term on the right involving `$(b+1)$' is absent. 
\end{lemma}
\begin{proof}
  The results are established straightforwardly using Definition~\ref{def:abstractPoppeproduct} for the abstract P\"oppe product.
  Consider for example the first product shown, we observe that,
  \begin{align*}
    \{ua\bs 0\}\,[\bs 0 bv]=&\;\bigl(ua\bs 0+(ua\bs 0)^\dag\bigr)\bigl(\bs 0 bv-(\bs 0 bv)^\dag\bigr)\\
    =&\;u(a+1)\bs0 bv+ua\bs 0(b+1)v+2\cdot ua\bs 01\bs 0bv\\
    &\;-u(a+1)(\bs 0bv)^\dag-ua\bs 0\bigl((b+1)v\bigr)^\dag\\
     &\;+\bigl(u(a+1)\bigr)^\dag\bs0 bv+(ua\bs 0)^\dag(b+1)v\\
    &\;-\bigl(u(a+1)\bs 0bv\bigr)^\dag-\bigl(ua\bs 0(b+1)v\bigr)^\dag-2\cdot\bigl(ua\bs 01\bs 0 bv\bigr)^\dag,
  \end{align*}
  which gives the first product result. The other three cases follow completely analogously.
  For the case, for example, when $[ua\bs0]=[\bs0]$ in the second product, we use that,
  since $[\bs0]=2\cdot\bs00\bs0^\dag$ and $\bs00\bs0^\dag=\bs0^\dag0\bs0$,
  we have $[\bs00\bs0^\dag]=[\bs0^\dag0\bs0]=\bs00\bs0^\dag+\bs0^\dag0\bs0=[\bs0]$.
  Hence we observe, since $[\bs0]=[\bs0^\dag0\bs0]$, we can use the latter form in
  the corresponding product already established, so using the properties of the skew and symmetric forms in
  Definition~\ref{def:skewsymmparts} we have,
  \begin{align*}
    [\bs0]\,\{\bs0bv\}=&\;[\bs0^\dag0\bs0]\,\{\bs0bv\}\\
    =&\;\bigl[\bs0^\dag1\{\bs 0 bv\}\bigr]+\bigl[\bs0^\dag0\bs 0\{(b+1)v\}\bigr]+2\cdot [\bs0^\dag0\bs 01\bs 0 bv]\\
    =&\;[\bs0^\dag1\bs 0 bv]+[\bs01\bs 0 bv]+[\bs0^\dag0\bs 0(b+1)v]+[\bs00\bs 0^\dag(b+1)v]+2\cdot [\bs0^\dag0\bs 01\bs 0 bv]\\
    =&\;[\bs0(b+1)v]-[\bs 0^\dag(b+1)v]+\bigl[(\bs0^\dag+\bs0+2\cdot\bs0^\dag0\bs 0)1\bs 0 bv\bigr]\\
    =&\;\bigl[\bs0\{(b+1)v\}\bigr]+2\cdot[\bs 01\bs 0 bv].
  \end{align*}
  The remaining cases follow completely analogously.
  \qed
\end{proof}
\begin{remark}\label{rmk:prodrulesufficiency}
  We observe that, using the properties of the skew-symmetric and symmetric forms `$[\,\cdot\,]$' and `$\{\,\cdot\,\}$'
  recorded in Definition~\ref{def:skewsymmparts}, the skew and symmetric P\"oppe products quoted in Lemma~\ref{lemma:skewandsymmPoppeproducts}
  are sufficient to resolve the P\"oppe products of all possible skew-symmetric or symmetric forms we might encounter.
  For example if the left factor is of the form $[ua\bs 0^\dag]$ or $\{ua\bs 0^\dag\}$, then
  we can use that $[ua\bs 0^\dag]=-[(ua)^\dag\bs 0]$ or $\{ua\bs 0^\dag\}=\{(ua)^\dag\bs 0\}$ and then apply
  the product rules shown to the latter forms. Similarly we can use that $[\bs 0^\dag bv]=-[\bs 0(bv)^\dag]$
  and $\{\bs 0^\dag bv\}=\{\bs 0(bv)^\dag\}$.  
\end{remark}
\begin{remark}[Minimal product set]\label{rmk:minimalproductset}
  We observe, to compute the P\"oppe product of any monomials of the form
  $[w_1\times\bs\varphi_1]\,[w_2\times\bs\varphi_2]\,\,\cdots\,[w_k\times\bs\varphi_k]$,
  we really only need the rule for $[\,\cdot\,]\,[\,\cdot\,]$ and say the rule
  for $[\,\cdot\,]\,\{\,\cdot\,\}$ in Lemma~\ref{lemma:skewandsymmPoppeproducts}.
  This is because we can work from right to left through the products in such a monomial.
  We can alternatively use $\{\,\cdot\,\}\,[\,\cdot\,]$ and work from left to right.
\end{remark}
\begin{remark}[Basic skew-form properties]\label{rmk:zeroword}
  The skew-form $[\bs0]$ corresponds to the general skew-form $[w\times\bs\varphi]$ in which the
  composition component/word $w=\nu$, the empty word, and $\bs\varphi=\bs0$, i.e.\/ we have $[\nu\times\bs0]=[\bs0]$.
  Note if $\bs\varphi=\bs0^\dag$, this simply corresponds to `$-[\bs0]$' or equivalently `$-[\nu\times\bs0]$'. In P\"oppe products,
  the skew-form $[\bs0]$ has some rather special properties, as highlighted in Lemma~\ref{lemma:skewandsymmPoppeproducts}.
  By Remark~\ref{rmk:minimalproductset} and the second result in proof of Lemma~\ref{lemma:skewandsymmPoppeproducts},
  we have, for example,
  \begin{align*}
    [\bs 0]\,[\bs0 bv]&=\bigl\{\bs0[(b+1)v]\bigr\}+2\cdot\{\bs 01\bs 0bv\},\\
    [ua\bs 0]\,[\bs 0]&=\bigl\{u(a+1)[\bs0]\bigr\}+2\cdot\{ua\bs 01\bs 0\},\\
    [\bs0]\,\{\bs0bv\}&=\bigl[\bs0\{(b+1)v\}\bigr]+2\cdot[\bs 01\bs 0bv].
  \end{align*}
  In particular, setting $bv$ to be the empty word $\nu$, we have $[\bs 0]^{2}=2\cdot\{\bs 01\bs0\}$.
\end{remark}
\begin{remark}[Homomorphic signature character]\label{rmk:homomorphicsigchar}
  Consider a multi-factor product of signature expansions of the form,
  \begin{equation*}
    [\bs n_1]\,[\bs n_2]\,\cdots\,[\bs n_k]
    =\sum\bigl(\chi(w_1)\chi(w_2)\cdots\chi(w_k)\bigr)\cdot[w_1\times\bs\varphi_1]\,[w_2\times\bs\varphi_2]\,\cdots\,[w_k\times\bs\varphi_k],
  \end{equation*}
  where the sum is over all words $w_1\times\bs\varphi_1$ with $w_1\in\mathcal C(n_1)$, $w_2\times\bs\varphi_2$ with $w_2\in\mathcal C(n_2)$, and so forth.
  Note, the form $[w_1\times\bs\varphi_1]\,[w_2\times\bs\varphi_2]\,\cdots\,[w_k\times\bs\varphi_k]$ generates many different
  words in $\mathbb C\langle \mathbb Z_{\bs 0}\rangle$.
  We observe that it would be convenient to encode $\chi(w_1)\chi(w_2)\cdots\chi(w_k)$ as $\chi(w_1\ot w_2\ot\cdots\ot w_k)$.
  Indeed, hereafter, we assume that $\chi$ acts \emph{homomorphically} on any such tensor product of compositions so that indeed we have,
  \begin{equation*}
    \chi(w_1\ot w_2\ot\cdots\ot w_k)\equiv\chi(w_1)\chi(w_2)\cdots\chi(w_k).
  \end{equation*}
\end{remark}
Let us now outline some simple examples. 
\begin{example}\label{ex:tensornotation}
  By definition $[\bs 0]\coloneqq\bs 0-\bs 0^\dag$.
  Using the notation $[\bs 0]^{2}=\bigl(\chi(0)\cdot[\bs 0]\bigr)\,\bigl(\chi(0)\cdot[\bs 0]\bigr)$ and so forth,
  then using the product rules in Lemma~\ref{lemma:skewandsymmPoppeproducts} we observe (also see Remark~\ref{rmk:zeroword}),
  \begin{align*}
    [\bs 0]^{2}=&\;\chi(0\ob0)\cdot\{\bs 01\bs0\},\\
    [\bs 0]^{3}=&\;[\bs 0]\,[\bs 0]^{2}\\
    =&\;\bigl(\chi(0)\cdot[\bs 0]\bigr)\,\bigl(\chi(0\ob0)\{\bs 01\bs0\}\bigr)\\
    =&\;\chi(0\ot0\ob0)\cdot\bigl[\bs 0\{2\bs 0\}\bigr]+\chi(0\ob0\ob0)\cdot[\bs 01\bs 0 1\bs 0],
  \end{align*}
  where the tensor notation `$\ob$' in the argument of $\chi=\chi(\cdot)$ indicates a tensor product `$\ot$'
  together with the fact that an extra real factor of `$2$' should be included
  with the $\chi=\chi(\cdot)$ factor shown. See Remark~\ref{rmk:explainotimeshat} just below.
\end{example}
\begin{remark}\label{rmk:explainotimeshat}
  Hereafter, we also use the tensor notation `$\ob$' in the argument of $\chi=\chi(\cdot)$ to indicate
  when the skew or symmetric form were generated by the `quasi' term $2\cdot[ua\bs01\bs0 bv]$ in one of
  the P\"oppe products in Lemma~\ref{lemma:skewandsymmPoppeproducts}.
  We illustrated this in Example~\ref{ex:tensornotation} just above. We observe therein that
  the result of the product $[\bs0]^2$ is `$2$' times the symmetric form $\{\bs 01\bs0\}$.
  This symmetric form emerges from the `quasi' term in the P\"oppe product of $\chi(0)\cdot[\bs 0]$
  with $\chi(0)\cdot[\bs 0]$ and a natural way to record this is the form $\chi(0\ob0)\cdot\{\bs 01\bs0\}$.
  The tensor product of the zeros in the argument of $\chi=\chi(\cdot)$ indicates that the symmetric
  form is the result of the product of $[\bs0]$ with $[\bs0]$, while the fact that the tensor product
  is `$\ob$' indicates it was the result of the `quasi' term in the P\"oppe product, and an extra factor 
  of $2$ is implied. In this case if we evaluate the signature character we include an extra factor of `$2$'
  in its evaluation. Also consider the product $[\bs0]^3$ in Example~\ref{ex:tensornotation}. 
  When we compute the P\"oppe product
  $\bigl(\chi(0)\cdot[\bs 0]\bigr)\,\bigl(\chi(0\ob0)\{\bs 01\bs0\}\bigr)$, the first skew form generated,
  i.e.\/ $\bigl[\bs 0\{2\bs 0\}\bigr]$, has the coefficient $\chi(0\ot0\ob0)$ as we might expect, using the
  homomorphic properties of $\chi$. However the second term generated by the product
  $\bigl(\chi(0)\cdot[\bs 0]\bigr)\,\bigl(\chi(0\ob0)\{\bs 01\bs0\}\bigr)$, which is $[\bs 01\bs 0 1\bs 0]$,
  has the coefficient $\chi(0\ob0\ob0)$. This is because this second term is the result of the `quasi'
  term $2\cdot[ua\bs01\bs0 bv]$ in the P\"oppe product; here $ua=\nu$ and $bv=1\bs0$. The factor `$2$'
  is absorbed/encoded by the fact that a `$\ob$' tensor (instead of just `$\ot$') is used between the first $0$ and the $0\ob0$,
  the respective $\chi$-arguments for $[\bs0]$ and $\{\bs 01\bs0\}$,
  in the coefficient $\chi(0\ob0\ob0)$ for $[\bs 01\bs 0 1\bs 0]$ in Example~\ref{ex:tensornotation}. 
  See Example~\ref{ex:tensornotation2} for further illustrations of this notation.
\end{remark}
\begin{example}\label{ex:tensornotation2}
 Using the P\"oppe products in Lemma~\ref{lemma:skewandsymmPoppeproducts} and that $[\bs 1]=\chi(1)\cdot[\bs 01\bs0]$, we have,
 \begin{align*}
    [\bs 1]\,[\bs 0]^{2}=&\;\bigl(\chi(1)\cdot[\bs 01\bs0]\bigr)\,\bigl(\chi(0\ob0)\cdot\{\bs 0 1\bs 0\}\bigr)\\
    =&\;\chi(1\ot0\ob0)\cdot\bigl(\bigl[\bs 02\{\bs01\bs0\}\bigr]+\bigl[\bs01\bs0\{2\bs0\}\bigr]\bigr)
    +\chi(1\ob0\ob0)\cdot[\bs01\bs01\bs01\bs0],\\
    [\bs 0]^{2}\,[\bs 1]=&\;\bigl(\chi(0\ob0)\cdot\{\bs 0 1\bs 0\}\bigr)\,\bigl(\chi(1)\cdot[\bs 01\bs0]\bigr)\\
    =&\;\chi(0\ob0\ot 1)\cdot\bigl(\bigl[\bs 02[\bs01\bs0]\bigr]+\bigl[\bs01\bs0[2\bs0]\bigr]\bigr)
    +\chi(0\ob0\ob1)\cdot[\bs01\bs01\bs01\bs0].
  \end{align*}
\end{example}
\begin{definition}[Derivation endomorphism]\label{def:derivation}
  Given any word $w\times\bs\varphi\in\mathbb C\langle \mathbb Z_{\bs 0}\rangle$ with
  $w\times\bs\varphi=\bs\theta_1a_1\bs\theta_2\cdots\bs\theta_ka_k\bs\theta_{k+1}$, we define the 
  \emph{derivation endomorphism} $\mfd$ on $\mathbb C\langle \mathbb Z_{\bs 0}\rangle$ 
  to be the linear expansion,
  \begin{align*}
    \mfd(w\times\bs\varphi)\coloneqq&\;\sum_{\ell=1}^k\bs\theta_1a_1\bs\theta_2\cdots\bs\theta_\ell(a_\ell+1)\bs\theta_{\ell+1}\cdots\bs\theta_ka_k\bs\theta_{k+1}\\
    &\;+\sum_{\ell=1}^{k+1}\bs\theta_1a_1\bs\theta_2\cdots\bs\theta_{\ell-1}a_{\ell-1}(\mfd\bs\theta_\ell)a_\ell\bs\theta_{\ell+1}\cdots\bs\theta_ka_k\bs\theta_{k+1},
  \end{align*}
  where $\mfd\bs\theta_\ell$ equals $\bs 01\bs 0$ or $\bs 0^\dag 1^\dag\bs 0^\dag$,
  depending respectively on whether $\bs\theta_\ell$ is $\bs 0$ or $\bs 0^\dag$. 
\end{definition}
\begin{remark}
  The action of the derivation endomorphism on $\bs 0$ and $\bs 0^\dag$ shown in the definition
  reflects the signature expansions, either at the kernel or abstract level. In this case here, we
  know $\pa V=V(\mathrm{i}P)_1V$ and $\pa V^\dag=V^\dag(\mathrm{i}P)_1^\dag V^\dag$ or equivalently $\bs 1=\chi(1)\cdot\bs 01\bs 0$
  and $\bs 1^\dag=\chi(1)\cdot\bs 0^\dag 1^\dag\bs 0^\dag=-\bs 0^\dag 1\bs 0^\dag$. Similarly, 
  the action of the derivation endomorphism on any signature expansion, say $\bs n$, is given by,
  $\mfd\colon\bs n\mapsto(\bs{n+1})$, and similarly for $\bs n^\dag$.
\end{remark}
Now suppose, within $\mathbb C\langle \mathbb Z_{\bs 0}\rangle$, we restrict ourselves to the set of skew-symmetric forms $[w\times\bs\varphi]$.
Naturally, as a vector space, $\mathbb C\langle \mathbb Z_{\bs 0}\rangle$ can be decomposed into the direct sum of the 
vector subspaces $\mathbb C[\mathbb Z_{\bs 0}]$ of skew-symmetric forms, and $\mathbb C\{\mathbb Z_{\bs 0}\}$ of symmetric forms:
\begin{equation*}
\mathbb C\langle \mathbb Z_{\bs 0}\rangle=\mathbb C[\mathbb Z_{\bs 0}]\,{\scriptsize \bigoplus}\,\mathbb C\{\mathbb Z_{\bs 0}\}.
\end{equation*}
We observe from the P\"oppe products in Lemma~\ref{lemma:skewandsymmPoppeproducts},
the product $[w_1\times\bs\varphi_1]\,[w_2\times\bs\varphi_2]$ does not generate a skew-symmetric form but a symmetric one.
However any triple product $[w_1\times\bs\varphi_1]\,[w_2\times\bs\varphi_2]\,[w_3\times\bs\varphi_3]$
does generate a skew-symmetric form. This is true for any P\"oppe products involving an odd number of skew-symmetric forms. 
Hence we can define a subalgebra of the P\"oppe algebra $\mathbb C\langle \mathbb Z_{\bs 0}\rangle$ which
we denote by $\mathbb C[\mathbb Z_{\bs 0}]\subseteq\mathbb C\langle \mathbb Z_{\bs 0}\rangle$, which is generated
by skew-symmetric forms and triple products of such forms.
\begin{definition}[Skew-P\"oppe algebra]\label{def:skewPoppealg}
We call $\mathbb C[\mathbb Z_{\bs 0}]$ the \emph{skew-P\"oppe algebra}. 
\end{definition}
\begin{remark}[Practical P\"oppe algebra computations]
  In practice, in particular in the next two sections, we perform calculations in the ``enveloping'' algebra $\mathbb C\langle \mathbb Z_{\bs 0}\rangle$,
  and at the end, show that the result remains closed within the skew-P\"oppe algebra $\mathbb C[\mathbb Z_{\bs 0}]$.
  However, the skew-P\"oppe algebra and its triple product structure is crucial to the proof of our main result in Section~\ref{sec:hierarchycoding}.
\end{remark}
%

%
%

\section{Hierarchy examples}\label{sec:ncNLS} 
We use the skew-P\"oppe algebra $\mathbb C[\mathbb Z_{\bs 0}]$ to establish integrability for examples from the
non-commutative nonlinear Schr\"odinger and modified Korteweg--de Vries hierarchy. 
This was first considered in an analogous context in Doikou \textit{et al.\/} \cite[Sec.~6]{DMSW:AGFintegrable}.
Recall the linear dispersion equation, the `base' equation, we introduced in Section~\ref{subsec:lineardispersion}.
Hereafter we assume Hilbert--Schmidt operators $P$ and $G$ satisfy, respectively, the linear dispersive
partial differential equation $\pa_t P=-\mu_n(\mathrm{i}\mathcal I)^{n-1}\pa^n P$ and the linear Fredholm equation
$\mathrm{i}P=G(\id+P^2)$. We also assume $P^\dag=P$. We observe that with $V\coloneqq(\id-\mathrm{i}P)^{-1}$,
then assuming it exists, we have,
\begin{equation*}
  G=V(\mathrm{i}P)V^\dag.
\end{equation*}
Note we scale this by a factor `$2$' presently so that it matches the expression in Definition~\ref{def:FGflow}.
We record the following identities that prove useful below; also see Doikou \textit{et al.} ~\cite[Sec.~6]{DMSW:AGFintegrable}.
Also recall the identities in Lemma~\ref{lemma:alg} and Remark~\ref{rmk:order}.
\begin{lemma}\label{lemma:blockidentities}
The block operators $P$ and $V$ satisfy the following idenitites,
\begin{equation*}
P\mathcal I=-\mathcal I P,\qquad \mathcal IV=V^\dag\mathcal I
\qquad\text{and}\qquad V\mathcal I=\mathcal IV^\dag.
\end{equation*}
\end{lemma}
\begin{proof}
The first identity follows from the block structures assumed for $P$ and $\mathcal I$. The latter
two identities follow using the power series expansion for $V\coloneqq(\id-\mathrm{i}P)^{-1}$. 
\qed
\end{proof}
We now rescale our definition for $G$ above by a factor `$2$', and set,
\begin{equation*}
  G\coloneqq V-V^\dag.
\end{equation*}
Hereafter, we are thus concerned with the quantity $[V]\coloneqq\lb V-V^\dag\rb$. 
Using that $\pa_tV=V\pa_t(\mathrm{i}P)V$ and $\pa_tV^\dag=V^\dag\pa_t(\mathrm{i}P)^\dag V^\dag$,
and that $\pa_t(\mathrm{i}P)=-\mu_n(\mathrm{i}\mathcal I)^{n-1}\pa^n(\mathrm{i}P)$
and $\pa_t(\mathrm{i}P)^\dag=-(-1)^{n-1}\mu_n\pa^n(\mathrm{i}P)^\dag(\mathrm{i}\mathcal I)^{n-1}$,
we observe that for any $n\in\mathbb Z$, we have
\begin{align*}
  \pa_t[V]=&\;V\pa_t(\mathrm{i}P)V-V^\dag\pa_t(\mathrm{i}P)^\dag V^\dag\\
  =&\; -\mu_n\Bigl(V(\mathrm{i}\mathcal I)^{n-1}\pa^n(\mathrm{i}P)V-(-1)^{n-1}V^\dag\pa^n(\mathrm{i}P)^\dag(\mathrm{i}\mathcal I)^{n-1}V^\dag\Bigr).
\end{align*}
For convenience set $\mathcal M_n\coloneqq-\mu_n(\mathrm{i}\mathcal I)^{n-1}$. Using the identities in Lemma~\ref{lemma:blockidentities}, we have,
\begin{equation*}
  \mathcal M_n^{-1}\pa_t[V]=\begin{cases}
                             \bigl[V(\mathrm{i}P)_nV\bigr],  & \text{when}~n~\text{is odd},\\
                             \bigl[V^\dag(\mathrm{i}P)_nV\bigr], & \text{when}~n~\text{is even}.
                           \end{cases}
\end{equation*}

We now establish integrability for some examples from the non-commutative nonlinear Schr\"odinger hierarchy.
We express $\mathcal M_n^{-1}\pa_t[V]$ in the skew-P\"oppe algebra as follows.
\begin{definition}[Time-derivation endomorphism]\label{def:timederiv}
Given $n\in\mathbb Z$, we define the \emph{time-derivation endomorphism}
$\mathfrak e_n\colon\mathbb C[\mathbb Z_{\bs 0}]\to\mathbb C[\mathbb Z_{\bs 0}]$ by,
\begin{equation*}
\mathfrak e_n\colon[\bs 0]\mapsto\begin{cases}
                                      [\bs 0n\bs 0],  & \text{when}~n~\text{is odd},\\
                                      [\bs 0n\bs 0^\dag], & \text{when}~n~\text{is even}.
                                    \end{cases}
\end{equation*}
\end{definition}
The nonlinear fields we seek are expressed in the skew-P\"oppe algebra as follows.
\begin{definition}[P\"oppe polynomials]\label{def:Poppepoly}
  For $n\in\mathbb N\cup\{0\}$, let $\pi_n=\pi_n\bigl([\bs 0],[\bs 1],\ldots,[\bs n]\bigr)$ denote a
  polynomial consisting of a linear combination of odd-degree monomials of signature expansions
  in the skew-P\"oppe algebra $\mathbb C[\mathbb Z_{\bs 0}]$ of the form,
  \begin{equation*}
    \pi_n\coloneqq\sum_{k=1(\text{odd})}^{n}\sum_{a_1a_2\cdots a_k} c_{a_1a_2\cdots a_k}\cdot [\bs a_1]\,[\bs a_2]\,\cdots\,[\bs a_k].
  \end{equation*}
  The first sum is over odd values of $k$. The second sum is over all words $a_1a_2\cdots a_k$ we can construct from
  the alphabet $\{0,1,2,\ldots,n\}$ such that $a_1+a_2+\cdots+a_k=n-(k-1)$. This ensures $\pi_n$ is an odd polynomial.
  The coefficients $c_{a_1a_2\cdots a_k}$ are scalar constants.
\end{definition}
Our goal is to show $\mathfrak e_n\bigl([\bs 0]\bigr)$ can be expressed in terms of a P\"oppe polynomial
$\pi_n=\pi_n\bigl([\bs 0],[\bs 1],\ldots,[\bs n]\bigr)$ in the skew-P\"oppe algebra $\mathbb C[\mathbb Z_{\bs 0}]$.
Thus for each $n\in\mathbb N\cup\{0\}$, our goal is to determine the coefficients $c_{a_1a_2\cdots a_k}$ such that,
\begin{equation*}
\mathfrak e_n=\pi_n.
\end{equation*}
The examples we explore here correspond to the simple cases $n=0,1,2,3,4$, as follows.
\begin{example}[Linear ordinary differential equation: $n=0$]\label{ex:ODE}
We observe $\mathfrak e_0\bigl([\bs 0]\bigr)=[\bs 0 0\bs 0^\dag]$.
Recall that $a^\dag=-a$ for letters from $\mathbb Z$ in $(\mathbb Z_{\bs 0})^\ast$,
including $a=0$. Hence we observe,
\begin{equation*}
[\bs 0 0\bs 0^\dag]=\bs 0 0\bs 0^\dag-\bs 0^\dag0^\dag\bs 0=\bs 0 0\bs 0^\dag+\bs 0^\dag0\bs 0=\bs 0-\bs 0^\dag=[\bs 0].
\end{equation*}
In other words $\mathfrak e_0\bigl([\bs 0]\bigr)=[\bs 0]$ which translates to the following linear ordinary differential
equation for $g=\lb G\rb$, with $\mathcal M_0=\mu_0\mathrm{i}\mathcal I$,
\begin{equation*}
\pa_tg=\mathcal M_0\,g.
\end{equation*}
\end{example}
\begin{example}[Linear wave equation: $n=1$]
  We observe $\mathfrak e_1\bigl([\bs 0]\bigr)=[\bs 01\bs 0]$.
  From Definition~\ref{def:signatureexpansion}, we have the signature expansion $[\bs 1]=[\bs 01\bs 0]$,
  since $\chi(1)=1$. From Definition~\ref{def:derivation} for the derivation endomorphism, we know $\mfd[\bs 0]=[\bs 1]$. 
  Hence we have, $\mathfrak e_1\bigl([\bs 0]\bigr)=\mfd[\bs 0]$ in $\mathbb C[\mathbb Z_{\bs 0}]$.
  This translates to the linear wave equation for $g=\lb G\rb$, with $\mathcal M_1=-\mu_1\,\id$,
  \begin{equation*}
    \pa_tg=\mathcal M_1\,\pa g.
  \end{equation*}
\end{example}
\begin{example}[Nonlinear Schr\"odinger equation: $n=2$]\label{ex:secondorder}
  We observe $\mathfrak e_2\bigl([\bs 0]\bigr)=[\bs 02\bs 0^\dag]$.
  Using the homomorphic properties of $\chi$, the values for the signature coefficients given in Definition~\ref{def:signaturecharacter}
  and that each tensor product `$\ob$' under $\chi$ generates a real factor of $2$, we have $\chi(0\ot0\ob0)=2$ and $\chi(0\ob0\ob0)=4$.
  Then from Example~\ref{ex:tensornotation}, we see that we have,
  $[\bs 0]^{3}=2\cdot\bigl[\bs 0\{2\bs 0\}\bigr]+4\cdot[\bs 01\bs 0 1\bs 0]=2\cdot\bigl[\bs 02\bs 0\bigr]-2\cdot[\bs 02\bs 0^\dag]+4\cdot[\bs 01\bs 0 1\bs 0]$,
  using that $\bigl[\bs 0\{2\bs 0\}\bigr]=[\bs 02\bs 0]-[\bs 02\bs 0^\dag]$,
  and that $2^\dag=-2$. The signature expansion for $[\bs 2]=\mfd^2[\bs 0]$ is given by,
  $[\bs 2]=\chi(2)\cdot[\bs 02\bs 0]+\chi(11)\cdot[\bs 01\bs 0 1\bs 0]=[\bs 02\bs 0]+2\cdot[\bs 01\bs 0 1\bs 0]$.
  Hence we observe, $\mathfrak e_2\bigl([\bs 0]\bigr)=[\bs 2]-\frac12\cdot[\bs 0]^{3}$ in $\mathbb C[\mathbb Z_{\bs 0}]$.
  This translates to the non-commutative nonlinear Schr\"odinger equation 
  for $g=\lb G\rb$, with $\mathcal M_2=-\mu_2\,(\mathrm{i}\mathcal I)$,
  \begin{equation*}
    \mathcal M_2^{-1}\pa_tg=\pa^2g-\tfrac12g^3.
  \end{equation*}
\end{example}
\begin{remark}[Rescaling]\label{rmk:rescaling}
  In all examples, rescaling the solution $g$ to $2\,g$ recovers the usual corresponding
  equations in the non-commutative nonlinear Schr\"odinger hierarchy. This is because we assumed 
  $G\coloneqq V-V^\dag$ rather than $G=V(\mathrm{i}P)V^\dag\equiv\frac12(V-V^\dag)$.
\end{remark}
\begin{example}[Modified Korteweg--de Vries equation: $n=3$]\label{ex:KdVexpansion}
  We observe $\mathfrak e_3\bigl([\bs 0]\bigr)=[\bs 03\bs 0]$. Recall from Example~\ref{ex:tensornotation} that,
  $[\bs 0]^{2}=\chi([0\ot 0])\cdot\{\bs 0 1\bs 0\}$.
  Note this lies in $\mathbb C\langle\mathbb Z_{\bs 0}\rangle$ as opposed to the skew-P\"oppe algebra $\mathbb C[\mathbb Z_{\bs 0}]$.
  From the results of Example~\ref{ex:tensornotation2}, evaluating the signature characteristers, we know
  $[\bs 1]\,[\bs 0]^{2}=2\cdot\bigl(\bigl[\bs 02\{\bs01\bs0\}\bigr]+\bigl[\bs01\bs0\{2\bs0\}\bigr]\bigr)+4\cdot[\bs01\bs01\bs01\bs0]$
  and $[\bs 0]^{2}\,[\bs 1]=2\cdot\bigl(\bigl[\bs 02[\bs01\bs0]\bigr]+\bigl[\bs01\bs0[2\bs0]\bigr]\bigr)+4\cdot[\bs01\bs01\bs01\bs0]$.
  Then using the properties of the skew form from Definition~\ref{def:skewsymmparts}, we see that,
  \begin{equation*}
    \bigl[\bs 02\{\bs01\bs0\}\bigr]+\bigl[\bs 02[\bs01\bs0]\bigr]=2\cdot\bigl[\bs 02\bs01\bs0\bigr]
    ~~\text{and}~~
    \bigl[\bs 01\bs0\{2\bs0\}\bigr]+\bigl[\bs 01\bs0[2\bs0]\bigr]=2\cdot\bigl[\bs01\bs02\bs0\bigr].    
  \end{equation*}
  The signature expansion for $[\bs 3]=\mfd^3[\bs 0]$ is given by,
  \begin{align*}
    [\bs 3]&=\chi(3)\cdot[\bs 03\bs 0]+\chi(21)\cdot[\bs02\bs01\bs0]+\chi(12)\cdot[\bs01\bs02\bs0]
    +\chi(111)\cdot[\bs 01\bs 01\bs 0 1\bs 0]\\
    &=[\bs 03\bs 0]+3\cdot[\bs02\bs01\bs0]+3\cdot[\bs01\bs02\bs0]+6\cdot[\bs 01\bs 01\bs 0 1\bs 0].
  \end{align*}
  Hence we see that, $\mathfrak e_3\bigl([\bs 0]\bigr)=[\bs 3]-\tfrac34\cdot\bigl([\bs 1]\,[\bs 0]^{2}+[\bs 0]^{2}\,[\bs 1]\bigr)$,
  in $\mathbb C[\mathbb Z_{\bs 0}]$. This translates to the non-commutative modified Korteweg--de Vries equation 
  for $g=\lb G\rb$, with $\mathcal M_3=\mu_3\,\id$,
  \begin{equation*}
    \mathcal M_3^{-1}\pa_tg=\pa^3g-\tfrac34\bigl((\pa g)g^2+g^2(\pa g)\bigr).
  \end{equation*}
\end{example}

\begin{table}
\caption{Non-zero signature coefficients appearing
in the expansion of the \emph{P\"oppe polynomial} $\pi_3$ in Example~\ref{ex:KdVexpansion}.
The coefficients are the $\chi$-images of the signature entries shown.
Each column shows the factor contributions to the real coefficients of the basis elements
shown in the very left column, for each of the monomials in $\pi_3$  
shown across the top row. The final column represents the coefficient on the right-hand side of the equation 
$\pi_3=[\bs0 3\bs0]$.}
\label{table:NLS3}
\begin{center}
\begin{tabular}{|l|cccc|c|}
\hline
$\phantom{\biggl|}$ basis &$[\bs 3]$&$[\bs0]\,[\bs1]\,[\bs0]$ &$[\bs1]\,[\bs0]^{2}$&$[\bs0]^2\,[\bs1]$&$B$\\
\hline
$\phantom{\Bigl|}$ $[\bs03\bs0]$     &$3$& $2\cdot(0\ot1\ot0)$  &&& $1$\\
$\phantom{\Bigl|}$ $[\bs03\bs0^\dag]$ &   & $-2\cdot(0\ot1\ot0)$ &&&\\
\hline
$\phantom{\Bigl|}$ $[\bs02\bs01\bs0]$          & $21$ & $2\cdot(0\ot1\ot0)$ & $2\cdot(1\ot0\ot0)$  & $2\cdot(0\ot0\ot1)$ &\\
$\phantom{\Bigl|}$ $[\bs02\bs0^\dag1\bs0^\dag]$ &      & $2\cdot(0\ot1\ot0)$ & $-2\cdot(1\ot0\ot0)$  & $2\cdot(0\ot0\ot1)$ &\\
\hline
$\phantom{\Bigl|}$ $[\bs01\bs02\bs0]$     &$12$  & $2\cdot(0\ot1\ot0)$ & $2\cdot(1\ot0\ot0)$  & $2\cdot(0\ot0\ot1)$ &\\
$\phantom{\Bigl|}$ $[\bs01\bs02\bs0^\dag]$ &      & $-2\cdot(0\ot1\ot0)$ & $-2\cdot(1\ot0\ot0)$  & $2\cdot(0\ot0\ot1)$ &\\
\hline
$\phantom{\Bigl|}$ $[\bs01\bs01\bs01\bs0]$ & $111$ & $4\cdot(0\ot1\ot0)$ & $4\cdot(1\ot0\ot0)$  & $4\cdot(0\ot0\ot1)$ &\\
\hline
\end{tabular}
\end{center}
\end{table}

\begin{example}[Fourth order quintic nonlinear Schr\"odinger equation: $n=4$]\label{ex:fourthorder}
In this case we know $\mathfrak e_4\bigl([\bs 0]\bigr)=[\bs 04\bs 0^\dag]$.
For this and higher orders, our procedure needs to be systematic. The P\"oppe polynomial in this
case has the form,
\begin{align*}
\pi_4\coloneqq &\;c_4\cdot[\bs 4]+c_{200}\cdot[\bs2]\,[\bs0]^{2}+c_{020}[\bs0]\,[\bs2]\,[\bs0]
+c_{002}\cdot[\bs0]^{2}\,[\bs2]\\
&\;+c_{110}\cdot[\bs1]^{2}\,[\bs0]+c_{101}\cdot[\bs1]\,[\bs0]\,[\bs1]
+c_{011}\cdot[\bs0]\,[\bs1]^{2}+c_{00000}\cdot[\bs0]^{5}.
\end{align*}
The signature expansion for $[\bs4]$ has the form,
\begin{align*}
[\bs4]=&\;\chi(4)\cdot[\bs04\bs0]+\chi(31)\cdot[\bs03\bs01\bs0]+\chi(22)\cdot[\bs02\bs02\bs0]+\chi(13)\cdot[\bs01\bs03\bs0]\\
&\;+\chi(211)\cdot[\bs02\bs01\bs01\bs0]+\chi(121)\cdot[\bs01\bs02\bs01\bs0]+\chi(112)\cdot[\bs01\bs01\bs02\bs0]\\
&\;+\chi(1111)\cdot[\bs01\bs01\bs01\bs01\bs0].
\end{align*}
Using the skew and symmetric P\"oppe products in Lemma~\ref{lemma:skewandsymmPoppeproducts} we find, for example, that,
\begin{align*}
  [\bs 2]\,[\bs 0]^{2}=&\;\bigl(\chi(2)\cdot[\bs 02\bs0]+\chi(11)\cdot[\bs01\bs01\bs0]\bigr)\,\bigl(\chi(0\ob0)\cdot\{\bs 0 1\bs 0\}\bigr)\\
  =&\;\chi(2\ot0\ob0)\cdot\bigl(\bigl[\bs 03\{\bs01\bs0\}\bigr]+\bigl[\bs02\bs0\{2\bs0\}\bigr]\bigr)
  +\chi(2\ob0\ob0)\cdot[\bs02\bs01\bs01\bs0]\\
  &\;+\chi(11\ot0\ob0)\cdot\bigl(\bigl[\bs 01\bs02\{\bs01\bs0\}\bigr]+\bigl[\bs01\bs01\bs0\{2\bs0\}\bigr]\bigr)\\
  &\;+\chi(11\ob0\ob0)\cdot[\bs01\bs01\bs01\bs01\bs0].    
\end{align*}
The other products shown in $\pi_4$ can be similarly expanded. In Table~\ref{table:NLS4} we list all the basis elements
and corresponding coefficients generated by all the P\"oppe products present in $\pi_4$. The values of the coefficients 
are the $\chi$-images of the tensored terms shown. Each row generates a linear algebraic
equation for the expansion coefficients $c_4$, $c_{200}$, $c_{110}$, \ldots $c_{00000}$. Note that in Table~\ref{table:NLS4}
rows are ordered according to descent order, with a sub-order for the positions of the $\bs0^\dag$ letters
as indicated. The ordering of the columns is self-evident from the structure present in the table. We discuss
this ordering in more explicitly in Section~\ref{sec:hierarchycoding}.
Using all the rows shown, we generate an over-determined linear system of algebraic equations, $AC=B$,
where $B$ is the column vector shown in the right-hand column in Table~\ref{table:NLS4}, $C$ is the vector of coefficients $c_4$,
$c_{020}$, and so forth in the order shown across the top row. The matrix $A$ consists of the $\chi$-images of the entries shown in the table
(neglecting the final column).
From the augmented matrix $[A\, B]$, we observe the first two rows generate a closed system of equations,
namely $c_4+2c_{020}=0$ and $-2c_{020}=1$. This system of equations corresponds to the following smaller 
augmented matrix subsystem $[A_0\,B_0]$ for $c_4$ and $c_{020}$, where
\begin{equation*}
  A_0=\begin{pmatrix} 1 & 2\\ 0 & -2 \end{pmatrix}\qquad\text{and}\qquad B_0=\begin{pmatrix} 0\\1 \end{pmatrix}.
\end{equation*}
Hence we deduce $c_4=1$ and $c_{020}=-\frac12$. With these values in hand, we then observe that the next two
rows also generate a closed system of equations for $c_{200}$ and $c_{011}$ given by $4c_4+2c_{020}+2c_{200}+2c_{011}=0$ and $2c_{020}-2c_{200}+2c_{011}=0$.
This linear system of equations for the two unknowns $c_{200}$ and $c_{011}$ is represented by the smaller augmented matrix subsystem $[A_1\,B_1]$, where,
\begin{equation*}
  A_1=\begin{pmatrix} 2 & 2\\ -2 & 2 \end{pmatrix}\qquad\text{and}\qquad B_1=\begin{pmatrix} -3\\1 \end{pmatrix}.
\end{equation*}
Solving the linear system of equations, we deduce that $c_{200}=-1$ and $c_{011}=-1/2$. The next four rows in the augmented matrix $[A\, B]$,
given the coefficients we have already solved for, generate a closed system of equations for $c_{110}$, $c_{101}$, $c_{002}$ and $c_{00000}$,
generated by the smaller augment matrix $[A_2\, B_2]$, where,
\begin{equation*}
A_2=\begin{pmatrix}
1 & 1 & 2 & 4 \\ 
-1& 1 & 0 & -4\\ 
1 &-1 & 0 & -4\\ 
-1&-1 & 2 & 4 
\end{pmatrix}
\qquad\text{and}\qquad B_2=\begin{pmatrix}
-5/2\\ 
-5/2\\ 
-1/2\\ 
-3/2 
\end{pmatrix}.
\end{equation*}
The solution to this system is given by $c_{110}=-1/2$, $c_{101}=-3/2$, $c_{002}=-1$ and $c_{00000}=3/8$.
It is easy to check the equations represented by the remaining rows in the big augmented matrix $[A\, B]$
above, are consistent. Thus, we deduce that $\mathfrak e_4\bigl([\bs 0]\bigr)=\pi_4$, where the coefficients $c_4$, $c_{020}$ and so forth,
are given by the unique values outlined above.
The fourth order non-commutative nonlinear Schr\"odinger equation for $g=\lb G\rb$, with $\mathcal M_4=\mu_4\mathrm{i}\mathcal I$, is given by,
 \begin{align*}
   \mathcal M_4^{-1}\pa_tg=&\;\pa^4g-(\pa^2g)g^2-\tfrac12g(\pa^2g)g-g^2(\pa^2g)\\
     &\;-\tfrac12(\pa g)^2g-\tfrac32(\pa g)g(\pa g)-\tfrac12g(\pa g)^2+\tfrac38g^5. 
 \end{align*}
This matches the form given in Malham~\cite{Malham:quinticNLS} and Nijhoff \textit{et al.} \cite[eq.~B.4a]{NQVDLCI}.
\end{example}
\begin{remark}[Basis elements]\label{rmk:basiselements}
  In Tables~\ref{table:NLS3} and \ref{table:NLS4} we record the basis elements of the form
  $[w\times\bs\varphi]$ in the far left column. The composition components $w$ are compositions of $n$,
  i.e. compositions of $3$ and $4$, in the respective tables. The $\mathbb Z_2^\ast$-component $\bs\varphi$ in the basis element
  is in principle any possible $(|w|+1)$-tuples that can be constructed from
  $\{\bs0,\bs0^\dag\}\cong\mathbb Z_2$. However, recall $\bigl[(w\times\bs\varphi)^\dag\bigr]=-[w\times\bs\varphi]$
  and the Definition~\ref{rmk:convention}. 
  Using this property for $\bigl[(w\times\bs\varphi)^\dag\bigr]$, for any basis element $[w\times\bs\varphi]$, we can thus always
  arrange for the first component of $\bs\varphi$ to be `$\bs0$'; as can be observed in the tables. 
  The basis elements are ordered according to descent order for the compositions $w$ and natural binary ordering for the
  $\mathbb Z_2^\ast$-components $\bs\varphi$. For more details see Definition~\ref{def:naturalordering}
  in Section~\ref{sec:hierarchycoding}, and the subsequent discussion therein.
  Note that though the first component of $\bs\varphi$ can always be arranged to be $\bs0$,
  in our computations involving P\"oppe products, we often utilise the symmetry $\bigl[(w\times\bs\varphi)^\dag\bigr]=-[w\times\bs\varphi]$
  in order to use the P\"oppe products listed in Lemma~\ref{lemma:skewandsymmPoppeproducts}.
  Thus temporarily, the first component in some factors is sometimes $\bs0^\dag$.
  However, we always use the same symmetry again to convert the final answer to the form with $\bs0$ as the
  first component in the basis element.
\end{remark}

\begin{landscape}
\begin{table}
\caption{Non-zero signature coefficients appearing
in the expansion of the \emph{P\"oppe polynomial} $\pi_4$ in Example~\ref{ex:fourthorder}.
The coefficients are the $\chi$-images of the signature entries shown.
Each column shows the factor contributions to the real coefficients of the basis elements
shown in the very left column, for each of the monomials in $\pi_4$ shown
across the top row. The final column represents the coefficient on the right-hand side of the equation $\pi_4=[\bs0 4\bs0^\dag]$.}
\label{table:NLS4}
\begin{center}
\begin{tabular}{|l|cccccccc|c|}
\hline
$\phantom{\biggl|}$ basis &$[\bs 4]$&$[\bs0]\,[\bs2]\,[\bs0]$ &$[\bs2]\,[\bs0]^{2}$&$[\bs0]\,[\bs1]^{2}$&$[\bs1]^{2}\,[\bs0]$&$[\bs1]\,[\bs0]\,[\bs1]$
&$[\bs0]^{2}\,[\bs2]$&$[\bs0]^{5}$&$B$\\
\hline
$\phantom{\Bigl|}$ $[\bs04\bs0]$     &$4$& $2\cdot(0\ot2\ot0)$  &&&&&&&\\
$\phantom{\Bigl|}$ $[\bs04\bs0^\dag]$ &   & $-2\cdot(0\ot2\ot0)$ &&&&&&&$1$\\
\hline
$\phantom{\Bigl|}$ $[\bs03\bs01\bs0]$          & $31$ & $0\ot2\ob0$ & $2\cdot(2\ot0\ot0)$  & $2\cdot(0\ot1\ot1)$ &&&&&\\
$\phantom{\Bigl|}$ $[\bs03\bs0^\dag1\bs0^\dag]$ &      & $0\ot2\ob0$ & $-2\cdot(2\ot0\ot0)$ & $2\cdot(0\ot1\ot1)$ &&&&&\\
\hline
$\phantom{\Bigl|}$ $[\bs02\bs02\bs0]$     &$22$& $0\ot11\ot0$ &$2\cdot(2\ot0\ot0)$ & $0\ot1\ot1$ & $1\ot1\ot0$ & $1\ot0\ot1$ & $2\cdot(0\ot0\ot2)$& $4\cdot(0\ot0\ot0\ot0\ot0)$ &\\
$\phantom{\Bigl|}$ $[\bs02\bs02\bs0^\dag]$ &  & $-0\ot11\ot0$& $-2\cdot(2\ot0\ot0)$ & $0\ot1\ot1$ & $-1\ot1\ot0$ & $1\ot0\ot1$ & & $-4\cdot(0\ot0\ot0\ot0\ot0)$ &\\
$\phantom{\Bigl|}$ $[\bs02\bs0^\dag2\bs0]$ &  & $-0\ot11\ot0$& & $0\ot1\ot1$ & $1\ot1\ot0$ & $-1\ot0\ot1$ & & $-4\cdot(0\ot0\ot0\ot0\ot0)$ &\\
$\phantom{\Bigl|}$ $[\bs02\bs0^\dag2\bs0^\dag]$ &  & $0\ot11\ot0$& & $0\ot1\ot1$ & $-1\ot1\ot0$ & $-1\ot0\ot1$  & $2\cdot(0\ot0\ot2)$& $4\cdot(0\ot0\ot0\ot0\ot0)$ &\\
\hline
$\phantom{\Bigl|}$ $[\bs01\bs03\bs0]$     &$13$  & $0\ob2\ot0$ & & & $2\cdot(1\ot1\ot0)$  &  & $2\cdot(0\ot0\ot2)$ &&\\
$\phantom{\Bigl|}$ $[\bs01\bs03\bs0^\dag]$ &      & $-0\ob2\ot0$ & &  & $-2\cdot(1\ot1\ot0)$ &   & $2\cdot(0\ot0\ot2)$ &&\\
\hline
$\phantom{\Bigl|}$ $[\bs02\bs01\bs01\bs0]$ &$211$& $0\ot11\ob0$& $2\cdot(2\ob0\ot0)$& $0\ot1\ob1$ & $1\ot1\ob0$ & $1\ot0\ob1$ & $2\cdot(0\ot0\ot11)$& $8\cdot(0\ot0\ot0\ot0\ot0)$ &\\
$\phantom{\Bigl|}$ $[\bs02\bs0^\dag1\bs0^\dag1\bs0^\dag]$ & & $-0\ot11\ob0$ & & $-0\ot1\ob1$ & $1\ot1\ob0$ & $1\ot0\ob1$ & $-2\cdot(0\ot0\ot11)$& $-8(0\ot0\ot0\ot0\ot0)$ &\\
$\phantom{\Bigl|}$ $[\bs01\bs02\bs01\bs0]$ &$121$& $0\ob2\ob0$ & $2\cdot(11\ot0\ot0)$ & $0\ob1\ot1$ &  $1\ot1\ob0$ & $1\ot0\ob1$ & $2\cdot(0\ot0\ot11)$& $8\cdot(0\ot0\ot0\ot0\ot0)$ &\\
$\phantom{\Bigl|}$ $[\bs01\bs02\bs0^\dag1\bs0^\dag]$ & & & $-2\cdot(11\ot0\ot0)$ & $0\ob1\ot1$ &  $1\ot1\ob0$ & $1\ot0\ob1$  & $-2\cdot(0\ot0\ot11)$& $-8\cdot(0\ot0\ot0\ot0\ot0)$ &\\
$\phantom{\Bigl|}$ $[\bs01\bs01\bs02\bs0]$ & $112$ & $0\ob11\ot0$ & $2\cdot(11\ot0\ot0)$ & $0\ob1\ot1$ &  $1\ob1\ot0$ & $1\ob0\ot1$  & $2\cdot(0\ot0\ob2)$& $8\cdot(0\ot0\ot0\ot0\ot0)$ &\\
$\phantom{\Bigl|}$ $[\bs01\bs01\bs02\bs0^\dag]$ & & $-0\ob11\ot0$ & $-2\cdot(11\ot0\ot0)$ & $0\ob1\ot1$ &  $-1\ob1\ot0$ & $1\ob0\ot1$ & & $-8\cdot(0\ot0\ot0\ot0\ot0)$ &\\
$\phantom{\Bigl|}$ $[\bs01\bs01\bs01\bs01\bs0]$ & $1111$ & $0\ob11\ob0$ & $2\cdot(11\ob0\ot0)$ & $0\ob1\ob1$ & $1\ob1\ob0$ & $1\ob0\ob1$  & $2\cdot(0\ot0\ob11)$& $16\cdot(0\ot0\ot0\ot0\ot0)$ &\\
\hline
\end{tabular}
\end{center}
\end{table}
\end{landscape}

\section{Non-commutative Lax hierarchy and the sine-Gordon equation}\label{sec:Laxhierarchy}
Herein we establish the non-commutative nonlinear Schr\"odinger and modified Korteweg--de Vries Lax hierarchy iteratively,
order by order. The non-commutative modified Korteweg--de Vries hierarchy can be found for example in Carillo and Schiebold~\cite[eq.~(9)]{CSI}.
Importantly, this iterative hierarchy extends to all negative orders. The first member of negative order, i.e.\/ for which $n=-1$, corresponds to
the non-commutative sine-Gordon cubic-form equation in Example~\ref{ex:SGform}; see, for example, Tracy and Widom~\cite{TW} for the scalar case.
Establishing the hierarchy for all orders $n\in\mathbb Z$ is particularly simple in the P\"oppe algebra $\mathbb C\la\mathbb Z_{\bs0}\ra$.
We need to define some natural actions on $\mathbb C\la\mathbb Z_{\bs0}\ra$ first.
\begin{definition}[Adjoint and symmetric algebra products and actions]
  We define the standard commutation and symmetric products, respectively,
  $\mathrm{ad}\colon\mathbb C\la\mathbb Z_{\bs0}\ra\times\mathbb C\la\mathbb Z_{\bs0}\ra\to\mathbb C\la\mathbb Z_{\bs0}\ra$ and
  $\mathrm{sd}\colon\mathbb C\la\mathbb Z_{\bs0}\ra\times\mathbb C\la\mathbb Z_{\bs0}\ra\to\mathbb C\la\mathbb Z_{\bs0}\ra$.
  For example, for $[\bs0]\in\mathbb C[\mathbb Z_{\bs0}]\subset\mathbb C\la\mathbb Z_{\bs0}\ra$
  and any word $w\times\bs\varphi\in\mathbb C\la\mathbb Z_{\bs0}\ra$ we have,
    \begin{align*}
      \mathrm{ad}_{[\bs0]}(w\times\bs\varphi)&\coloneqq[\bs0]\,(w\times\bs\varphi)-(w\times\bs\varphi)\,[\bs0], \\
      \mathrm{sd}_{[\bs0]}(w\times\bs\varphi)&\coloneqq[\bs0]\,(w\times\bs\varphi)+(w\times\bs\varphi)\,[\bs0], 
    \end{align*}
    which is the exclusive form of their action we use below.
    We also define the following two actions on the skew-P\"oppe algebra $\mathbb C[\mathbb Z_{\bs0}]$.
    For $[\bs0]\in\mathbb C[\mathbb Z_{\bs0}]$ set,
    \begin{align*}
    A\coloneqq\frac14\mathrm{ad}_{[\bs0]}\mfd^{-1}\mathrm{ad}_{[\bs0]},\\
    S\coloneqq\frac14\mathrm{sd}_{[\bs0]}\mfd^{-1}\mathrm{sd}_{[\bs0]}.
    \end{align*}
\end{definition}
That the actions of $A$ and $S$ are closed in $\mathbb C[\mathbb Z_{\bs0}]$ is established as part of
the proof of the following crucial lemma.
\begin{lemma}[Natural iteration]\label{lemma:naturaliteration}
  For any $n\in\mathbb Z$ we have:
  \begin{align*}
    (\mfd-A)[\bs0n\bs0^\dag]&=[\bs0(n+1)\bs0],\\
    (\mfd-S)[\bs0n\bs0]&=[\bs0(n+1)\bs0^\dag].
  \end{align*}
\end{lemma}
\begin{proof}
By direct computation using Lemma~\ref{lemma:skewandsymmPoppeproducts} we have,
 \begin{align*}
   \mathrm{ad}_{[\bs0]}\,[\bs0n\bs0^\dag]&=[\bs0][\bs0n\bs0^\dag]-[\bs0n\bs0^\dag][\bs0]\\
   &=2\cdot\bigl(\{\bs0(n+1)\bs0^\dag\}+\{\bs01\bs0n\bs0^\dag\}+\{\bs0n\bs0^\dag1^\dag\bs0^\dag\}\bigr)\\
   &=2\cdot\mfd\,\{\bs0n\bs0^\dag\}.  
  \end{align*}
Hence we observe, using that $\{\bs0 n\bs0^\dag\}=-\{\bs0^\dag n\bs0\}$ we have,
\begin{align*}
  A\,[\bs0n\bs0^\dag]
  &=\tfrac12\bigl([\bs0]\,\{\bs0n\bs0^\dag\}+\{\bs0^\dag n\bs0\}\,[\bs0]\bigr)\\
  &=[\bs0(n+1)\bs0^\dag]-[\bs0(n+1)\bs0]+[\bs01\bs0n\bs0^\dag]+[\bs0^\dag n\bs01\bs0]\\
  &=\mfd\,[\bs0n\bs0^\dag]-[\bs0(n+1)\bs0].
\end{align*}
This gives the first result. Again, by direct computation, we have,
\begin{align*}
  \mathrm{sd}_{[\bs0]}\,[\bs0n\bs0]
  =&\;[\bs0]\,[\bs0n\bs0]+[\bs0n\bs0]\,[\bs0]\\
  =&\;2\,\bigl(\{\bs0(n+1)\bs0\}+\{\bs01\bs0n\bs0\}+\{\bs0n\bs01\bs0\}\bigr)\\
  =&\;2\cdot\mfd\,\{\bs0n\bs0\}.
\end{align*}
Hence we observe, 
\begin{align*}
  S\,[\bs0n\bs0]
  =&\;\tfrac12\bigl([\bs0]\,\{\bs0n\bs0\}+\{\bs0n\bs0\}\,[\bs0]\bigr)\\
  =&\;[\bs0(n+1)\bs0]+[\bs01\bs0n\bs0]+[\bs0n\bs01\bs0]-[\bs0(n+1)\bs0^\dag]\\
  =&\;\mfd\,[\bs0n\bs0]-[\bs0(n+1)\bs0^\dag].
\end{align*}
This gives the second result.\qed
\end{proof}
The following immediate corollary is established straightforwardly by induction, both when $n$ is positive as well as negative.
For the remainder of this section we refer to P\"oppe polyomials $\pi_n$ for $n\in\mathbb Z$, though in Definition~\ref{def:Poppepoly},
$n\in\mathbb N\cup\{0\}$ for which $\pi_n=\pi_n\bigl([\bs0],[\bs1],\ldots,[\bs n]\bigr)$. The form of $\pi_n$ for the
negative integer cases is given presently.
\begin{corollary}[Non-commutative Lax hierarchy iteration]\label{cor:Laxhierarchyiteration}
Let $n\in\mathbb Z$ be a given integer, and consider the equation, $\mathfrak e_n\bigl([\bs 0]\bigr)=\pi_n$,
where $\pi_n$ is a P\"oppe polynomial and $\mathfrak e_n\bigl([\bs 0]\bigr)$ equals $[\bs 0n\bs 0]$ if $n$ is odd, and equals $[\bs 0n\bs 0^\dag]$ if $n$ is even.
For $n=0,1,2,3,4$ such polynomials $\pi_n$ exist, as demonstrated in Examples~\ref{ex:ODE}---\ref{ex:fourthorder}. Then we have,
\begin{equation*}
  \mathfrak e_{n+1}\bigl([\bs 0]\bigr)=\begin{cases} (\mfd-A)\pi_n, & \text{when $n$ is even,}\\
                                                    (\mfd-S)\pi_n, & \text{when $n$ is odd.}
  \end{cases}
\end{equation*}
\end{corollary}
\begin{proof}
This follows directly from Lemma~\ref{lemma:naturaliteration} for $n\in\mathbb Z$ by induction.
\end{proof} 
Further, we have the following additional immediate corollary.
\begin{corollary}[Non-commutative Lax hierarchy]\label{cor:Laxhierarchy}
  For any $n\in\mathbb Z$, the $(n+1)$th order member equation of the non-commutative Lax hierarchy is given by,
  \begin{equation*}
    \mathfrak e_{n+1}\bigl([\bs 0]\bigr)=\begin{cases} (\mfd-A)\,\bigl((\mfd-S)(\mfd-A)\bigr)^{\frac{n}{2}}\,[\bs0], & \text{when $n$ is even,}\\
                                                      \bigl((\mfd-S)(\mfd-A)\bigr)^{\frac12(n+1)}\,[\bs0], & \text{when $n$ is odd.}
    \end{cases}
    \end{equation*}
\end{corollary}
The Lax hierarchy stated in Corollary~\ref{cor:Laxhierarchy}, at each odd order $n$, exactly matches that quoted 
for the non-commutative modified Korteweg--de Vries hierarchy in Carillo and Schiebold~\cite[eq.~(9)]{CSI}.
Given the main existence and uniqueness result we prove in Section~\ref{sec:hierarchycoding}, this is expected.
The cases $n=0,1,2,3,4$ in Corollary~\ref{cor:Laxhierarchy} naturally match the non-commutative equation
members given in Examples~\ref{ex:ODE}---\ref{ex:fourthorder}. However, Corollary~\ref{cor:Laxhierarchy} also
applies for negative $n$. Consider the example case of order `$-1$'. 
\begin{example}[Non-commutative sine-Gordon cubic-form equation: order `$-1$']\label{ex:minusonecase}
Setting $n=-2$, a case when $n$ is even, in Corollary~\ref{cor:Laxhierarchy}, generates the equation,
\begin{equation*}
  \mathfrak e_{-1}\bigl([\bs 0]\bigr)=(\mfd-A)\bigl((\mfd-S)(\mfd-A)\bigr)^{-1}\,[\bs0]\quad\Leftrightarrow\quad
  (\mfd-S)\,\mathfrak e_{-1}\bigl([\bs 0]\bigr)=[\bs0].
\end{equation*}
We can express the equation on the right as follows,
\begin{equation*}
    \mfd\mathfrak e_{-1}\bigl([\bs 0]\bigr)
    =[\bs0]+\tfrac14\Bigl(\bigl(\mfd^{-1}\mathfrak e_{-1}([\bs 0]^{2})\bigr)\,[\bs0]
    +[\bs 0]\,\bigl(\mfd^{-1}\mathfrak e_{-1}([\bs 0]^{2})\bigr)\Bigr),
\end{equation*}
where we have used that $\mathrm{sd}_{[\bs0]}\,\mathfrak e_{-1}([\bs 0])=\mathfrak e_{-1}([\bs 0]^2)$ from Lemma~\ref{lemma:timederivsq} just below.
This relation in $\mathbb C[\mathbb Z_{\bs0}]$ translates to the non-commutative sine-Gordon cubic-form equation 
given in Example~\ref{ex:SGform} for $g=\lb G\rb$, with $\mathcal M_{-1}=\mu_{-1}\,\id$, i.e.\/
\begin{equation*}
    \pa\pa_tg-\tfrac14\bigl((\pa^{-1}\pa_tg^2)\,g+g\,(\pa^{-1}\pa_tg^2)\bigr)=\mu_{-1}g.
\end{equation*}
Invoking the factor `2' rescaling mentioned in Remark~\ref{rmk:rescaling} gives the exact match.
\end{example}
\begin{lemma}\label{lemma:timederivsq}
  For any $n\in\mathbb Z$ we have,
  \begin{equation*}
    \mathfrak e_{n}\bigl([\bs 0]^2\bigr)=\begin{cases}
              \mathrm{sd}_{[\bs0]}\,\mathfrak e_{n}\bigl([\bs 0]\bigr), & \text{when $n$ is odd},\\
              \mathrm{ad}_{[\bs0]}\,\mathfrak e_{n}\bigl([\bs 0]\bigr), & \text{when $n$ is even}.
    \end{cases}
  \end{equation*}
\end{lemma}
\begin{proof}
  This result is established at the operator level. Recall $\mathcal M_{n}\coloneqq-\mu_{n}(\mathrm{i}\mathcal I)^{n-1}$
  and the properties outlined in Lemma~\ref{lemma:blockidentities}. We observe when $n$ is even, we have,
  \begin{align*}
  \pa_t[V]^2=&\;\mathcal M_{n}\bigl[V(\mathrm{i}P)_{n}V^\dag\bigr]\,[V]+[V]\,\mathcal M_{n}\bigl[V(\mathrm{i}P)_{n}V^\dag\bigr]\\
  =&\;\mathcal M_{n}\bigl[V(\mathrm{i}P)_{n}V^\dag\bigr]\,[V]-\mathcal M_{n}[V]\,\bigl[V(\mathrm{i}P)_{n}V^\dag\bigr],
  \end{align*}
  where we have used that  $[V]\,\mathcal M_{n}=\lb V-V^\dag\rb\,\mathcal M_{n}=\mathcal M_{n}\,\lb V^\dag-V\rb=-\mathcal M_{n}[V]$. 
  When $n$ is odd, we follow an analogous computation with $[V(\mathrm{i}P)_{n}V]$ replacing $[V(\mathrm{i}P)_{n}V^\dag]$ just above,
  and that in this case there is no sign change in the second term on the right as 
  $[V]\,\mathcal M_{n}=\lb V-V^\dag\rb\,\mathcal M_{n}=\mathcal M_{n}\,\lb V-V^\dag\rb=\mathcal M_{n}[V]$.
\qed
\end{proof}
Naturally we can continue to consider further negative order hierarchy member equations. For example,  
the non-commutative order `$-2$' equation can be generated by setting $n=-3$, a case when $n$ is odd,
in Corollary~\ref{cor:Laxhierarchy}. With $\mathfrak e_{-2}\bigl([\bs 0]\bigr)=[\bs 0(-2)\bs 0^\dag]$,
this generates the non-commutative equation,
\begin{equation*}
  \mathfrak e_{-2}\bigl([\bs 0]\bigr)=\bigl((\mfd-S)(\mfd-A)\bigr)^{-1}\,[\bs0]\quad\Leftrightarrow\quad
  (\mfd-S)(\mfd-A)\,\mathfrak e_{-2}\bigl([\bs 0]\bigr)=[\bs0].
\end{equation*}
And so forth.
The example integrable equations in Examples~\ref{ex:ODE}--\ref{ex:fourthorder} in Section~\ref{sec:ncNLS}
were of the form $\pa_tg=\pi(g,\pa g,\pa^2g,\ldots)$ for $n=0,1,2,3,4$. Indeed for these examples, we show 
the equation is unique in this class. In other words, at each of the orders considered, given the base dispersion equation
for $P$ and the form of the Marchenko equation, the right-hand side in the non-commutative nonlinear partial differential equation,
the `nonlinear field', is of the form $\pi=\pi(g,\pa g,\pa^2g,\ldots)$, where $\pi$ is a polynomial in its arguments.
In Section~\ref{sec:hierarchycoding}, we establish this fact to all orders $n\geqslant0$.

\section{Hierarchy uniqueness}\label{sec:hierarchycoding}
Herein we prove that at each order $n\geqslant0$, the P\"oppe polynomial signature
expansion $\pi_n$ such that $\mathfrak e_n([\bs 0])=\pi_n$, exists, and is unique. This is our main result.
Before presenting this result in the general case, we present one further example, the $n=5$ case.
This case acts a useful reference for our general argument.
\begin{example}[Fifth order quintic modified Korteweg--de Vries equation: $n=5$]\label{ex:fifthorder}
In this case $\mathfrak e_5\bigl([\bs 0]\bigr)=[\bs 05\bs 0]$ and 
the P\"oppe polynomial $\pi_5$, in general, has the form,
\begin{align*}
\pi_5\coloneqq &\;c_5\cdot[\bs 5]+c_{300}\cdot[\bs3]\,[\bs0]^{2}+c_{030}\cdot[\bs0]\,[\bs3]\,[\bs0]
+c_{003}\cdot[\bs0]^{2}\,[\bs3]+c_{210}\cdot[\bs 2]\,[\bs1]\,[\bs0]\\
&\;+c_{201}\cdot[\bs 2]\,[\bs0]\,[\bs1]+c_{120}\cdot[\bs 1]\,[\bs2]\,[\bs0]+c_{102}\cdot[\bs 1]\,[\bs0]\,[\bs2]+c_{021}\cdot[\bs 0]\,[\bs2]\,[\bs1]\\
&\;+c_{012}\cdot[\bs 0]\,[\bs1]\,[\bs2]+c_{111}\cdot[\bs 1]\,[\bs1]\,[\bs1]
+c_{10000}\cdot[\bs1]\,[\bs0]^{4}+c_{01000}\cdot[\bs0]\,[\bs1]\,[\bs0]^{3}\\
&\;+c_{00100}\cdot[\bs0]^{2}\,[\bs1]\,[\bs0]^{2}+c_{00010}\cdot[\bs0]^{3}\,[\bs1]\,[\bs0]+c_{00001}\cdot[\bs1]\,[\bs0]^{4}.
\end{align*}
The signature expansion for $[\bs5]$ has the form,
\begin{align*}
  [\bs5]=&\;\chi(5)\cdot[\bs05\bs0]\\
  &\;+\chi(41)\cdot[\bs04\bs01\bs0]+\chi(32)\cdot[\bs03\bs02\bs0]+\chi(23)\cdot[\bs02\bs03\bs0]+\chi(14)\cdot[\bs01\bs04\bs0]\\
  &\;+\chi(311)\cdot[\bs03\bs01\bs01\bs0]+\chi(221)\cdot[\bs02\bs02\bs01\bs0]+\chi(212)\cdot[\bs02\bs01\bs02\bs0]\\
  &\;+\chi(131)\cdot[\bs01\bs03\bs01\bs0]+\chi(122)\cdot[\bs01\bs02\bs02\bs0]+\chi(113)\cdot[\bs01\bs01\bs03\bs0]\\
  &\;+\chi(2111)\cdot[\bs02\bs01\bs01\bs01\bs0]+\chi(1211)\cdot[\bs01\bs02\bs01\bs01\bs0]+\chi(1121)\cdot[\bs01\bs01\bs02\bs01\bs0]\\
  &\;+\chi(1112)\cdot[\bs01\bs01\bs01\bs02\bs0]+\chi(11111)\cdot[\bs01\bs01\bs01\bs01\bs01\bs0].
\end{align*}
Using the skew and symmetric P\"oppe products in Lemma~\ref{lemma:skewandsymmPoppeproducts} we find, for example, that,
\begin{align*}
  [\bs1]&\,[\bs2]\,[\bs0]\\
  =&\;\bigl(\chi(1)\cdot[\bs 01\bs0]\bigr)\,\bigl(\chi(2)\cdot[\bs02\bs0]+\chi(11)\cdot[\bs01\bs01\bs0]\bigr)
  \,\bigl(\chi(0)\cdot[\bs0]\bigr)\\
  =&\;\bigl(\chi(1)\cdot[\bs 01\bs0]\bigr)\,\bigl(\chi(2\ot0)\cdot\bigl\{\bs03[\bs0]\bigr\}+\chi(2\ob0)\cdot\{\bs02\bs01\bs0\}\\
   &\;\qquad\qquad\qquad\quad+\chi(11\ot0)\cdot\bigl\{\bs01\bs02[\bs0]\bigr\}+\chi(11\ob0)\cdot\{\bs01\bs01\bs01\bs0\}\bigr)\\  
  =&\;\chi(1\ot2\ot0)\cdot\bigl(\bigl[\bs 02\{\bs03[\bs0]\}\bigr]+\bigl[\bs01\bs0\{4[\bs0]\}\bigr]\bigr)
  +\chi(1\ob2\ot0)\cdot\bigl[\bs01\bs01\bs03[\bs0]\bigr]\\
  &\;+\chi(1\ot2\ob0)\cdot\bigl(\bigl[\bs 02\{\bs02\bs01\bs0\}\bigr]+\bigl[\bs01\bs0\{3\bs01\bs0\}\bigr]\bigr)
  +\chi(1\ob2\ob0)\cdot[\bs01\bs01\bs02\bs01\bs0]\\
&\;+\chi(1\ot11\ot0)\cdot\bigl(\bigl[\bs 02\{\bs01\bs02[\bs0]\}\bigr]+\bigl[\bs01\bs0\{1\bs02[\bs0]\}\bigr]\bigr)
  +\chi(1\ob11\ot0)\cdot\bigl[\bs01\bs01\bs01\bs02[\bs0]\bigr]\\
&\;+\chi(1\ot11\ob0)\cdot\bigl(\bigl[\bs 02\{\bs01\bs01\bs01\bs0\}\bigr]+\bigl[\bs01\bs0\{2\bs01\bs01\bs0\}\bigr]\bigr)\\
  &\;+\chi(1\ob11\ob0)\cdot\bigl[\bs01\bs01\bs01\bs01\bs01\bs0\bigr].
\end{align*}
The other products shown in $\pi_5$ can be similarly expanded. In Tables~\ref{table:NLS5a} and \ref{table:NLS5b}
we list the essential basis elements and corresponding signature coefficients generated by all the P\"oppe products present in $\pi_5$.
The values of the coefficients are the $\chi$-images of the tensored terms shown. Each row generates a linear algebraic
equation for the expansion coefficients $c_5$, $c_{300}$, $c_{210}$, \ldots $c_{00001}$.
The ordering of rows and columns in Tables~\ref{table:NLS5a} and \ref{table:NLS5b} is self-evident.
We discuss this ordering in detail presently.
Using all the rows shown, we generate an over-determined linear system of algebraic equations, $AC=B$,
where $B$ is the column vector shown in the right-hand column in Table~\ref{table:NLS5b} and $C$ is the vector of coefficients $c_5$,
$c_{030}$, $c_{300}$, $c_{021}$ and so forth in the order shown. The matrix $A$ is populated with the $\chi$-images of the signature coefficients
shown in the tables. We can in fact solve $AC=B$ for $C$ systematically, block by block, as follows.
The first two rows corresponding to $[\bs05\bs0]$ and $\bs0^\dag5\bs0$ generate the pair of equations,
$c_5+2c_{030}=1$ and $c_{030}=0$. This system of equations corresponds to the smaller augmented subsystem $[A_0\,B_0^\prime]$
where $A_0$ is the same matrix as in Example~\ref{ex:fourthorder} for the order $n=4$ case, and $B_0^\prime=(1,0)^{\mathrm{T}}$.
We deduce $c_5=1$ and $c_{030}=0$. The next two rows corresponding to $[\bs04\bs01\bs0]$ and $[\bs04\bs0^\dag1\bs0^\dag]$ generate
the pair of equations $5c_5+2c_{030}+2c_{300}+2c_{021}=0$ and $2c_{030}-2c_{300}+2c_{021}=0$ for $c_{300}$ and $c_{021}$.
This pair corresponds to the smaller augmented subsystem $[A_1\,B_1^\prime]$ where $A_1$ is the same matrix as in Example~\ref{ex:fourthorder},
and $B_1^\prime=(-5,0)^{\mathrm{T}}$. Hence we deduce $c_{021}=-5/4$ and $c_{300}=-5/4$. Then the equations corresponding to
the rows $[\bs03\bs02\bs0]$, $[\bs03\bs02\bs0^\dag]$, $[\bs03\bs0^\dag2\bs0]$ and $[\bs03\bs0^\dag2\bs0^\dag]$
generate the following smaller augmented matrix subsystem $[A_2\, B_2^\prime]$ for $c_{210}$, $c_{201}$, $c_{012}$ and $c_{01000}$, where,
$A_2$ is the same as the corresponding matrix in Example~\ref{ex:fourthorder} and $B_2^\prime=(-25/4,-5/4,5/4,5/4)^{\mathrm{T}}$.
This system of linear equations is easily solved to reveal $c_{210}=-5/4$, $c_{201}=-5/2$, $c_{012}=-5/4$ and $c_{01000}=0$.
As we shall see, it is no coincidence that the coefficent matrices $A_0$, $A_1$ and $A_2$ match those for the
analogous blocks of basis elements in Example~\ref{ex:fourthorder}.
The next set of rows corresponding to the basis elements 
$[\bs02\bs03\bs0]$, $[\bs02\bs03\bs0^\dag]$, $[\bs02\bs0^\dag3\bs0]$ and $[\bs02\bs0^\dag3\bs0^\dag]$,
generates exactly the same augmented matrix subsystem $[A_2\, B_2^\prime]$ as that discussed just above,
but now for the unknown coefficients $c_{120}$, $c_{102}$, $c_{003}$ and $c_{00010}$. 
This linear system reveals $c_{120}=-5/4$, $c_{102}=-5/2$, $c_{003}=-5/4$ and $c_{00010}=0$.
We do not deduce any new information from the equations corresponding to the rows $[\bs01\bs04\bs0]$ and $[\bs01\bs04\bs0^\dag]$,
other than that they are consistent.
We consider the next block of four rows shown in Tables~\ref{table:NLS5a} and \ref{table:NLS5b} corresponding to
the rows $[\bs02\bs02\bs01\bs0]$, $[\bs02\bs02\bs0^\dag1\bs0^\dag]$, $[\bs02\bs0^\dag2\bs01\bs0]$ and $[\bs02\bs0^\dag2\bs0^\dag1\bs0^\dag]$.
These basis elements generate the following augmented matrix subsystem $[A_3^\prime\, B_3]$ for $c_{111}$, $c_{10000}$, $c_{00100}$ and $c_{00001}$, where,
\begin{equation*}
 A_3^\prime=\begin{pmatrix}
    1 & 4 & 4 & 4  \\
    1 & -4& -4& 4  \\
    1 & -4& 4 & -4 \\
    1 & 4 & -4& -4 
 \end{pmatrix}
\qquad\text{and}\qquad 
B_3=\begin{pmatrix}
    5 \\
    -5 \\
    -5 \\
    -5 
 \end{pmatrix}.
\end{equation*}
This system of linear equations is easily solved to reveal $c_{111}=-5/2$, $c_{10000}=5/8$, $c_{00002}=5/8$ and $c_{00001}=5/8$. 
We have determined the unique set of coefficients for which $\mathfrak e_5\bigl([\bs 0]\bigr)=\pi_5$.
In principle we can check the equations generated by the rows corresponding to the remaining basis elements are consistent.
However, we explain in our proof of our main result (Step~8) why this is not necessary.
Hence the fifth order non-commutative modified Korteweg--de Vries for $g=\lb G\rb$, with $\mathcal M_5=\mu_5\,\id$,
is given by (this matches the form given in Nijhoff \textit{et al.} \cite[eq.~B.5a]{NQVDLCI}),
\begin{align*}
  \mathcal M_5^{-1}\pa_tg=&\;\pa^5g-\frac54\Bigl(\bigl(\pa^3g\bigr)g^2+g^2\bigl(\pa^3g\bigr)
  +\bigl(\pa^2g\bigr)\bigl(\pa g\bigr)g\\
  &\;+2\,\bigl(\pa^2g\bigr)g\bigl(\pa g\bigr)+\bigl(\pa g\bigr)\bigl(\pa^2g\bigr)g+2\,\bigl(\pa g\bigr)g\bigl(\pa^2g\bigr)\\
  &\;+g\bigl(\pa^2g\bigr)\bigl(\pa g\bigr)+g\bigl(\pa g\bigr)\bigl(\pa^2g\bigr)
    +2\,\bigl(\pa g\bigr)\bigl(\pa g\bigr)\bigl(\pa g\bigr)\Bigr)\\
    &\;+\frac58\Bigl(\bigl(\pa g\bigr)g^4+g^2\bigl(\pa g\bigr)g^2+g^4\bigl(\pa g\bigr)\Bigr).
\end{align*}
\end{example}

\begin{landscape}
\begin{table}
\caption{Non-zero signature coefficients appearing
  in the expansion of the \emph{P\"oppe polynomial} $\pi_5$ in Example~\ref{ex:fifthorder}.
  Not all the coefficients are shown. The remaining columns are shown in Table~\ref{table:NLS5b}.
  The coefficients are the $\chi$-images of the signature entries shown.
  Each column shows the factor contributions to the real coefficients of the basis elements
  shown in the very left column, for each of the monomials in $\pi_5$ shown
  across the top row. }
\label{table:NLS5a}
\begin{center}
\begin{tabular}{|l|cccccccc|}
\hline
$\phantom{\biggl|}$ basis &$[\bs 5]$ &$[\bs0]\,[\bs3]\,[\bs0]$&$[\bs3]\,[\bs0]^{2}$&$[\bs0]\,[\bs2]\,[\bs1]$&$[\bs2]\,[\bs1]\,[\bs0]$&$[\bs2]\,[\bs0]\,[\bs1]$
&$[\bs 0]\,[\bs1]\,[\bs2]$ &$[\bs0]\,[\bs1]\,[\bs0]^{3}$ \\
\hline
$\phantom{\Bigl|}$ $[\bs05\bs0]$     & $5$ & $2\cdot(0\ot3\ot0)$ & & & & & &  \\
$\phantom{\Bigl|}$ $[\bs05\bs0^\dag]$ &    & $-2\cdot(0\ot3\ot0)$ & & & & & &  \\
\hline
$\phantom{\Bigl|}$ $[\bs04\bs01\bs0]$     & $41$ & $0\ot3\ob0$  & $3\ot0\ob0$ & $2\cdot(0\ot2\ot1)$ & & & &  \\
$\phantom{\Bigl|}$ $[\bs04\bs0^\dag1\bs0^\dag]$ &      & $0\ot3\ob0$ & $-3\ot0\ob0$ & $2\cdot(0\ot2\ot1)$ & & & &  \\
\hline
$\phantom{\Bigl|}$ $[\bs03\bs02\bs0]$     & $32$  & $0\ot21\ot0$ & $3\ot0\ob0$  & $0\ot2\ot1$ & $2\ot1\ot0$ & $2\ot0\ot1$ & $2\cdot(0\ot1\ot2)$ & $2\cdot(0\ot1\ot0\ob0\ot0)$  \\
$\phantom{\Bigl|}$ $[\bs03\bs02\bs0^\dag]$ &      & $-0\ot21\ot0$ & $-3\ot0\ob0$ & $0\ot2\ot1$ & $-2\ot1\ot0$ & $2\ot0\ot1$ &                    & $-2\cdot(0\ot1\ot0\ob0\ot0)$ \\
$\phantom{\Bigl|}$ $[\bs03\bs0^\dag2\bs0]$ &      & $-0\ot21\ot0$ &              & $0\ot2\ot1$ & $2\ot1\ot0$ & $-2\ot0\ot1$ &                    & $-2\cdot(0\ot1\ot0\ob0\ot0)$ \\
$\phantom{\Bigl|}$ $[\bs03\bs0^\dag2\bs0^\dag]$ &  & $0\ot21\ot0$ &              & $0\ot2\ot1$& $-2\ot1\ot0$ & $-2\ot0\ot1$ & $2\cdot(0\ot1\ot2)$& $2\cdot(0\ot1\ot0\ob0\ot0)$ \\
\hline
$\phantom{\Bigl|}$ $[\bs02\bs03\bs0]$          & $23$ & $0\ot12\ot0$ & & & $2\cdot(2\ot1\ot0)$  & & $0\ot1\ot2$ &   \\
$\phantom{\Bigl|}$ $[\bs02\bs03\bs0^\dag]$     &      & $-0\ot12\ot0$ & & & $-2\cdot(2\ot1\ot0)$ & & $0\ot1\ot2$ &  \\
$\phantom{\Bigl|}$ $[\bs02\bs0^\dag3\bs0]$     &      & $-0\ot12\ot0$ & & &                     & & $0\ot1\ot2$ &  \\
$\phantom{\Bigl|}$ $[\bs02\bs0^\dag3\bs0^\dag]$ &      & $0\ot12\ot0$ & & &                     & & $0\ot1\ot2$ &  \\
\hline
$\phantom{\Bigl|}$ $[\bs01\bs04\bs0]$     & $14$ & $0\ob3\ot0$ & & & & & &  \\
$\phantom{\Bigl|}$ $[\bs01\bs04\bs0^\dag]$ &      & $-0\ob3\ot0$ & & & & & & \\
\hline
$\phantom{\Bigl|}$ $[\bs02\bs02\bs01\bs0]$             & $221$ & $0\ot12\ob0$ & $21\ot0\ob0$  &  $0\ot11\ot1$ & $2\ot1\ob0$ & $2\ot0\ob1$ & $0\ot1\ot11$ & $0\ot1\ot0\ob0\ob0$  \\
$\phantom{\Bigl|}$ $[\bs02\bs02\bs0^\dag1\bs0^\dag]$     &   &         & $-21\ot0\ob0$ & $0\ot11\ot1$  & $2\ot1\ob0$ & $2\ot0\ob1$ & $-0\ot1\ot11$ & $-0\ot1\ot0\ob0\ob0$ \\
$\phantom{\Bigl|}$ $[\bs02\bs0^\dag2\bs01\bs0]$         &   &              &               & $-0\ot11\ot1$ &             &            & $0\ot1\ot11$  & $0\ot1\ot0\ob0\ob0$  \\
$\phantom{\Bigl|}$ $[\bs02\bs0^\dag2\bs0^\dag1\bs0^\dag]$&   & $-0\ot12\ob0$&               & $-0\ot11\ot1$ &             &            & $-0\ot1\ot11$ & $-0\ot1\ot0\ob0\ob0$  \\
\hline
$\phantom{\Bigl|}$$\qquad~\vdots$ & & & & & & & & \\
\hline
\end{tabular}
\end{center}
\end{table}
\end{landscape}

\begin{landscape}
\begin{table}
\caption{The remaining non-zero signature coefficients appearing
  in the expansion of the \emph{P\"oppe polynomial} $\pi_5$ in Example~\ref{ex:fifthorder}.
  The first set of columns appear in Table~\ref{table:NLS5a}.
  The final column represents the coefficient on the right-hand side of the equation $\pi_5=[\bs05\bs0]$.}
\label{table:NLS5b}
\begin{center}
\begin{tabular}{|l|cccccccc|c|} 
\hline
$\phantom{\biggl|}$ basis & $[\bs1]\,[\bs2]\,[\bs0]$&$[\bs1]\,[\bs0]\,[\bs2]$& $[\bs0]^{2}\,[\bs3]$&$[\bs0]^{3}\,[\bs1]\,[\bs0]$&$[\bs1]^{3}$&$[\bs1]\,[\bs0]^{4}$
& $[\bs0]^{2}\,[\bs1]\,[\bs0]^{2}$&$[\bs0]^{4}\,[\bs1]$& $B$ \\
\hline
$\phantom{\Bigl|}$ $[\bs05\bs0]$      & & & & & & & & & $1$ \\
$\phantom{\Bigl|}$ $[\bs05\bs0^\dag]$   & & & & & & & & & \\
\hline
$\phantom{\Bigl|}$ $[\bs04\bs01\bs0]$       & & & & & & & & & \\
$\phantom{\Bigl|}$ $[\bs04\bs0^\dag1\bs0^\dag]$   & & & & & & & & & \\
\hline
$\phantom{\Bigl|}$ $[\bs03\bs02\bs0]$          & & & & & & & & & \\
$\phantom{\Bigl|}$ $[\bs03\bs02\bs0^\dag]$      & & & & & & & & &  \\
$\phantom{\Bigl|}$ $[\bs03\bs0^\dag2\bs0]$      & & & & & & & & &  \\
$\phantom{\Bigl|}$ $[\bs03\bs0^\dag2\bs0^\dag]$ & & & & & & & & & \\
\hline
$\phantom{\Bigl|}$ $[\bs02\bs03\bs0]$           & $1\ot2\ot0$ & $1\ot0\ot2$  & $0\ob0\ot3$  & $2\cdot(0\ob0\ot0\ot1\ot0)$ & & & & & \\
$\phantom{\Bigl|}$ $[\bs02\bs03\bs0^\dag]$      & $-1\ot2\ot0$ & $1\ot0\ot2$  &              & $-2\cdot(0\ob0\ot0\ot1\ot0)$ & & & & & \\
$\phantom{\Bigl|}$ $[\bs02\bs0^\dag3\bs0]$      & $1\ot2\ot0$  & $-1\ot0\ot2$ &              & $-2\cdot(0\ob0\ot0\ot1\ot0)$ & & & & & \\
$\phantom{\Bigl|}$ $[\bs02\bs0^\dag3\bs0^\dag]$  & $-1\ot2\ot0$  & $-1\ot0\ot2$  & $0\ob0\ot3$ & $2\cdot(0\ob0\ot0\ot1\ot0)$ & & & & & \\
\hline
$\phantom{\Bigl|}$ $[\bs01\bs04\bs0]$      & $2\cdot(1\ot2\ot0)$  & & $0\ob0\ot3$ & & & & & &\\
$\phantom{\Bigl|}$ $[\bs01\bs04\bs0^\dag]$  & $-2\cdot(1\ot2\ot0)$ & & $0\ob0\ot3$ & & & & & &\\
\hline
$\phantom{\Bigl|}$ $[\bs02\bs02\bs01\bs0]$              & $1\ot2\ob0$ & $1\ot0\ot11$ & $0\ob0\ot21$ &  $0\ob0\ot0\ot1\ob0$  & $1\ot1\ot1$ & $1\ot0\ob0\ot0\ob0$ & $0\ob0\ot1\ot0\ob0$  & $0\ob0\ot0\ob0\ot1$ & \\
$\phantom{\Bigl|}$ $[\bs02\bs02\bs0^\dag1\bs0^\dag]$      &            & $-1\ot0\ot11$ &              &  $0\ob0\ot0\ot1\ob0$ & $1\ot1\ot1$ & $-1\ot0\ob0\ot0\ob0$ & $-0\ob0\ot1\ot0\ob0$ & $0\ob0\ot0\ob0\ot1$ & \\
$\phantom{\Bigl|}$ $[\bs02\bs0^\dag2\bs01\bs0]$          &             & $-1\ot0\ot11$ &              &  $-0\ob0\ot0\ot1\ob0$ & $1\ot1\ot1$ & $-1\ot0\ob0\ot0\ob0$ & $0\ob0\ot1\ot0\ob0$ & $-0\ob0\ot0\ob0\ot1$ & \\
$\phantom{\Bigl|}$  $[\bs02\bs0^\dag2\bs0^\dag1\bs0^\dag]$& $1\ot2\ob0$ & $1\ot0\ot11$ & $-0\ob0\ot21$ &  $-0\ob0\ot0\ot1\ob0$ & $1\ot1\ot1$ & $1\ot0\ob0\ot0\ob0$ & $-0\ob0\ot1\ot0\ob0$ & $-0\ob0\ot0\ob0\ot1$ & \\
\hline
$\phantom{\Bigl|}$ $\qquad~\vdots$ &&&&&&&&&\\
\hline
\end{tabular}
\end{center}
\end{table}
\end{landscape}

We now consider the general order $n\geqslant0$ case. Our goal is to establish the following. 
\begin{theorem}[Main result: existence and uniqueness]\label{thm:main}
For every $n\in\mathbb N\cup\{0\}$, there exists a unique P\"oppe polynomial $\pi_n=\pi_n([\bs0],[\bs1],\ldots,[\bs n])$
in $\mathbb C[\mathbb Z_{\bs0}]$ such that $\mathfrak e_n([\bs 0])=\pi_n$.
\end{theorem}
We prove this result through a sequence of steps. It requires some preparation and the 
rest of this section is devoted to outlining the notation, strategy, ideas
and intermediate results we require to carry through the proof in Steps~1--7 before giving the
overall proof in Step~8. We have the following immediate corollary of Theorem~\ref{thm:main}.
\begin{corollary}[The non-commutative Lax hierarchy is unique]\label{cor:Laxhierarchyuniqueness}
  For any integer $n\geqslant0$, the Lax hierarchy generated by the iteration indicated in
  Corollary~\ref{cor:Laxhierarchyiteration} is unique. The non-commutative nonlinear equation at order $n$ in the Lax hierarchy
  is simply that represented by $\mathfrak e_n([\bs 0])=\pi_n$ in Theorem~\ref{thm:main}.
\end{corollary}
\begin{proof}
When $n$ is even, we know from Corollary~\ref{cor:Laxhierarchyiteration}
that $\mathfrak e_{n+1}([\bs 0])=(\pa-A)\pi_n$. Now from Theorem~\ref{thm:main} we know
$\mathfrak e_{n+1}([\bs 0])=\pi_{n+1}$ with $\pi_{n+1}=\pi_{n+1}([\bs0],[\bs1],\ldots,[\bs n+1])$
a unique polynomial in its arguments. We thus deduce the result of the corollary when $n$ is even.
An exactly analogous argument follows for the case when $n$ is odd.\qed
\end{proof}
Let us now return to the proof of Theorem~\ref{thm:main}.
Roughly, in the proof of Theorem~\ref{thm:main}, we construct a `table' for the coefficients of $\pi_n$,
much like we did for the cases $n=3,4,5$ in Tables~\ref{table:NLS3}---\ref{table:NLS5b}.
Indeed, Tables~\ref{table:NLS5a} and \ref{table:NLS5b} act as a useful reference.
We proceed systematically, considering basis elements $[w\times\bs\varphi]$ in descending order with respect to the composition
$\mathcal C$-component $w$ and using a natural binary order for the $\mathbb Z_2^\ast$-component $\bs\varphi$.
Both these orders are implicit in Tables~\ref{table:NLS3}---\ref{table:NLS5b}.
\medskip

\emph{Step~1: Descent and binary order.}
We introduce an order, the descent order, on the set of basis elements. We also define a new respresentation for the
basis elements, the composition-binary representation, that we use hereafter.
The descent order for compositions is given in Malham~\cite{Malham:KdVhierarchy}; we present it here for completeness. 
\begin{definition}[Descent ordering of compositions]\label{def:naturalordering}
A composition $u\in\Cb$ precedes another composition $v\in\Cb$ if the length of the compostion $u$, i.e.\/ the number of digits it contains,
is strictly less than the length of $v$. If $u$ and $v$ have the same length, say $k$, so $u=u_1u_2\cdots u_k$ and $v=v_1v_2\cdots v_k$,
then $u$ precedes $v$ if for some $\ell\in\{1,2,\dots,k\}$ we have $u_1=v_1$, $u_2=v_2$, \ldots, $u_{\ell-1}=v_{\ell-1}$ and $u_\ell<v_\ell$. 
Otherwise, $v$ precedes $u$. The resulting ordering induced on $\Cb$, is the \emph{descent ordering}.
\end{definition}
The $\mathbb Z_2^\ast$-component $\bs\varphi$ in the basis element $[w\times\bs\varphi]$ is any $(|w|+1)$-tuple constructed from
$\{\bs0,\bs0^\dag\}\cong\mathbb Z_2$. Recall from Remark~\ref{rmk:basiselements},
we can always arrange for the first component of $\bs\varphi$ to be `$\bs0$'; see Tables~\ref{table:NLS3}---\ref{table:NLS5b}.
Thus for any given composition $w\in\mathcal C$, the component $\bs\varphi\in\mathbb Z_2^\ast$ in a basis element $[w\times\bs\varphi]$
is one of the following forms,
\begin{align*}
  &\bs\phi_1\coloneqq\bs0\bs0\cdots\bs0\bs0\bs0\bs0,~
  \bs\phi_2\coloneqq\bs0\bs0\cdots\bs0\bs0\bs0\bs0^\dag,~
  \bs\phi_3\coloneqq\bs0\bs0\cdots\bs0\bs0\bs0^\dag\bs0,~
  \bs\phi_4\coloneqq\bs0\bs0\cdots\bs0\bs0\bs0^\dag\bs0^\dag,\\
  &\bs\phi_5\coloneqq\bs0\bs0\cdots\bs0\bs0^\dag\bs0\bs0,~
  \bs\phi_6\coloneqq\bs0\bs0\cdots\bs0\bs0^\dag\bs0\bs0^\dag,~
  \bs\phi_7\coloneqq\bs0\bs0\cdots\bs0\bs0^\dag\bs0^\dag\bs0,
\end{align*}
and so forth, all the way up to $\bs\phi_{2^{|w|}}\coloneqq\bs0\bs0^\dag\cdots\bs0^\dag\bs0^\dag\bs0^\dag\bs0^\dag$.
This is the \emph{natural binary ordering} of $\mathbb Z_2^\ast$ we refer to just above.
In this and the next sections, we mainly use a modified encoding of the basis elements $[w\times\bs\varphi]$, as follows.
We replace the free monoid $\mathbb Z_2^\ast$ of all forms $\bs\varphi$ that can be constructed from $\{\bs0,\bs0^\dag\}\cong\mathbb Z_2$
by the vector space $\mathbb R\la\mathbb B\ra$ representing the span over all the elements $\mathbb B\coloneqq\{\bs\phi_i\}_{i\geqslant1}$.
We thus express any $\bs\varphi\in\mathbb Z_2^\ast$ in the form $\bs\varphi=\beta_1\bs\phi_1+\beta_2\bs\phi_2+\beta_3\bs\phi_3+\cdots$, where the $\beta_i$
for integer $i\geqslant1$, represent the coefficients of the basis element components $\bs\phi_i$. Henceforth we
represent any element $\bs\varphi\in\mathbb Z_2^\ast$ by a $2^{|w|}$-tuple $\bs\beta\coloneqq(\beta_1,\beta_2,\ldots,\beta_{2^{|w|}})\in\mathbb R\la\mathbb B\ra$.
Thus we replace,
\begin{equation*}
[w\times\bs\varphi]\rightsquigarrow[w]\times\bs\beta.
\end{equation*}
\begin{example}
   Some examples matching the old notation with the new are as follows: $[\bs0a\bs0]=[a]\times(1,0)$, $[\bs0a\bs0^\dag]=[a]\times(0,1)$, and then also,
  \begin{align*}
    [\bs0a_1\bs0a_2\bs0]&=[a_1a_2]\times(1,0,0,0),\\
    [\bs0a_1\bs0a_2\bs0^\dag]&=[a_1a_2]\times(0,1,0,0),\\
    [\bs0a_1\bs0^\dag a_2\bs0]&=[a_1a_2]\times(0,0,1,0),\\
    [\bs0a_1\bs0^\dag a_2\bs0^\dag]&=[a_1a_2]\times(0,0,0,1),\\
    [\bs0a_1\bs0a_2\bs0a_3\bs0]&=[a_1a_2a_3]\times(1,0,0,0,0,0,0,0),\\
    [\bs0a_1\bs0a_2\bs0a_3\bs0^\dag]&=[a_1a_2a_3]\times(0,1,0,0,0,0,0,0),\\
    \vdots &\\
    [\bs0a_1\bs0^\dag a_2\bs0^\dag a_3\bs0^\dag]&=[a_1a_2a_3]\times(0,0,0,0,0,0,0,1),
  \end{align*}
  and so forth. Naturally, as constructed, for linear combinations, we have for example that for any real scalar constants $\beta_1$ and $\beta_2$,
\begin{equation*}
  \beta_1\cdot[\bs0a_1\bs0a_2\bs0]+\beta_2\cdot[\bs0a_1\bs0^\dag a_2\bs0]=[a_1a_2]\times(\beta_1,0,\beta_2,0),
\end{equation*}
and so forth. There is a natural basis for elements of $\mathbb R\la\mathbb B\ra$ of a given length $2^n$. Such a basis is
given by the elements of length $2^n$ of the form $\bs\beta_i\coloneqq(0,\ldots,0,1,0,\ldots)$, where the `$1$' is in the $i$th position. 
\end{example}
\begin{definition}[Composition-binary representation]\label{def:composition-binary}
  We call the representation $[w]\times\bs\beta$ where $w\in\mathcal C$ and $\bs\beta\in\mathbb R\la\mathbb B\ra$
  the \emph{composition-binary representation} of the basis elements. We call $[w]$ the composition component and
  $\bs\beta$ the $\mathbb R\la\mathbb B\ra$-component.
\end{definition}
\medskip

\emph{Step~2: Triple product action.}
All the P\"oppe polynomials $\pi_n$ are polynomials in the skew-P\"oppe algebra $\mathbb C[\mathbb Z_{\bs 0}]$
and thus necessarily of odd degree. As such all the monomials therein can be constructed from triple products
of signature expansions $[\bs n]$, as we have seen in Examples~\ref{ex:secondorder}---\ref{ex:fifthorder}.
In particular we characterise the following triple product action which is established straightforwardly using
the P\"oppe product rules in Lemma~\ref{lemma:skewandsymmPoppeproducts}.
\begin{lemma}\label{lemma:tripleproductaction}
  Consider the two signature expansions $[\bs a]$ and $[\bs b]$ and a basis element $[cw\times\bs\varphi]\in\mathbb C[\mathbb Z_{\bs 0}]$,
  with $c\in\mathbb Z$ the first letter in the composition `$cw$'. Let $\hat{\bs\varphi}$ denote the element of $\mathbb Z_2^\ast$
  given by $\bs\varphi$ with its first letter `$\,\bs0$' removed.
  Then at leading order the triple product action $[\bs a]\,[\bs b]\,\bigl([cw\times\bs\varphi]\bigr)$ is given by,
  \begin{align*}
    \chi(a\ot b\ot cw)
    \cdot\Bigl(&[\bs 0(a+1)\bs0(b+1)\bs0 c(w\times\hat{\bs\varphi})]+[\bs 0(a+1)\bs0^\dag(b+1)\bs0 c(w\times\hat{\bs\varphi})]\\
    +&[\bs 0(a+1)\bs0(b+1)\bs0^\dag c(w\times\hat{\bs\varphi})^\dag]+[\bs 0(a+1)\bs0^\dag(b+1)\bs0^\dag c(w\times\hat{\bs\varphi})^\dag]\\
    +&[\bs 0(a+1)\bs0 b\bs0 (c+1)(w\times\hat{\bs\varphi})]+[\bs 0(a+1)\bs0^\dag b\bs0^\dag (c+1)(w\times\hat{\bs\varphi})]\\
    +&[\bs 0(a+1)\bs0 b\bs0 (c+1)(w\times\hat{\bs\varphi})^\dag]+[\bs 0(a+1)\bs0^\dag b\bs0^\dag (c+1)(w\times\hat{\bs\varphi})^\dag]\\
    +&2\cdot[\bs 0a\bs0(b+2)\bs0 c(w\times\hat{\bs\varphi})]+2\cdot[\bs 0a\bs0(b+2)\bs0^\dag c(w\times\hat{\bs\varphi})^\dag]\\
    +&[\bs 0a\bs0(b+1)\bs0(c+1)(w\times\hat{\bs\varphi})]+[\bs 0a\bs0(b+1)\bs0^\dag(c+1)(w\times\hat{\bs\varphi})]\\
    +&[\bs 0a\bs0(b+1)\bs0(c+1)(w\times\hat{\bs\varphi})^\dag]+[\bs 0a\bs0(b+1)\bs0^\dag(c+1)(w\times\hat{\bs\varphi})^\dag]\Bigr)+\cdots.
  \end{align*}
  Here by leading order, we mean, we do not retain terms generated by lower (descent) order terms
  in the signature expansions of $[\bs a]$ and $[\bs b]$, nor do we retain terms generated with a quasi-product term---i.e.\/
  generated using any of the final terms with real factor `$\,2$' in the P\"oppe products in Lemma~\ref{lemma:skewandsymmPoppeproducts}.
  We use the notation `$+\cdots$' do denote these missing terms.
\end{lemma}
We also define the following auto-tensorial action on $\mathbb R\la\mathbb B\ra$.
\begin{definition}[Tensorial action]
  Given $\bs\beta\coloneqq(\beta_1,\beta_2,\ldots)$ and $\bs\gamma\coloneqq(\gamma_1,\gamma_2,\ldots)$
  in $\mathbb R\la\mathbb B\ra$ we define the (left) \emph{auto-tensorial action} of $\bs\beta$ on $\bs\gamma$,
  denoted $\bs\beta\lhd\bs\gamma$, to be,
  \begin{equation*}
  \bs\beta\lhd\bs\gamma\coloneqq \bigl(\beta_1\cdot\bs\gamma, \beta_2\cdot\bs\gamma,\ldots\bigr),
  \end{equation*}
  where for each $i=1,2,\ldots$, we note $\beta_i\cdot\bs\gamma=(\beta_i\gamma_1,\beta_i\gamma_2,\ldots)$.
\end{definition}
In the new notation, with the tensorial action on $\mathbb R\la\mathbb B\ra$ just defined,
the triple product action on $\mathbb C[\mathbb Z_{\bs 0}]$ given in Lemma~\ref{lemma:tripleproductaction}
can be expressed more succinctly as follows.
\begin{corollary}[Triple product action] \label{cor:tripleproductaction}
  Given two signature expansions $[\bs a]$ and $[\bs b]$ and a generic basis element
  $[cw]\times\bs\beta\in\mathbb C[\mathcal C]\times\mathbb R\la\mathbb B\ra\cong\mathbb C[\mathbb Z_{\bs0}]$,
  the triple product action $[\bs a]\,[\bs b]\,\bigl([cw]\times\bs\beta\bigr)$ on $\mathbb C[\mathcal C]\times\mathbb R\la\mathbb B\ra$
  is given at leading order by,
  \begin{align*}
  [\bs a]\,[\bs b]\,\bigl([cw]\times\bs\beta\bigr)\phantom{uuuu}&\\
  =\chi(a\ot b\ot cw)\cdot\Bigl(&[(a+1)(b+1)cw]\times\bigl((1,0,1,0)\lhd\bs\beta+(-1)^{|w|}\,(0,1,0,1)\lhd\bs\beta^\dag\bigr)\\
    +&[(a+1)b(c+1)w]\times\bigl((1,0,0,1)\lhd(\bs\beta+(-1)^{|w|}\,\bs\beta^\dag)\bigr)\\    
    +&[a(b+2)cw]\times\bigl((2,0,0,0)\lhd\bs\beta+(-1)^{|w|}\,(0,2,0,0)\lhd\bs\beta^\dag\bigr)\\
    +&[a(b+1)(c+1)w]\times \bigl((1,1,0,0)\lhd(\bs\beta+(-1)^{|w|}\,\bs\beta^\dag)\bigr)
    \Bigr)+\cdots.
  \end{align*}
  Here, if $\bs\beta=(\beta_1,\beta_2,\beta_3,\ldots,\beta_{2^n})$, then $\bs\beta^\dag=(\beta_{2^n},\beta_{2^n-1},\ldots,\beta_2,\beta_1)$.
\end{corollary}
\begin{proof}
The result of the corollary is just a restatement of the triple product action in Lemma~\ref{lemma:tripleproductaction}.
That the adjoint of $\bs\beta$, denoted $\bs\beta^\dag$, corresponds to reversing the elements in $\bs\beta$ is explained as follows.
The entries in $\bs\beta$ correspond to the coefficients of the basis $\bs\phi_1$, $\bs\phi_2$, $\ldots$, $\hat{\bs\phi}_{2^n}$, for some $n\in\mathbb N$.
The triple product action in Lemma~\ref{lemma:tripleproductaction} involves the components $\hat{\bs\varphi}^\dag$ where
$\hat{\bs\varphi}$ corresponds to $\bs\varphi$ with its first letter `$\,\bs0$' removed. Let $\hat{\bs\phi}_1$, $\hat{\bs\phi}_2$, $\ldots$, $\hat{\bs\phi}_{2^n}$,
be the same sequence of basis elements each of which has the first letter `$\,\bs0$' removed. We observe,
$\bigl\{\hat{\bs\phi}_1^\dag,\hat{\bs\phi}_2^\dag,\ldots,\hat{\bs\phi}_{2^n}^\dag\bigr\}=\bigl\{\hat{\bs\phi}_{2^n},\hat{\bs\phi}_{2^n-1},\ldots,\hat{\bs\phi}_{1}\bigr\}$.
\qed
\end{proof}
\begin{example}\label{ex:generaltripleproduct}
  Consider computing the triple product $[\bs a]\,[\bs b]\,[\bs c]$ to leading order. In this case,
  to leading order $[\bs c]=[\bs0 c\bs 0]+\cdots$ and so in Corollary~\ref{cor:tripleproductaction} we have $w=\nu$, the empty word, with $|w|=0$. 
  We observe, $[\bs 0c\bs0]=[c]\times\bs\gamma$ with $\bs\gamma=(1,0)\in\mathbb R\la\mathbb B\ra$.
  Hence we have,
  \begin{align*}
    (1,0,1,0)\lhd\bs\gamma&=\bigl(1\cdot(1,0),\,0\cdot(1,0),\,1\cdot(1,0),\,0\cdot(1,0)\,\bigr)=(1,0,0,0,1,0,0,0),\\
    (0,1,0,1)\lhd\bs\gamma^\dag&=\bigl(0\cdot(0,1),\,1\cdot(0,1),\,0\cdot(0,1),\,1\cdot(0,1)\,\bigr)=(0,0,0,1,0,0,0,1),
  \end{align*}
  and so forth. Hence we observe from Corollary~\ref{cor:tripleproductaction} that at leading order,
  \begin{align*}
  [\bs a]\,[\bs b]\,[\bs c]
  =\chi(a\ot b\ot c)\cdot\bigl(&[(a+1)(b+1)c]\times(1,0,0,1,1,0,0,1)\\
    &\;+[(a+1)b(c+1)]\times(1,1,0,0,0,0,1,1)\\
    &\;+[a(b+2)c]\times(2,0,0,2,0,0,0,0)\\
    &\;+[a(b+1)(c+1)]\times(1,1,1,1,0,0,0,0)\bigr)+\cdots.
  \end{align*}
\end{example}
We extensively use such computations hereafter. 
\medskip

\emph{Step~3: Generators, a coarse-grain overview.}
For a given basis element $[w]\times\bs\beta$ and composition $w$ of $n\in\mathbb N$, it is
useful to identify the types of odd-degree monomials of signature expansions that might generate it.
\begin{definition}[Monomial generator]\label{def:generator}
  Given a basis element $[w]\times\bs\beta$ with a composition component $[w]$ and an $\mathbb R\la\mathbb B\ra$-component $\bs\beta$,
  where $w$ is a composition of $n\in\mathbb N$ and $\bs\beta$ has length $2^n$, we call any odd-degree
  monomial of signature expansions of the form $[\bs a_1]\,[\bs a_2]\,\cdots\,[\bs a_{2m+1}]$ that produces
  $[w]\times\bs\beta$ as one of the terms in its expansion, a \emph{monomial generator} or just \emph{generator} of $[w]\times\bs\beta$.
\end{definition}
At this stage and in this step, it is useful to give a brief coarse overview of our overall strategy,
which we implement in detail in the subsequent steps below. We show in this step how, for any given composition component,
we can identify, for the associated basis elements, specific collections of generators. We call the sets of
basis elements and corresponding collections of generators ``coefficient blocks'' or simply ``blocks''.
We show that such blocks are necessarily square.  
To start, consider any one-part composition $w$ of $n$, so the two corresponding basis elements
are $[\bs0 n\bs0]\rightsquigarrow[n]\times(1,0)$ and $[\bs0 n\bs0^\dag]\rightsquigarrow[n]\times(0,1)$.
The basis element $[n]\times(1,0)$ is generated by signature expansion $[\bs n]$, while
$[n]\times(0,1)$ is not. On the other hand, both basis elements are generated by $[\bs0]\,[\bs{n-2}]\,[\bs0]$.
This exhausts all the possible odd-degree monomials in $\pi_n$ that could generate
$[n]\times(1,0)$ and  $[n]\times(0,1)$. Thus for a one-part composition, the possible odd-degree generators
have the form,
\begin{equation*}
[\star],~[\bs0]\,[\star]\,([\bs0]),
\end{equation*}
where $[\star]$ represents the appropriate generic signature expansion, i.e.\/ in the first instance it is $[\bs n]$
and in the second instance, i.e.\/ for $[\bs0]\,[\star]\,([\bs0])$, the middle $[\star]$ factor is $[\bs{n-2}]$.
Note, we allow $[\star]=[\bs0]$.
The only other possibilities are $[\star]\,[\bs0]\,[\bs0]$ and $[\bs0]\,[\bs0]\,[\star]$.
However for $n\geqslant3$, we can rule these two possibilities out as $[\bs0]\,[\bs0]=2\cdot\{\bs01\bs0\}$ and any subsequent P\"oppe product
of this term with $[\bs0]$ would generate a basis element $[w]\times\bs\beta$ where the composition $w$ has two parts.
Note of course, the triple P\"oppe product $[\bs a]\,[\bs b]\,[\bs c]$ is naturally associative.

Now consider any two part composition $w=a_1a_2$ of $n$. We observe that basis elements with such a two-part  
composition component can in principle be generated by $[\star]$ and $[\bs0]\,[\star]\,[\bs0]$, which we have already come across just above.
However such basis elements can also be generated by any of the following four generators of the form,
\begin{equation*}
[\star]\,[\star]\,([\bs0]),~[\star]\,[\bs0]\,([\star]),~[\bs0]\,[\star]\,([\star]),~[\bs0]\,[\star]\,\bigl([\bs0]\,[\star]\,([\bs0])\bigr).
\end{equation*}
We observe that each possible generator above contains only two `$[\star]$' factors, consistent with the
two-part composition component of the basis elements we are aiming to generate. Further note that
we can also see that the four generators above can be constructed from the previous two generators
$[\star]$ and $[\bs0]\,[\star]\,([\bs0])$ corresponding to one-part compositions, by applying
one of the three actions $[\star]\,[\star]\,(\cdot)$, $[\star]\,[\bs0]\,(\cdot)$ or $[\bs0]\,[\star]\,(\cdot)$ to them.
For example, the first generator above is constucted by applying the action $[\star]\,[\star]\,(\cdot)$ to $[\star]=[\bs0]$,
where we must set the argument $[\star]=[\bs0]$ to preserve the two-part composition component of the basis element
we wish to generate. The next two generators above are constructed by applying the actions
$[\star]\,[\bs0]\,(\cdot)$ or $[\bs0]\,[\star]\,(\cdot)$ to $[\star]$. Now consider the final quintic generator above.
Applying the action $[\star]\,[\star]\,(\cdot)$ to $[\bs0]\,[\star]\,([\bs0])$ would produce a generator with too many `$[\star]$'
factors, while in principle, either of the actions $[\star]\,[\bs0]\,(\cdot)$ or $[\bs0]\,[\star]\,(\cdot)$ could be
applied to $[\bs0]\,[\star]\,([\bs0])$. However the action of $[\star]\,[\bs0]\,(\cdot)$ on $[\bs0]\,[\star]\,(\cdot)$
is nilpotent. We demonstrate this below in Lemma~\ref{lemma:nilpotentaction}. Hence the action
$[\star]\,[\bs0]\,(\cdot)$ rigorously applied to $[\bs0]\,[\star]\,([\bs0])$ produces zero. Thus the
only viable action is $[\bs0]\,[\star]\,(\cdot)$ on $[\bs0]\,[\star]\,([\bs0])$ producing the quintic generator shown.
From another perspective, for the quintic generator, in the case $n\geqslant5$, any other quintic arrangment with 
two $[\star]$- and three $[\bs0]$-factors, would necessitate a consecutive pair `$[\bs0]\,[\bs0]$' that would result
in generating a basis element whose composition component has more than two parts.

In the case of basis elements with a three-part composition component $w=a_1a_2a_3$,
the possible generators are, in principle, any of the generators we have already seen, as well as, 
the generators of the form,
\begin{align*}
  [\star]\,[\star]\,([\star]),~
  &[\star]\,[\star]\,\bigl([\bs0]\,[\star]\,([\bs0])\bigr),
  ~[\star]\,[\bs0]\,\bigl([\star]\,[\star]\,([\bs0])\bigr),
  ~[\star]\,[\bs0]\,\bigl([\star]\,[\bs0]\,([\star])\bigr),\\
  &[\bs0]\,[\star]\,\bigl([\star]\,[\star]\,([\bs0])\bigr),
  ~[\bs0]\,[\star]\,\bigl([\star]\,[\bs0]\,([\star])\bigr),
  ~[\bs0]\,[\star]\,\bigl([\bs0]\,[\star]\,([\star])\bigr),\\
  &[\bs0]\,[\star]\,\bigl([\bs0]\,[\star]\,\bigl([\bs0]\,[\star]\,([\bs0])\bigr)\bigr).
\end{align*}
We remark that each possible generator above contains only three `$[\star]$' factors.
We see that the first two generators are constructed by applying the action $[\star]\,[\star]\,(\cdot)$
to the generators for basis elements with one-part composition components.
The next set of generators are constructed by applying the action $[\star]\,[\bs0]\,(\cdot)$ to
the generators for basis elements with two-part composition components, taking into account 
the nilpotent action of $[\star]\,[\bs0]\,(\cdot)$ on $[\bs0]\,[\star]\,(\cdot)$.
This accounts for the next two generators. Then the final four generators are constructed
by applying the action $[\bs0]\,[\star]\,(\cdot)$ to all four of the generators 
for basis elements with two-part composition components.

For the case of basis elements with a four-part composition component $w=a_1a_2a_3a_4$,
the possible generators are, in principle, besides any of the generators we have already seen, 
generators of the form,
\begin{align*}
  &[\star]\,[\star]\,\bigl([\star]\,[\star]\,([\bs0])\bigr),~[\star]\,[\star]\,\bigl([\star]\,[\bs0]\,([\star])\bigr),
  ~[\star]\,[\star]\,\bigl([\bs0]\,[\star]\,([\star])\bigr),~[\star]\,[\bs0]\,\bigl([\star]\,[\star]\,([\star])\bigr),\\
  &[\bs0]\,[\star]\,\bigl([\star]\,[\star]\,([\star])\bigr),
  ~[\star]\,[\star]\,\bigl([\bs0]\,[\star]\,\bigl([\bs0]\,[\star]\,([\bs0])\bigr)\bigr),
  ~[\star]\,[\bs0]\,\bigl([\star]\,[\star]\,\bigl([\bs0]\,[\star]\,([\bs0])\bigr)\bigr),\\
  &[\star]\,[\bs0]\,\bigl([\star]\,[\bs0]\,\bigl([\star]\,[\star]\,([\bs0])\bigr)\bigr),
  ~[\star]\,[\bs0]\,\bigl([\star]\,[\bs0]\,\bigl([\star]\,[\bs0]\,([\star])\bigr)\bigr),
  ~[\bs0]\,[\star]\,\bigl([\star]\,[\star]\,\bigl([\bs0]\,[\star]\,([\bs0])\bigr)\bigr),\\
  &[\bs0]\,[\star]\,\bigl([\star]\,[\bs0]\,\bigl([\star]\,[\star]\,([\bs0])\bigr)\bigr),
  ~[\bs0]\,[\star]\,\bigl([\star]\,[\bs0]\,\bigl([\star]\,[\bs0]\,([\star])\bigr)\bigr),
  ~[\bs0]\,[\star]\,\bigl([\bs0]\,[\star]\,\bigl([\star]\,[\star]\,([\bs0])\bigr)\bigr),\\
  &[\bs0]\,[\star]\,\bigl([\bs0]\,[\star]\,\bigl([\star]\,[\bs0]\,([\star])\bigr)\bigr),
  ~[\bs0]\,[\star]\,\bigl([\bs0]\,[\star]\,\bigl([\bs0]\,[\star]\,([\star])\bigr)\bigr),\\
  &[\bs0]\,[\star]\,\bigl([\bs0]\,[\star]\,\bigl([\bs0]\,[\star]\,\bigl([\bs0]\,[\star]\,([\bs0])\bigr)\bigr)\bigr).
\end{align*}
Again, each possible generator above contains only four `$[\star]$' factors.
They are constructed by applying the action 
$[\star]\,[\star]\,(\cdot)$ to the generators for basis elements with two-part composition components,
applying the action $[\star]\,[\bs0]\,(\cdot)$ to
the generators for basis elements with three-part composition components, taking into account 
the nilpotent action of $[\star]\,[\bs0]\,(\cdot)$ on $[\bs0]\,[\star]\,(\cdot)$, and
then also applying the action $[\bs0]\,[\star]\,(\cdot)$ to all of the generators 
for basis elements with three-part composition components.

We have seen that for basis elements with a composition component with $k=1,2,3$ or $4$ parts,
the number of generators that might produce such a basis element is $2^{k}$.
We have not shown that corresponding to a given basis element, the generators
constructed in the manner indicated are unique at leading order. We demonstrate
this below in Steps $7$ and $8$. Assuming this is the case for the moment,
we have the following.
\begin{lemma}[Generator block size]\label{lemma:magicalmatch1}
  For a given basis element with composition component $w$, the number of
  monomial generators that can generate that basis element at leading order is $2^{|w|}$.
\end{lemma}
\begin{proof}
  As observed, the result is true for $|w|=1,2,3,4$. Assume the result is
  true for $|w|=1,2,\ldots,k$ for some $k\in\mathbb N$.
  The set of generators for basis elements with composition components of $k+1$ parts
  are constructed by: (i) Applying the action  $[\star]\,[\star]\,(\cdot)$
  to the generators for basis elements with $(k-1)$-part composition components of which there are $2^{k-1}$ by assumption;
  (ii) Applying the action $[\star]\,[\bs0]\,(\cdot)$ to the generators for basis elements with $k$-part composition components, taking into account 
  the nilpotent action of $[\star]\,[\bs0]\,(\cdot)$ on $[\bs0]\,[\star]\,(\cdot)$. Since there are $2^k$ generators corresponding to any
  basis element with a composition of $k$-parts, and half of these start with the factor `$\,[\bs0]\,\,[\star]$', there are $2^{k-1}$
  generators constructed in this way; and then finally (iii) Applying the action $[\bs0]\,[\star]\,(\cdot)$ to all of the generators 
  for basis elements with $k$-part composition components, of which there are $2^k$. Adding these three contributions up, $2^{k-1}+2^{k-1}+2^k=2^{k+1}$,
  and the result follows by induction.
\qed
\end{proof}
We can also view this last result from another perspective as follows.
For each $k$-part composition, when $k$ is odd, the set of new generators are characterised 
as follows. First, we include the degree $k$ monomial $[\star]\,[\star]\,\cdots\,[\star]$,
of which there is only one choice. We also include the degree $k+2$ monomials which
contain two non-adjacent `$\,[\bs0]$' factors; there are $k+1$ choose $2$ possible monomials of this form.
Then we can also include degree $k+4$ monomials which contain four non-adjacent `$\,[\bs0]$' factors;
there are $k+1$ choose $4$ possible monomials of this form. And so forth until we reach the single
degree $2k+1$ monomial of the form $[\bs0]\,[\star]\,[\bs0]\,[\star]\,[\bs0]\,\cdots\,[\star]\,[\bs0]$.
Here we have implicitly used that the number of ways to place $r$ objects in non-adjacent slots whose
total number is $m$, is given by $m-r+1$ choose $r$. In the examples just presented, we
considered the number of possible ways of placing $2\ell$ factors of the form `$\,[\bs0]$' in
a monomial of degree $k+2\ell$, for $\ell=0,1,\ldots,(k+1)/2$. Hence the total number of
monomials of each degree outlined being $k+1$ choose $2\ell$. Thus with $k$ odd, the total number of
such new odd-degree monomials is given by the sum over $\ell=0,1,\ldots,(k+1)/2$ of $k+1$ choose $2\ell$,
i.e.\/ the sum on the left shown in Lemma~\ref{lemma:magicalmatch2}. 
Suppose now $k$ is even. The lowest degree monomials that might generate the corresponding basis element  
are those of degree $k+1$ with a single factor `$\,[\bs0]$'. There are $k+1$ such monomials.
We can also include degree $k+3$ monomials with three non-adjacent factors `$\,[\bs0]$'; there
are $k+1$ choose $3$ such possible monomials, and so forth. In the final highest degree monomial,
of degree $2k+1$ has the single form $[\bs0]\,[\star]\,[\bs0]\,[\star]\,[\bs0]\,\cdots\,[\star]\,[\bs0]$.
Thus with $k$ even, the total number of such new odd-degree monomials is given by the sum
over $\ell=0,1,\ldots,k/2$ of $k+1$ choose $2\ell+1$, i.e.\/ the sum on the right shown in Lemma~\ref{lemma:magicalmatch2}. 
In consequence we have the following important result.
\begin{lemma}\label{lemma:magicalmatch2}
  The aforementioned sums, in the respective $k$ is odd and then even cases, are equal to $2^k$.
  In other words, respectively, when $k$ is odd and then even, we have, 
  \begin{equation*}
    \sum_{\ell=0}^{(k+1)/2}\begin{pmatrix} k+1\\ 2\ell\end{pmatrix}=2^k\qquad\text{and}\qquad
    \sum_{\ell=0}^{k/2}\begin{pmatrix} k+1\\ 2\ell+1\end{pmatrix}=2^k.
  \end{equation*}
\end{lemma}
\begin{proof}
  Suppose $k$ is odd. Then by direct computation, we observe,
  \begin{align*}
    \sum_{\ell=0}^{(k+1)/2}\!\frac{(k+1)!}{(k+1-2\ell)!(2\ell)!}
     =&\;2+\!\sum_{\ell=1}^{(k-1)/2}\!\frac{k!}{(k-2\ell)!(2\ell-1)!}\biggl(\frac{1}{k-2\ell+1}+\frac{1}{2\ell}\biggr)\\
     =&\;2+\!\sum_{\ell=1}^{(k-1)/2}\!\frac{k!}{(k+1-2\ell)!(2\ell-1)!}+\!\sum_{\ell=1}^{(k-1)/2}\!\frac{k!}{(k-2\ell)!(2\ell)!}\\
     =&\;1+\begin{pmatrix} k\\ 1 \end{pmatrix}+\begin{pmatrix} k\\ 2\end{pmatrix}+\begin{pmatrix} k\\ 3\end{pmatrix}+\cdots+\begin{pmatrix} k\\ k-1\end{pmatrix}+1,
  \end{align*}
  where we matched up respective pairs from the sums and then used that $2^k=(1+1)^k$. This gives the first result.
  When $k$ is even, we again use that $2^k=(1+1)^k$, and observe,
  \begin{equation*}
    2^k=\begin{pmatrix} k\\ 0 \end{pmatrix}+\begin{pmatrix} k\\ 1\end{pmatrix}+\begin{pmatrix} k\\ 2\end{pmatrix}+\cdots
    +\begin{pmatrix} k\\ k\end{pmatrix}=\sum_{\ell=0}^{k/2}\frac{(k+1)!}{(k-2\ell)!(2\ell+1)!},
  \end{equation*}
  where we paired up successive terms and parameterised the pairs by $\ell=0,1,\ldots,k/2$. 
  This gives the second result.
\qed
\end{proof}
The crucial observation from the result of Lemmas~\ref{lemma:magicalmatch1} and \ref{lemma:magicalmatch2} is the following.
\begin{corollary}[Generator-tuple dimension match]\label{corollary:dimmatch}
  For a given composition component $w$ of a block set of basis elements $[w]\times\bs\beta$ parameterised
  by the tuples $\bs\beta\in\mathbb R\la\mathbb B\ra$, 
  the number of new monomial generators equals the dimension of the tuple block, i.e. $2^{|w|}$.
\end{corollary}
One of our main concerns now is to show that the resulting square block of signature coefficients has full rank.
The next three steps address this issue, making the analysis of this section more precise.\medskip

\emph{Step~4: The three standard triple actions.}
We have seen that the triple product action in Corollary~\ref{cor:tripleproductaction}, in the full form given therein,
as well as in the special forms $[\bs a]\,[\bs0]\,\bigl(\cdot\bigr)$ and $[\bs0]\,[\bs b]\,\bigl(\cdot\bigr)$,
are used to construct the generators corresponding to a given basis element. We call these three actions the
standard triple actions.
\begin{definition}[Standard triple actions]\label{def:standardtripleactions}
  We call the actions $[\bs a]\,[\bs b]\,\bigl(\cdot\bigr)$, $[\bs a]\,[\bs0]\,\bigl(\cdot\bigr)$
  and $[\bs0]\,[\bs b]\,\bigl(\cdot\bigr)$ the three \emph{standard} triple actions.
\end{definition}
The result of the action $[\bs a]\,[\bs b]\,\bigl(\cdot\bigr)$ is given in Corollary~\ref{cor:tripleproductaction}.
As we use them frequently hereafter, we record the result of the standard actions
$[\bs a]\,[\bs0]\,(\cdot)$ and $[\bs0]\,[\bs b]\,(\cdot)$ in the following Corollary.
They are just special cases which we call the \emph{special actions}. 
\begin{corollary}[Special actions]\label{cor:specialactions}
  The two special actions $[\bs a]\,[\bs0]\,\bigl(\cdot\bigr)$ and $[\bs0]\,[\bs b]\,\bigl(\cdot\bigr)$
  are given at leading order by,
  \begin{align*}
  [\bs a]\,[\bs0]\,\bigl([cw]\times\bs\beta\bigr)
  =&\;\chi(a\ot 0\ot cw)\!\cdot\![(a+1)(c+1)w]\times\bigl((1,-1)\lhd\bigl(\bs\beta+(-1)^{|w|}\bs\beta^\dag\bigr)\bigr)+\cdots,\\  
  [\bs0]\,[\bs b]\,\bigl([cw]\times\bs\beta\bigr)
  =&\;\chi(0\ot b\ot cw)\!\cdot\!\Bigl([(b+2)cw]\times\bigl((2,0)\lhd\bs\beta+(-1)^{|w|}(0,2)\lhd\bs\beta^\dag\bigr)\\
  &\;\qquad\qquad+[(b+1)(c+1)w]\times\bigl((1,1)\lhd\bigl(\bs\beta+(-1)^{|w|}\bs\beta^\dag\bigr)\bigr)\Bigr)+\cdots.
  \end{align*}
\end{corollary}
\begin{proof}
  By direct computation using the P\"oppe product rules in Lemma~\ref{lemma:skewandsymmPoppeproducts}, we observe that
  $[\bs a]\,[\bs0]\,\bigl([cw\times\bs\varphi]\bigr)$ equals,
\begin{align*}
  \chi(a\ot 0\ot cw)\cdot\bigl(&\;[\bs0(a+1)\bs0(c+1)(w\times\hat{\bs\varphi})]+[\bs0(a+1)\bs0(c+1)(w\times\hat{\bs\varphi})^\dag]\\
  &\;-[\bs0(a+1)\bs0^\dag(c+1)(w\times\hat{\bs\varphi})]-[\bs0(a+1)\bs0^\dag(c+1)(w\times\hat{\bs\varphi})^\dag\bigr)+\cdots,
\end{align*}
at leading order, giving the first result.
Then, by direct computation for the other case, we observe that $[\bs0]\,[\bs b]\,\bigl([cw\times\bs\varphi]\bigr)$ equals,
\begin{align*}
  \chi(0\ot b\ot cw)\cdot\bigl(&\;2\cdot[\bs0(b+2)\bs0c)(w\times\hat{\bs\varphi})]+2\cdot[\bs0(b+2)\bs0^\dag c)(w\times\hat{\bs\varphi})^\dag]\\
  &\;+[\bs0(b+1)\bs0(c+1)(w\times\hat{\bs\varphi})]+[\bs0(b+1)\bs0^\dag(c+1)(w\times\hat{\bs\varphi})]\\
  &\;+[\bs0(b+1)\bs0(c+1)(w\times\hat{\bs\varphi})^\dag]+[\bs0(b+1)\bs0^\dag(c+1)(w\times\hat{\bs\varphi})^\dag\bigr)+\cdots,
\end{align*}
at leading order, giving the second result. \qed
\end{proof}
\begin{remark}
  Comparing the results of Corollary~\ref{cor:specialactions} with Corollary~\ref{cor:tripleproductaction} we emphasise
  two observations, that there is: (i) A natural contraction of the action forms due to the `$\,[\bs0]$' factors in the action;
  (ii) An apparent change of sign in the second term on the right in the first example. We can view both cases as
  the consequence of substituting $[\nu\times\bs0]$ for $[\bs b]$ in the first case and then $[\nu\times\bs0]$ for $[\bs a]$
  in the second case. The sign change, perhaps more easily observed from the corresponding result in Lemma~\ref{lemma:tripleproductaction},
  is a consquence of the fact that to make the appropriate substitution of $[\nu\times\bs0]$ for $[\bs b]$, we should
  convert the two terms with $\bs0^\dag b\bs0^\dag$ on the right, to $-\bs0^\dag b^\dag\bs0^\dag$ first.
\end{remark}
\medskip

\emph{Step~5: Generating blocks.} We now show precisely how, given a block of basis elements characterised
by a given composition component $w$ and parameterised by the corresponding $2^{|w|}$ basis elements $\bs\beta$ of $\mathbb R\la\mathbb B\ra$,
we can use the three standard actions to enumerate all the monomial generators that produce the basis elements
of that block at leading order, and also establish the corresponding signature coefficient associated with
each such basis element. Let us examine the three standard actions given in 
Corollaries~\ref{cor:tripleproductaction} and \ref{cor:specialactions} more closely.
If we examine the right-hand side of $[\bs a]\,[\bs b]\,\bigl([cw]\times\bs\beta\bigr)$
in Corollary~\ref{cor:tripleproductaction}, then we observe that in terms of descent order, the first
composition term `$[(a+1)(b+1)cw]$' on the right is highest, and thus we retain that term only.
In Corollary~\ref{cor:specialactions}, at leading order, the action $[\bs a]\,[\bs0]\,\bigl([cw]\times\bs\beta\bigr)$ in unique, while
the right-hand side of $[\bs0]\,[\bs b]\,\bigl([cw]\times\bs\beta\bigr)$ contains two terms, the
first of which is higher in terms of descent order, which is the one we retain.
Thus at leading order the three standard actions on $[cw]\times\bs\beta$ are:
\begin{align*}
[\bs a]\,[\bs b]\,(\cdot)&\!=\!\chi(cw)\!\cdot\![(a+1)(b+1)cw]\!\times\!\bigl((1,0,1,0)\lhd\bs\beta+(-1)^{|w|}\,(0,1,0,1)\lhd\bs\beta^\dag\bigr)+\cdots,\\
[\bs a]\,[\bs0]\,(\cdot)&\!=\!\chi(cw)\!\cdot\![(a+1)(c+1)w]\!\times\!\bigl((1,-1)\lhd\bs\beta+(-1)^{|w|}(1,-1)\lhd\bs\beta^\dag\bigr)+\cdots,\\
[\bs0]\,[\bs b]\,(\cdot)&\!=\!\chi(cw)\!\cdot\![(b+2)cw]\!\times\!\bigl((2,0)\lhd\bs\beta+(-1)^{|w|}(0,2)\lhd\bs\beta^\dag\bigr)+\cdots.
\end{align*}
Here we have used the homomorphic properties of $\chi$, in particular that $\chi(a\ot b\ot cw)=\chi(a\ot 0\ot cw)=\chi(0\ot b\ot cw)=\chi(cw)$.
Consider the following respective replacements in each of the three actions above:
(i) $a\to a-1$, $b\to b-1$, $c\to\nu$; (ii) $a\to a-1$, $c\to c-1$ and (iii) $b\to b-2$.
With these three choices, each of the actions generates the same composition $acw$---in the first case we relabel $b$ as $c$ and
in the third case we relabel $b$ as $a$. Recall from our coarse-grain overview in Step~3 that to enumerate the generators
corresponding to basis elements with composition components with $k\geqslant2$ parts, we apply the first action to the generators
at level $k-2$, and the two special actions to the generators at level $k-1$, taking into account the nilpotent action outlined
just below in Lemma~\ref{lemma:nilpotentaction}.
We note that, for any sequence $\hat u\in\mathbb R\la\mathbb B\ra$, with $|\hat u|=2^{k-2}$, we have,
\begin{align*}
  (1,0,1,0)\lhd\hat u&=(1,1)\lhd(1,0)\lhd\hat u=(1,1)\lhd(\hat u,0),\\
  (0,1,0,1)\lhd\hat u^\dag&=(1,1)\lhd(0,1)\lhd\hat u^\dag=(1,1)\lhd(0,{\hat u}^\dag),
\end{align*}
where $(\hat u,0)$ and $(0,\hat u^\dag)$ are of length $2^{k-1}$.
Putting these observations together, we have thus established the following lemma.
\begin{lemma}[Actions generating the same composition]\label{lemma:actionsamecomp}
At leading order, with the choices mentioned above, the following three standard actions generate the same composition with the
respective $\mathbb R\la\mathbb B\ra$ components indicated,
\begin{align*}
  [\bs{a-1}]\,[\bs{c-1}]\,\bigl([w]\!\times\!\bs\beta\bigr)&\!=\!\chi(w)\cdot[acw]\times\bigl((1,1)\lhd\bigl((\hat u,0)
  -(-1)^{|w|}(\hat u,0)^\dag\bigr)\bigr),\\ 
  [\bs{a-1}]\,[\bs0]\,\bigl([(c-1)w]\!\times\!\bs\beta\bigr)&\!=\!\chi((c-1)w)\cdot[acw]\times\bigl((1,-1)\lhd\bigl((\hat a,\hat b)
  +(-1)^{|w|}(\hat a,\hat b)^\dag\bigr)\bigr),\\
  [\bs0]\,[\bs{a-2}]\,\bigl([cw]\!\times\!\bs\beta\bigr)&\!=\!\chi(cw)\cdot[acw]\times\bigl((2,0)\lhd(\hat a,\hat b)
  +(-1)^{|w|}(0,2)\lhd(\hat a,\hat b)^\dag\bigr).
\end{align*}
Here, in the first case $\bs\beta=(\hat u,0)\in\mathbb R\la\mathbb B\ra$ with $\hat u$ arbitrary, and in the second and third cases
$\bs\beta=(\hat a,\hat b)\in\mathbb R\la\mathbb B\ra$ is arbitrary. Each such $\bs\beta$ is of length $2^{|acw|-1}$,
and $\hat a$ and $\hat b$ have the same length---matching that of $\hat u$.
\end{lemma}
\begin{remark}
  Note, in the statement of Lemma~\ref{lemma:actionsamecomp}, the case of the first action which corresponds to
  the action $[\bs a]\,[\bs b]\,(\cdot)$ applied to $[cw]\times\bs\beta$ in the discussion preceding the Lemma.
  In that discussion, when we set $c\to\nu$, we equivalently replaced $cw$ by $w$. This means that
  we should effectively consider the length of $w$ to be one less than would otherwise be the
  case. This explains why the sign in front of the term with the factor $(-1)^{|w|}$ in the first action case is negative
  in the statement of the Lemma.
\end{remark}
Some further clarifications on the statement of Lemma~\ref{lemma:actionsamecomp} are required. Note that,
\begin{equation*}
  [\bs0]\rightsquigarrow[\nu]\times(1),
\end{equation*}
where $\nu$ is the empty composition and $(1)$ is the element of $\mathbb R\la\mathbb B\ra$ corresponding to
compositions of zero parts. The special action $[\bs0]\,[\bs{a-2}]\,(\cdot)$  in Lemma~\ref{lemma:actionsamecomp}
still applies when the argument $[cw]\times\bs\beta=[\nu]\times(1)$ and thus when $(\hat a,\hat b)=(1)$.
The result is that at leading order we have,
\begin{equation*}
  [\bs0]\,[\bs{a-2}]\,\bigl([\nu]\times(1)\bigr)=\chi(\nu)\cdot[a]\times\bigl((2,0)\lhd(1)-(0,2)\lhd(1)^\dag\bigr)=[a]\times(2,-2).
\end{equation*}
Here, by convention, we take $\chi(\nu)\coloneqq1$. Since we have taken $cw\to\nu$,
we can think of the number of parts of $w$ to be `$-1$', explaining the sign
in front of the $\mathbb R\la\mathbb B\ra$-element $(0,2)$. This is consistent with
just computing $[\bs0]\,[\bs{a-2}]\,[\bs0]$. Further, the first two actions in Lemma~\ref{lemma:actionsamecomp}
don't make sense when $cw\to\nu$, though if $w\to\nu$, the special action $[\bs{a-1}]\,[\bs0]\,(\cdot)$
applies with the appropriate adaptations. And of course we can compute
$[\bs{a-1}]\,[\bs{c-1}]\,\bigl([\nu]\times(1)\bigr)=[\bs{a-1}]\,[\bs{c-1}]\,[\bs0]$.

Finally, we now also observe the following (aforementioned) nilpotency property.
\begin{lemma}[Nilpotent action]\label{lemma:nilpotentaction}
  At leading order, if we first apply the action $[\bs0]\,[\bs b]\,(\cdot)$ to an arbitrary $\mathbb R\la\mathbb B\ra$ component,
  then apply the action $[\bs a]\,[\bs 0]\,(\cdot)$ to the result, this generates the zero $\mathbb R\la\mathbb B\ra$ component.
  In other words at leading order we have,
  \begin{equation*}
    [\bs a]\,[\bs 0]\,\bigl([\bs0]\,[\bs b]\,(\cdot)\bigr)=0,
  \end{equation*}
  where the `$\,0$' on the right-hand side represents the zero $\mathbb R\la\mathbb B\ra$ component.
\end{lemma}
\begin{proof}
  We focus on the effect of the actions on the $\mathbb R\la\mathbb B\ra$ components only.
  The third (special) action applied to the input $(a,b)$ generates $2\cdot(a,b,\pm b^\dag,\pm a^\dag)$.
  Set $A,B\in\mathbb R\la\mathbb B\ra$ to be the sub-components $A\coloneqq(a,b)$ and $B\coloneqq\pm(b^\dag,a^\dag)$.
  With these identifications we note that
  $B=\pm A^\dag$. Ignoring the real factor $2$, apply the second (special) action to the input $(A,B)$.
  This is (note the sign of the second term of the action changes), $(1,-1)\lhd\bigl((A,B)\mp(A,B)^\dag\bigr)$,
  which equals, $(A\mp B^\dag,B\mp A^\dag,-A\pm B^\dag,-B\pm A^\dag)$.
  Since $B=\pm A^\dag$, this result is the zero $\mathbb R\la\mathbb B\ra$ component.
\qed
\end{proof}
We now explore, through a series of examples, how to construct the generators and coefficient blocks
associated with any given composition. In particular we consider the cases of compositions with
$1$, $2$ and $3$ parts, before exploring the case of any given composition. Compositions
containing a `$1$' need to be singled out, as explained below.
\begin{example}[One-part compositions]\label{ex:onepartcomps}
  We observe that there are two basis elements corresponding to the one-part composition $w=a$, namely,
  $[a]\times(1,0)$ and $[a]\times(0,1)$. We assume $n=a\geqslant2$. At leading we know from the corresponding
  signature expansion $[\bs a]=[a]\times(1,0)+\cdots$. From our discussion succeeding Lemma~\ref{lemma:actionsamecomp},
  we know the first two actions do not make sense when $cw\to\nu$, while the final special action
  does make sense. As we saw directly, at leading order we have $[\bs0]\,[\bs{a-2}]\,[\bs0]=[a]\times(2,-2)$.
  We have thus enumerated the generators corresponding to $[a]\times(1,0)$ and $[a]\times(0,1)$ and that
  the signature coefficient matrix is,
  \begin{equation*}
  A_0=\begin{pmatrix} 1 & 2\\ 0 & -2 \end{pmatrix}.
  \end{equation*}
\end{example}
\begin{example}[Two-part compositions]\label{ex:twopartcomps}
  Consider the basis elements with a two-part composition component $a_1a_2$, i.e.\/ basis
  elements of the form $[a_1a_2]\times\bs\beta_i$, where the $\bs\beta_i$ are the
  four basis elements of length $4$, which are zero apart from a `$1$' in the $i$th position.
  For the moment assume neither $a_1$ nor $a_2$ equal $1$; we consider each of these two special
  cases separately below. Using Lemma~\ref{lemma:actionsamecomp}, noting that for
  each of the standard actions our goal is to obtain the composition component $[a_1a_2]$
  on the right-hand side, we observe the following. For the first action, setting $w=\nu$,
  $a=a_1$ and $c=a_2$, we find that at leading order, we get,
  %
  %
  \begin{equation*}
    [\bs{a_1-1}]\,[\bs{a_2-1}]\,[\bs0]=[a_1a_2]\times\bigl((1,1)\lhd\bigl((1,0)-(0,1)\bigr)\bigr)=[a_1a_2]\times(1,-1,1,-1).
  \end{equation*}
  The first special action in Lemma~\ref{lemma:actionsamecomp}, with the same identifications gives, 
  \begin{equation*}
    [\bs{a_1-1}]\,[\bs0]\,\bigl([a_2-1]\times(\hat a,\hat b)\bigr)=[a_1a_2]\times\bigl((1,-1)\lhd\bigl((\hat a,\hat b)+(\hat b,\hat a)\bigr)\bigr).
  \end{equation*}
  We saw in Example~\ref{ex:onepartcomps}, the basis element $[a_2-1]\times(\hat a,\hat b)$
  can be generated both by the corresponding signature expansion $[\bs{a_2-1}]=[a_2-1]\times(1,0)+\cdots$,
  and by the generator $[\bs0]\,[\bs{a-2}]\,[\bs0]$. We discount the latter case due to the nilpotent action property.
  Hence using this expression for $[\bs{a_2-1}]$ and inserting $(\hat a,\hat b)=(0,1)$ into the expression above, we deduce, 
  \begin{equation*}
    [\bs{a_1-1}]\,[\bs0]\,\bigl([\bs{a_2-1}]\bigr)=[a_1a_2]\times(1,1,-1,-1),
  \end{equation*}
  to leading order. Now consider the second special action in Lemma~\ref{lemma:actionsamecomp}.
  Again with the same identifications for $a$, $c$ and $w$, we observe that to leading order,
  \begin{equation*}
    [\bs0]\,[\bs{a_1-2}]\,\bigl([a_2]\times(\hat a,\hat b)\bigr)=[a_1a_2]\times\bigl((2,0)\lhd(\hat a,\hat b)+(0,2)\lhd(\hat b,\hat a)\bigr).
  \end{equation*}
  We know from  Example~\ref{ex:onepartcomps}, the basis element $[a_2]\times(\hat a,\hat b)$
  can be generated either by the signature expansion $[\bs{a_2}]=[a_2]\times(1,0)+\cdots$,
  or by the generator $[\bs0]\,[\bs{a_2-2}]\,[\bs0]=[a_2]\times(2,-2)+\cdots$. Respectively substituting the expressions
  $[a_2]\times(1,0)$ and $[a_2]\times(2,-2)$ for $[a_2]\times(\hat a,\hat b)$ in the relation just above, we observe that to leading order,
  \begin{align*}
    [\bs0]\,[\bs{a_1-2}]\,\bigl([\bs{a_2}]\bigr)&=[a_1a_2]\times\bigl((2,0)\lhd(1,0)+(0,2)\lhd(0,1)\bigr)\\
    &=[a_1a_2]\times(2,0,0,2),\\
    [\bs0]\,[\bs{a_1-2}]\,\bigl([\bs0]\,[\bs{a_2-2}]\,[\bs0]\bigr)&=[a_1a_2]\times\bigl((2,0)\lhd(2,-2)+(0,2)\lhd(-2,2)\bigr)\\
    &=[a_1a_2]\times(4,-4,-4,4).
  \end{align*}
  We have thus enumerated the four generators corresponding to the four basis elements $[a_1a_2]\times(1,0,0,0)$,
  $[a_1a_2]\times(0,1,0,0)$, $[a_1a_2]\times(0,0,1,0)$ and $[a_1a_2]\times(0,0,0,1)$. They are
  $[\bs{a_1-1}]\,[\bs{a_2-1}]\,[\bs0]$, $[\bs{a_1-1}]\,[\bs0]\,[\bs{a_2-1}]$, $[\bs0]\,[\bs{a_1-2}]\,[\bs{a_2}]$
  and the quintic generator $[\bs0]\,[\bs{a_1-2}]\,[\bs0]\,[\bs{a_2-2}]\,[\bs0]$. The corresponding signature coefficient matrix is,
    \begin{equation*}
  A_2\coloneqq\begin{pmatrix}
    1 & 1 & 2 & 4 \\
   -1 & 1 & 0 & -4 \\
    1 & -1& 0 & -4 \\
   -1 & -1 & 2& 4
  \end{pmatrix}
  \end{equation*}
  which is the subsystem coefficient matrix $A_2$ in Examples~\ref{ex:fourthorder} and \ref{ex:fifthorder}
  respectively concerning the quartic and quintic non-commutative nonlinear Schr\"odinger equations.

  Let us now consider the case when $a_2=1$. If we substitute this value for $a_2$ into the generators above,
  we see that the first two generators coincide and are given by $[\bs{a_1-1}]\,[\bs0]\,[\bs{0}]=[a_11]\times(1,-1,1,-1)+\cdots$
  and $[\bs{a_1-1}]\,[\bs0]\,[\bs{0}]=[a_11]\times(1,1,-1,-1)+\cdots$. Since we can add them together under 
  the same coefficient $c_{(a_1-1)00}$, in this case we have the single generator,
  $[\bs{a_1-1}]\,[\bs0]\,[\bs{0}]=[a_11]\times(2,0,0,-2)+\cdots$.
  The third generator above becomes, $[\bs0]\,[\bs{a_1-2}]\,[\bs1]=[a_11]\times(2,0,0,2)+\cdots$.
  The final quintic generator cannot be a generator in this case if we insist on only including signature expansions  
  corresponding to non-negative integers. There are thus only two independent generators. Hence this this case,
  the corresponding signature coefficient matrix, ignoring the middle two rows, is
  \begin{equation*}
  A_1\coloneqq\begin{pmatrix} 2 & 2\\ -2 & 2 \end{pmatrix}.
  \end{equation*}
  See Examples~\ref{ex:fourthorder} and \ref{ex:fifthorder} and the equations
  for the coefficients $c_{(n-2)00}$ and $c_{0(n-3)1}$ in those cases for when $w=(n-1)1$, as
  well as with the coefficients in Tables~\ref{table:NLS4}--\ref{table:NLS5b}. 
  Note, when $a_1=a_2=1$, there is only one generator, $[\bs0]^3$, as we saw in Example~\ref{ex:secondorder}.
  We treat the more general case when $a_1=1$ at the end of this step..
\end{example}
\begin{example}[Three-part compositions]\label{ex:threepartcomps}
  Consider basis elements with a three-part composition component $a_1a_2a_3$, i.e.\/ basis  
  elements of the form $[a_1a_2a_3]\times\bs\beta_i$, where the $\bs\beta_i$ for $i=1,\ldots,8$,
  contain `$1$' in the $i$th position and zeros in the remaining seven positions.  
  For the moment assume neither $a_1$ nor $a_2$ nor $a_3$ are unity.
  Using Lemma~\ref{lemma:actionsamecomp}, the standard actions, setting $a=a_1$, $c=a_2$ and $w=a_3$ give to leading order,
  \begin{align*}
  [\bs{a_1-1}]\,[\bs{a_2-1}]\,\bigl([a_3]\times(\hat u,0)\bigr)&=[a_1a_2a_3]\times\bigl((1,1)\lhd(\hat u,\hat u^\dag)\bigr),\\ 
  [\bs{a_1-1}]\,[\bs0]\,\bigl([(a_2-1)a_3]\times(\hat a,\hat b)\bigr)&=\chi((a_2-1)a_3)\cdot[a_1a_2a_3]\\
  &\qquad\qquad\qquad\quad\times\bigl((1,-1)\lhd\bigl((\hat a,\hat b)
  -(\hat a,\hat b)^\dag\bigr)\bigr),\\
  [\bs0]\,[\bs{a_1-2}]\,\bigl([a_2a_3]\times(\hat a,\hat b)\bigr)&=\chi(a_2a_3)\cdot[a_1a_2a_3]\times(2\hat a,2\hat b,-2\hat b^\dag,-2\hat a^\dag).
  \end{align*}
  We observe, with these three relations, the task of finding the generators for any basis element
  with a three-part composition component, becomes the task of finding the generators for the basis element with the one-part
  composition component `$[a_3]$' in the first case, and then the generators for basis elements with the two-part  
  components `$[(a_2-1)a_3]$' and `$[a_2a_3]$' in the second and third cases. We can construct the
  generators in these cases via Examples~\ref{ex:onepartcomps} and \ref{ex:twopartcomps} just above.
  In the first case, from Example~\ref{ex:onepartcomps}, the two generators for $[a_3]\times(1,0)$ and $[a_3]\times(0,1)$
  are $[\bs a_3]=[a_3]\times(1,0)+\cdots$ and $[\bs0]\,[\bs{a_3-2}]\,[\bs0]=[a_3]\times(2,-2)+\cdots$.
  Hence if we substitute these expressions into the first case above, respectively setting $\hat u=(1,0)$ 
  and then $\hat u=(2,-2)$, we find,
  \begin{align*}
    [\bs{a_1-1}]\,[\bs{a_2-1}]\,\bigl([\bs{a_3}]\bigr)&=[a_1a_2a_3]\times\bigl((1,1)\lhd(1,0,0,1)\bigr)\\
                                                      &=[a_1a_2a_3]\times(1,0,0,1,1,0,0,1),\\
    [\bs{a_1-1}]\,[\bs{a_2-1}]\,\bigl([\bs0]\,[\bs{a_3-2}]\,[\bs0]\bigr)&=[a_1a_2a_3]\times\bigl((1,1)\lhd(2,-2,-2,2)\bigr)\\
                                                                        &=[a_1a_2a_3]\times(2,-2,-2,2,2,-2,-2,2).
  \end{align*}
  For the second case above with composition component `$[(a_2-1)a_3]$', we know from Example~\ref{ex:twopartcomps},
  there are four possible generators. However once we observe the nilpotent action property, we are
  left with two, namely, $[\bs{a_2-2}]\,[\bs{a_3-1}]\,[\bs0]=[(a_2-1)a_3]\times(1,-1,1,-1)+\cdots$ and
  $[\bs{a_2-2}]\,[\bs0]\,[\bs{a_3-1}]=[(a_2-1)a_3]\times(1,1,-1,-1)+\cdots$. Substituting these expressions
  into the second case above, respectively setting $(\hat a,\hat b)=(1,-1,1,-1)$ and then $(\hat a,\hat b)=(1,1,-1,-1)$, we find,
  \begin{align*}
    [\bs{a_1-1}]\,&[\bs0]\,\bigl([\bs{a_2-2}]\,[\bs{a_3-1}]\,[\bs0]\bigr)\\
    &=\chi((a_2-1)a_3)\cdot[a_1a_2a_3]\times\bigl((1,-1)\lhd\bigl((1,-1,1,-1)-(1,-1,1,-1)^\dag\bigr)\bigr)\\
    &=\chi((a_2-1)a_3)\cdot[a_1a_2a_3]\times(2,-2,2,-2,-2,2,-2,2),\\ 
    [\bs{a_1-1}]\,&[\bs0]\,\bigl([\bs{a_2-2}]\,[\bs0]\,[\bs{a_3-1}]\bigr)\\
    &=\chi((a_2-1)a_3)\cdot[a_1a_2a_3]\times\bigl((1,-1)\lhd\bigl((1,1,-1,-1)-(1,1,-1,-1)^\dag\bigr)\bigr)\\
    &=\chi((a_2-1)a_3)\cdot[a_1a_2a_3]\times(2,2,-2,-2,-2,-2,2,2).
  \end{align*}
  For the third case above with composition component `$[a_2a_3]$', again, we know from Example~\ref{ex:twopartcomps},
  there are four possible generators. These are all four of the generators shown in Example~\ref{ex:twopartcomps}
  once we replace $a_1$ and $a_2$ therein respectively by $a_2$ and $a_3$. If we substitute the corresponding
  four expressions with the replacements mentioned into the third case above, respectively setting  
  $(\hat a,\hat b)=(1,-1,1,-1)$, $(\hat a,\hat b)=(1,1,-1,-1)$, $(\hat a,\hat b)=(2,0,0,2)$ and then $(\hat a,\hat b)=(4,-4,-4,4)$,
  we find at leading order,
  \begin{align*}
    [\bs0]\,[\bs{a_1-2}]\,\bigl([\bs{a_2-1}]\,[\bs{a_3-1}]\,[\bs0]\bigr)
    &=\chi(a_2a_3)\cdot[a_1a_2a_3]\times(2,-2,2,-2,2,-2,2,-2),\\
    [\bs0]\,[\bs{a_1-2}]\,\bigl([\bs{a_2-1}]\,[\bs0]\,[\bs{a_3-1}]\bigr)
    &=\chi(a_2a_3)\cdot[a_1a_2a_3]\times(2,2,-2,-2,2,2,-2,-2),\\
    [\bs0]\,[\bs{a_1-2}]\,\bigl([\bs0]\,[\bs{a_2-2}]\,[\bs{a_3}]\bigr)
    &=\chi(a_2a_3)\cdot[a_1a_2a_3]\times(4,0,0,4,-4,0,0,-4),
    \end{align*}
  and finally,
  \begin{equation*}
    [\bs0]\,[\bs{a_1\!-\!2}]\,\bigl([\bs0]\,[\bs{a_2\!-\!2}]\,[\bs0]\,[\bs{a_3\!-\!2}]\,[\bs0]\bigr)
    \!=\!\chi(a_2a_3)\cdot[a_1a_2a_3]\times(8,-8,-8,8,-8,8,8,-8).
  \end{equation*}
  Hence, for the eight basis elements $[a_1a_2a_3]\times\bs\beta_i$, $i=1,\ldots,8$,
  with the columns corresponding to the generators above in descent order following by degree, 
  the corresponding signature coefficient matrix, is the full rank matrix,
  \begin{equation*}
  A_3\coloneqq\begin{pmatrix}
    1 &  2 &  2 & 2  & 2  & 2  & 4  & 8 \\   
    0 & -2 & -2 & 2  & -2 & 2  & 0  & -8 \\   
    0 & -2 &  2 & -2 & 2  & -2 & 0  & -8 \\
    1 & 2  & -2 & -2 & -2 & -2 & 4  & 8 \\
    1 & 2  & -2 & -2 &  2 & 2  & -4 & -8 \\
    0 & -2 &  2 & -2 & -2 & 2  & 0  & 8 \\
    0 & -2 & -2 &  2 &  2 & -2 & 0  & 8 \\
    1 & 2  & 2  & 2  & -2 & -2 & -4 & -8 
  \end{pmatrix},
  \end{equation*}
  where columns $3$ and $4$ should involve the factor $\chi((a_2-1)a_3)$, while
  columns $5$ through to $8$ should involve the factor $\chi(a_2a_3)$. The factors
  are omitted in $A_3$ for clarity.
\end{example}  
\begin{remark}
Examples~\ref{ex:onepartcomps}--\ref{ex:threepartcomps} precisely reflect the analysis we outlined in Step~3.
\end{remark}
We can now discern the pattern. Suppose we wish to construct all the generators corresponding to a full set of $2^k$
basis elements associated with a given composition component $[a_1a_2\cdots a_k]$. We preclude 
for the moment, that any of $a_1$ through to $a_k$ are equal to `$1$'. Using the standard actions
in Lemma~\ref{lemma:actionsamecomp}, we find,
\begin{align*}
  [\bs{a_1-1}]\,[\bs{a_2-1}]\,\bigl([a_3\cdots a_k]\times(\hat u,0)\bigr)&=[a_1\cdots a_k]\\
  &\quad\times\bigl((1,1)\lhd\bigl((\hat u,0)-(-1)^{k-2}(0,\hat u^\dag)\bigr)\bigr),\\ 
  [\bs{a_1-1}]\,[\bs0]\,\bigl([(a_2-1)a_3\cdots a_k]\times(\hat a,\hat b)\bigr)&=\chi((a_2-1)a_3)\cdot[a_1\cdots a_k]\\
  &\quad\times\bigl((1,-1)\lhd\bigl((\hat a,\hat b)+(-1)^{k-2}(\hat a,\hat b)^\dag\bigr)\bigr),\\
  [\bs0]\,[\bs{a_1-2}]\,\bigl([a_2\cdots a_k]\times(\hat a,\hat b)\bigr)&=\chi(a_2a_3)\cdot[a_1\cdots a_k]\\
  &\quad\times\bigl((2,0)\lhd(\hat a,\hat b)+(-1)^{k-2}(0,2)\lhd(\hat a,\hat b)^\dag\bigr),
\end{align*}
to leading order. We then act iteratively to substitute for the generators corresponding to the composition:
(i) $[a_3\cdots a_k]$ with $(k-2)$-parts; (ii) $[(a_2-1)a_3\cdots a_k]$ with $(k-1)$-parts, taking into account
the nilpotency action property; and (iii) $[a_2\cdots a_k]$ with $(k-1)$-parts. We know from Step~3,
for a given composition $a_1\cdots a_k$, we can construct $2^k$ unique generators in this way.
In Step~6 we demonstrate that the corresponding signature coefficient matrix generated in this way has full rank.
However, at this stage, we note that we have the following straightforward result.
\begin{lemma}[Generator sets unique to compositions]\label{lemma:uniquegenerators}
  If we use the procedure above to construct the $2^k$ generators associated with the $2^k$ basis elements
  with a given composition component $a_1\cdots a_k$, then each such set of generators is unique to the given composition $a_1\cdots a_k$,
  i.e.\/ each such set of generators corresponding to a given composition $a_1\cdots a_k$, does not appear elsewhere
  in generator sets for other compositions.
\end{lemma}
Lastly, apart from the case of the composition `$a_11$' in Example~\ref{ex:twopartcomps}, including the case `$11$',
our analysis above has precluded compositions $a_1\cdots a_k$ containing `$1$' in the composition sequence.
We saw at the end of Example~\ref{ex:twopartcomps}, that provided $a_1\neq1$ then the set of generators
for the basis elements with composition components $a_11$ reduces to two generators only, however, both generators
only generate the basis elements $[a_11]\times(1,0,0,0)$ and $[a_11]\times(0,0,0,1)$.
Consider the case of the composition $a_1a_21$, with both $a_1\neq1$ and $a_2\neq1$.
Using arguments analogous to those for the case `$a_11$' at the end of Example~\ref{ex:twopartcomps},
if we examine the eight generators listed in Example~\ref{ex:threepartcomps},
we observe that the second and last cannot be generators in the case when $a_3=1$,
while the third and fourth generators combine, and the fifth and sixth generators combine,
in much the same way as for the case of `$a_11$' in Example~\ref{ex:twopartcomps}.
The latter two correspond to adding the third and fourth, and also the fifth and sixth, columns
in the $8\times8$-matrix $A_3$ above. Thus for the case of the composition $a_1a_21$, the resulting coefficient matrix, ignoring the
second, third, sixth and seventh rows which are zero, corresponds to the coefficient matrix $A_3^\prime$
in Example~\ref{ex:fifthorder}. With these last two examples in hand, we deduce that for any composition
of the form $a_1a_2\cdots a_{k-1}1$, where we preclude any of $a_1$ through to $a_{k-1}$ to be `$1$',
we have a unique set of generators in the sense of Lemma~\ref{lemma:uniquegenerators}, albeit with
a signature coefficient matrix of size $2^{k-1}\times2^{k-1}$. Now consider the case when the
composition component is $[1a_2\cdots a_k]$, assume for the moment none of $a_2$ throught to $a_k$ equal `$1$'.
Looking at the standard actions in Lemma~\ref{lemma:actionsamecomp}, we observe that for basis elements with
such a composition component, the two valid actions are,
\begin{align*}
  [\bs0]\,[\bs{a_2-1}]\,\bigl([a_3\cdots a_k]\times(\hat u,0)\bigr)&=\chi(a_3\cdots a_k)\cdot[a_1\cdots a_k]\times(\star,\star),\\
  [\bs0]\,[\bs0]\,\bigl([(a_2-1)a_3\cdots a_k]\times(\hat a,\hat b)\bigr)&=\chi((a_2-1)a_3\cdots a_k)\cdot[a_1\cdots a_k]\times(\star,\star),
\end{align*}
where the two expressions $(\star,\star)$ are proxies for the appropriate $\mathbb R\la\mathbb B\ra$-components whose exact form is
not important at this stage. However, we now observe that if we apply the final action in Lemma~\ref{lemma:actionsamecomp} respectively
for the cases of the compositions $[(a_2+1)a_3\cdots a_k]$ and $[2(a_2-1)a_3\cdots a_k]$, we find,
\begin{align*}
  [\bs0]\,[\bs{a_2-1}]\,\bigl([a_3\cdots a_k]\times(\hat u,0)\bigr)&=\chi(a_3\cdots a_k)\cdot[(a_2+1)\cdots a_k]\times(\star,\star),\\
  [\bs0]\,[\bs0]\,\bigl([(a_2-1)a_3\cdots a_k]\times(\hat a,\hat b)\bigr)&=\chi((a_2-1)a_3\cdots a_k)\cdot[2(a_2-1)a_3\cdots a_k]\times(\star,\star).
\end{align*}
We observe that the first two respective factors of the generators and their arguments match the two cases corresponding to
the composition $[1a_2\cdots a_k]$. However the latter two cases generate basis elements with the respective composition components
$[(a_2+1)a_3\cdots a_k]$ and $[2(a_2-1)a_3\cdots a_k]$, both of which occur before the composition $[1a_2\cdots a_k]$ in descent order.
Thus these generators will not be new. A similar scenario could occur if one or more letters $a_2$ through to $a_k$ are equal to $1$.
In our main proof below in Step~8, we are able to discount any compositions $a_1\cdots a_k$ in which any one of the letters $a_1$ through to $a_{k-1}$ 
is equal to `$1$'.
\medskip

\emph{Step~6: Full rank blocks.}
Our goal in this step is to show that for a given composition of length $k$, for which in general there are $2^k$ different possible 
$\mathbb R\la\mathbb B\ra$ components, there are $2^k$ independent generators, generated by the first action
acting on generators at level $k-2$ and the second and third special actions on generators at level $k-1$.
The following results establish that this is indeed the case.
\begin{lemma}[Actions and independence]\label{lemma:indeptaction}
  We have the following, at leading order:

  (i) Given an independent set of input $\mathbb R\la\mathbb B\ra$ components of length $2^{k-1}$, 
  of the form $(u,0)$ for the first action, or of the form $(a,b)$ for the second and third actions,
  each individual action produces an independent set of $\mathbb R\la\mathbb B\ra$ components of length $2^k$;

  (ii) Given any arbitrary length $2^{k-1}$ non-zero inputs, of the form $(u,0)$ for the first action or of the form $(a,b)$ 
  for the second and third actions, the set of three actions generate independent $\mathbb R\la\mathbb B\ra$ components of length $2^k$.  
\end{lemma}
\begin{proof}
   We focus on the effect of the actions on the $\mathbb R\la\mathbb B\ra$ components only.
   In order, consider (i). It is sufficient to prove the result for two independent inputs as the general case
   follows suit. Consider the first action and suppose $u$ and $\hat u$ are two non-trivial independent
   $\mathbb R\la\mathbb B\ra$ components. Consider an arbitrary linear combination, with scalar coefficients $\kappa_1$ and $\kappa_2$,
   of the first action applied to the input $(u,0)$ and the first action applied to the input $(\hat u,0)$.
   Set the linear combination to zero. This gives,
   $\kappa_1\cdot(u,\pm u^\dag,u,\pm u^\dag)+\kappa_2\cdot(\hat u,\pm\hat u^\dag,\hat u,\pm\hat u^\dag)=0$,
   where the right-hand side represents the zero $\mathbb R\la\mathbb B\ra$ component of the appropriate length.
   Pairing up, we observe the equation above is equivalent to $\kappa_1\cdot u+\kappa_2\cdot\hat u=0$,
   with the other pairings generating the same equation. Since by assumption $u$ and $\hat u$ are two independent
   $\mathbb R\la\mathbb B\ra$ components, the result follows. Now consider the second action.
   Suppose $(a,b)$ and $(\hat a,\hat b)$ are two non-trivial independent $\mathbb R\la\mathbb B\ra$ components.
   As above, we construct the arbitrary linear combination,
   $\kappa_1\cdot(a\pm b^\dag,b\pm a^\dag,-a\mp b^\dag,-b\mp a^\dag)
   +\kappa_2\cdot(\hat a\pm\hat b^\dag,\hat b\pm\hat a^\dag,-\hat a\mp\hat b^\dag,-\hat b\mp\hat a^\dag)=0$,
   for arbitrary scalar coefficients $\kappa_1$ and $\kappa_2$. We assume $a\neq \pm b^\dag$ and
   $\hat a\neq\pm\hat b^\dag$---we observe from our proof of Lemma\ref{lemma:nilpotentaction} that the
   second action is trivial if and only if $a=\pm b^\dag$.
   Since in the last equation the final two components generate the same equation as the first two, the last equation is equivalent to,
   $\kappa_1\cdot(a\pm b^\dag,b\pm a^\dag)+\kappa_2\cdot(\hat a\pm b^\dag,\hat b\pm a^\dag)=0$.
   This reduces to $\kappa_1\cdot(a,b)+\kappa_2\cdot(\hat a,\hat b)=0$.
   Hence by our assumption on $(a,b)$ and $(\hat a,\hat b)$, the result follows.
   We now consider the third action. 
   Suppose $(a,b)$ and $(\hat a,\hat b)$ are two non-trivial independent $\mathbb R\la\mathbb B\ra$ components,
   As above, we construct the linear combination,
   $\kappa_1\cdot(a,b,\pm b^\dag,\pm a^\dag)+\kappa_2\cdot(\hat a,\hat b,\pm\hat b^\dag,\pm\hat a^\dag)=0$,
   for arbitrary scalar coefficients $\kappa_1$ and $\kappa_2$. This last equation is equivalent to
   $\kappa_1\cdot(a,b)+\kappa_2\cdot(\hat a,\hat b)=0$---the final two components
   generate the same equation as the first two. By our independence assumption on $(a,b)$ and $(\hat a,\hat b)$,
   the result follows.

   We now consider (ii). For arbitrary $(u,0)$ and $(a,b)$ in $\mathbb R\la\mathbb B\ra$,
   consider the following linear combination of the three actions, set equal to the zero $\mathbb R\la\mathbb B\ra$ component, namely:
   $\kappa_1\cdot(u,u^\dag,u,u^\dag)+\kappa_2\cdot(a-b^\dag,b-a^\dag,b^\dag-a,a^\dag-b)+\kappa_3\cdot(a,b,-b^\dag,-a^\dag)=0$,
   where $\kappa_1$, $\kappa_2$ and $\kappa_3$ are arbitrary scalar coefficients.
   Note we assume $u$ and $(a,b)$ are non-trivial.
   If $\kappa_1\neq0$ and $\kappa_2=\kappa_3=0$ then we observe that necessarily $u$
   must be the zero $\mathbb R\la\mathbb B\ra$ component, which contradicts our assumptions.
   Similarly if $\kappa_3\neq0$ and $\kappa_1=\kappa_2=0$ then necessarily $(a,b)$ is
   the zero $\mathbb R\la\mathbb B\ra$ component, which again contradicts our assumptions.
   If $\kappa_1\neq0$, $\kappa_2\neq0$ and $\kappa_3=0$, the first and second components above
   reveal that necessarily $\kappa_1\cdot u+\kappa_2\cdot(a-b^\dag)=0$ and $\kappa_1\cdot u^\dag+\kappa_2\cdot(b-a^\dag)=0$.
   Taking the adjoint of the second equation and adding the result to the first equation, implies $u=0$.
   The third and fourth components generate the same information. Thus we have a contradiction.
   Analogously, in the cases $\kappa_1\neq0$, $\kappa_3\neq0$ and $\kappa_2=0$, as well as
   $\kappa_2\neq0$, $\kappa_3\neq0$ and $\kappa_1=0$, it is straightforward to show that
   a necessary consequence is that $(a,b)$ is the zero $\mathbb R\la\mathbb B\ra$ component,
   and we have a contradiction. Now consider the case when all of $\kappa_1$, $\kappa_2$, $\kappa_3$
   are non-zero. Pairing up the first component from the linear combination above with the
   adjoint of the fourth component reveals that necessarily $a$ is the zero $\mathbb R\la\mathbb B\ra$ component.
   Pairing the second component and the adjoint of the third component reveals that necessarily $b$ 
   is the zero $\mathbb R\la\mathbb B\ra$ component. We thus reach another contradiction.
   The final case we have not considered is the case $\kappa_1=\kappa_3=0$ and $\kappa_2\neq0$. In this case we necessarily
   deduce $a=b^\dag$. As we have seen above, this is precisely the condition we need to rule
   out for the input when we apply the second action. The proof is complete.  \qed
\end{proof}
Putting the results of this and the previous steps together, we observe the following.
\begin{proposition}[Full rank linear system for all compositions]\label{prop:fullrank}
Suppose we are given a composition $a_1\cdots a_k\in\mathcal C$ of $k$-parts.
Assume that all of $a_1$ through to $a_{k-1}$ are not equal to `$1$'. 
Then associated with the $2^k$ set of basis elements with composition component $a_1\cdots a_k$,
are a unique set of $2^k$ generators, and the signature coefficient matrix has full rank.
If $a_k=1$, the statement still holds but instead with $2^{k-1}$ basis elements and $2^{k-1}$ generators.
\end{proposition}
\medskip

\emph{Step~7: Composition and generator counts.}
In light of Proposition~\ref{prop:fullrank}, we are interested in the following counts.
Given $n\in\mathbb N$, what are the total numbers of: (i) Generators; (ii) Basis elements
with composition components avoiding `$1$', i.e.\/ compositions $a_1\cdots a_k$ for which
none of the letters $a_1$ through to $a_k$ are `$1$';
(iii) Basis elements with compositions ending in `$1$' with rest of the composition avoiding `$1$',
i.e.\/ compositions of the form $a_1\cdots a_{k-1}1$ for which $a_1\cdots a_{k-1}$ avoids `$1$'. 
Such information will help us keep track of size of the linear system of equations for the unknown coeffcients $\{c_\star\}$ we solve
as part of the proof of Theorem~\ref{thm:main} below in Step~8.

Let us begin with the total number of generators, i.e.\/ item (i). Note that all monomial generators
are of odd-degree. For a given $n\in\mathbb N$, there is one generator of degree $1$, namely $[\bs{n}]$.
We then need to enumerate all the possible degree $3$ generators. Each P\"oppe product corresponds to
increasing the order of the compositions they generate by $1$, and there are two P\"oppe products in any
degree $3$ generator. Hence all the degree $3$ generators consist of all the possible compositions
of $(n-2)$ with $1$, $2$ and $3$ parts and all the possible ways to assort them into three factors
which can include ``packing factors'' of $0$. So for example, the only $1$-part compositions of $(n-2)$
assorted in this way are $[\bs{n-2}]\,[\bs0]\,[\bs0]$, $[\bs0]\,[\bs{n-2}]\,[\bs0]$ and $[\bs0]\,[\bs0]\,[\bs{n-2}]$.
The first set of generators of this form associated with the $2$-part compositions of $(n-2)$ are
$[\bs{n-3}]\,[\bs1]\,[\bs0]$, $[\bs{n-3}]\,[\bs0]\,[\bs1]$ and $[\bs0]\,[\bs{n-3}]\,[\bs1]$, and so forth.
These are just the \emph{weak compositions} of $(n-2)$ into three parts.
The next set of generators are those of degree $5$, and all the generators of this degree would
consist of the weak compositions of $(n-4)$ with $5$ parts, and so forth.
The number of weak compositions of $m$ into $k$ parts is $m+k-1$ choose $k-1$.
We can also think of this as the number of ways of distributing $m$ balls into $k$ slots, allowing empty slots.
From our discussion above, when $n$ is odd, we see that we are interested in, for $k=1,3,5,\ldots,n$,
the number of ways of distributing $m=n-k+1$ balls into $k$ slots, or in other words,
\begin{equation*}
\sum_{k=1\text{($k$ odd)}}^{n} \begin{pmatrix} n\\ k-1\end{pmatrix}=\sum_{\ell=0}^{\frac12(n-1)}\begin{pmatrix} n\\ 2\ell\end{pmatrix}=2^{n-1},
\end{equation*}
where we use the substitution $k=2\ell+1$ for the second sum. That the sum total on the right equals $2^{n-1}$ follows, with some care,
from the corresponding result in Lemma~\ref{lemma:magicalmatch2}. Similarly, when $n$ is even,
we are interested in, for $k=1,3,5,\ldots,n+1$, the number of ways of distributing $m=n-k+1$ balls into $k$ slots, or in other words,
\begin{equation*}
1+\sum_{k=1\text{($k$ odd)}}^{n-1} \begin{pmatrix} n\\ k-1\end{pmatrix}=1+\sum_{\ell=0}^{\frac12(n-2)}\begin{pmatrix} n\\ 2\ell\end{pmatrix}=2^{n-1}.
\end{equation*}
Again we used the substitution $k=2\ell+1$ for the second sum. Note that the initial `$1$' in the sum corresponds to the
case $k=n+1$, i.e.\/ corresponding to the monomial generator $[\bs0]^{n+1}$. That the sum total is $2^{n-1}$ again follows
from the first result in Lemma~\ref{lemma:magicalmatch2}. We have thus established the following.
\begin{lemma}[Total number of generators]\label{lemma:gentot}
Given $n\in\mathbb N$, the total number of monomial generators is equal to $2^{n-1}$.
\end{lemma}

The number of compositions of $n$ with $k$-parts is $n-1$ choose $k-1$, and accumulating
these coefficients from $k=1$ to $k=n$, the total number of compositions of $n$ is $2^{n-1}$.
Of course associated with each composition, there are one or more basis elements.
For example for a composition $a_1\cdots a_k$ with $k$-parts which avoids $1$, there are $2^{k}$ associated basis elements.
The number of compositions of $n$ into $k$ parts avoiding $1$ is given by,
\begin{equation*}
\begin{pmatrix} n-k-1\\k-1\end{pmatrix}.
\end{equation*}
To see this we observe the following---see Axenovich and Ueckerdt~\cite[p.~24]{AU} or Beck and Robbins~\cite{BR}.
There is a bijection between: (a) the arrangements of $n$ balls into $k$ slots with each slot containing two or more balls;
and (b) the arrangements of $n-k$ balls into $k$ slots with no empty slots. For the map from (a) to (b),
we simply remove one ball from each slot. For the map from (b) to (a) we just add one ball to each slot.
The count for (b) is $n-k-1$ choose $k-1$, giving the result above. Thus, given that for a given composition with $k$ parts that
avoids `$1$' there are $2^k$ corresponding basis elements, the total number of basis elements
with composition components that avoid `$1$' is given, with $\lambda=2$, respectively when $n$ is odd and then when $n$ is even, by,
\begin{equation*}
  p(n;\lambda)\coloneqq\sum_{k=1}^{(n-1)/2}\begin{pmatrix} n-k-1\\k-1\end{pmatrix}\lambda^k
    \qquad\text{and}\qquad
  p(n;\lambda)\coloneqq\sum_{k=1}^{n/2}\begin{pmatrix} n-k-1\\k-1\end{pmatrix}\lambda^k. 
\end{equation*}
By direct enumeration, as well as from Examples~\ref{ex:secondorder}--\ref{ex:fifthorder},
we know $p(2;2)=p(3;2)=2$, $p(4;2)=6$ and $p(5;2)=10$. 
In fact, in general, we have the following result.
\begin{lemma}[Weighted compositions avoiding `$1$' count]\label{lemma:wildcount}
  Given an integer $n\geqslant2$ and a real number $\lambda>0$,
  the weighted sum $p=p(n;\lambda)$ of the total number of basis elements
  with composition components which avoid `$1$' satisfies the weighted Fibonacci sequence satisfying,
  \begin{equation*}
      p(n;\lambda)=p(n-1;\lambda)+\lambda\,p(n-2;\lambda).
  \end{equation*}
  In particular, when $\lambda=2$, $p(2;2)=p(3;2)=2$ and $p(n;2)=\frac23(2^{n-1}+(-1)^n)$.
\end{lemma}
\begin{proof}
  By direct computation, for $n$ odd, we observe, $p(n-1;\lambda)+\lambda\,p(n-2;\lambda)$ equals,
  \begin{align*}
   \sum_{k=1}^{(n-1)/2}&\begin{pmatrix} n-k-2\\k-1\end{pmatrix}\lambda^k+\sum_{k=1}^{(n-3)/2}\begin{pmatrix} n-k-3\\k-1\end{pmatrix}\lambda^{k+1}\\
       &=\begin{pmatrix} n-3\\0\end{pmatrix}\lambda+\sum_{k=2}^{(n-1)/2}\Biggl(\begin{pmatrix} n-k-2\\k-1\end{pmatrix}
         +\begin{pmatrix} n-k-2\\k-2\end{pmatrix}\Biggr)\lambda^k,  
  \end{align*}
  which equals $p(n;\lambda)$ once we combine the two terms in the coefficient of $\lambda^k$ shown and observe that
  the coefficient of the $\lambda$ term is one. Note that in the
  first step we made the change of variables $\ell=k+1$ in the second sum, before relabelling $\ell$ as $k$.
  When $n$ is even, we similarly observe that $p(n-1;\lambda)+\lambda\,p(n-2;\lambda)$ equals,
  \begin{align*}
   \sum_{k=1}^{(n-2)/2}&\begin{pmatrix} n-k-2\\k-1\end{pmatrix}\lambda^k+\sum_{k=1}^{(n-2)/2}\begin{pmatrix} n-k-3\\k-1\end{pmatrix}\lambda^{k+1}\\
       &=\begin{pmatrix} n-3\\0\end{pmatrix}\lambda+\sum_{k=2}^{(n-2)/2}\Biggl(\begin{pmatrix} n-k-2\\k-1\end{pmatrix}
         +\begin{pmatrix} n-k-2\\k-2\end{pmatrix}\Biggr)\lambda^k+\begin{pmatrix} n/2-2\\ n/2-2\end{pmatrix}\lambda^{n/2},  
  \end{align*}
  which equals $p(n;\lambda)$ once we combine the terms in the coefficient of $\lambda^k$, and
  observe that the coefficients of the $\lambda$ and $\lambda^{n/2}$ terms are one.
  We also used the same change of variables in the first step. 
  The final statement specific to $\lambda=2$, follows directly by solving the difference
  equation for $p=p(n;2)$ for the initial conditions indicated.
\qed
\end{proof}
The result of Lemma~\ref{lemma:wildcount} provides an answer to item (ii), stated at the beginning of this step.
Item (iii) is now straightforward. The total number of compositions with $k$-parts of the form $a_1\cdots a_{k-1}1$
for which $a_1\cdots a_{k-1}$ avoids `$1$' is simply $n-k-1$ choose $k-2$.
This is because here we require $n-1$ balls to fit into $k-1$ slots with each slot containing two or more balls.
Hence the total number of basis elements associated with composition components which end in `$1$', but avoid `$1$' elsewhere,
when $n$ is odd so $n-1$ is even, is given by,
\begin{equation*}
  \sum_{k=2}^{(n+1)/2}\begin{pmatrix} n-k-1\\k-2\end{pmatrix}\lambda^{k-1}=\sum_{k=1}^{(n-1)/2}\begin{pmatrix} n-k-2\\k-1\end{pmatrix}\lambda^{k},
\end{equation*}
which equals $p(n-1;\lambda)$. When $n$ is even and thus $n-1$ is odd, the same count is,
\begin{equation*}
  \sum_{k=2}^{n/2}\begin{pmatrix} n-k-1\\k-2\end{pmatrix}\lambda^{k-1}=\sum_{k=1}^{(n-2)/2}\begin{pmatrix} n-k-2\\k-1\end{pmatrix}\lambda^{k},
\end{equation*}
which equals $p(n-1;\lambda)$. Finally we observe the following.
\begin{lemma}[Basis element count]\label{lemma:basiselementoverallcount}
  The total number of basis elements with composition components which avoid `$1$', or end in `$1$' and avoid `$1$' elsewhere,
  equals $2^{n-1}$.
\end{lemma}
\begin{proof}
  The count in question is $p(n;2)+p(n-1;2)$. Using the explicit soluton for $p(n;2)$ given in Lemma~\ref{lemma:wildcount},
  the result follows.
\qed
\end{proof}
Combining the results of Lemmas~\ref{lemma:gentot} and \ref{lemma:basiselementoverallcount}, we deduce the
rather remarkable fact:
\begin{quote}
  For any given $n\in\mathbb N$, the total number of basis elements whose composition component avoids `$1$', or ends in `$1$' but avoids `$1$' elsewhere,
  exactly equals the total number of generators.
\end{quote}
Naturally this result is important in our proof of Theorem~\ref{thm:main} just below. We remark that the inclusion of
basis elements whose composition components end in `$1$' but avoid `$1$' elsewhere, rather than composition components
containing `$1$' at some other position with avoidance elsewhere, is merely a consequence of the ordering we have imposed,
namely the descent order. 
\medskip

\emph{Step~8: Proof of Theorem~\ref{thm:main}.}
In this last step we provide the overall proof of our main theorem. We combine together the knowledge
we gained in Steps~1--7. This final stage of the argument, though significantly adapted, is analogous to that outlined for
the non-commutative potential Korteweg--de Vries hierarchy in Malham~\cite{Malham:KdVhierarchy}.
\begin{proof}[of Theorem~\ref{thm:main}]
  To complete the proof, we essentially construct a table of signature coefficients for the arbitrary order $n$ case,
  much like Tables~~\ref{table:NLS4}--\ref{table:NLS5b} for the $n=4$ and $n=5$ cases.
  Indeed we refer to these example tables to demonstrate examples of the general procedure.
  We already know from Steps~1-7 that we can construct a linear algebraic equations of the form $AC=B$,
  where the vector $C$ of length $2^{n-1}$ lists the unknown coefficients of the generator monomials
  in the P\"oppe polynomial $\pi_n=\pi_n([\bs0],[\bs1],\ldots,[\bs n])$. The vector $B$, whose length exceeds $2^{n-1}$,
  is a vector of zeros apart from a single non-zero value `$1$' in the first position when $n$ is odd and
  in the second position when $n$ is even. This is due to the ordering we impose which we outline briefly now,
  and in some more detail just below. 
  The coefficients of $C$, are ordered according to the descent order and blocks as outlined in Steps~3 and 5.
  The signature coefficient matrix $A$ has a lower triangular block form.
  It has $2^{n-1}$ columns and its total number of rows exceeds $2^{n-1}$, though equals the length of $B$.
  The columns of $A$ are parametrised by the blocks of monomial generators mentioned, or equivalently the order of the coeffcients in $C$.
  The rows of $A$ are parametrised by the basis elements, characterised by the composition components of the basis elements listed in descent order,
  and within individual composition $w$, the basis elements are listed according to the binary order of the $\mathbb R\la\mathbb B\ra$-components.
  The number of such $\mathbb R\la\mathbb B\ra$-components is $2^{|w|-\upsilon(w)}$, where $\upsilon(w)$ counts the number of $1$'s in the
  composition $w$.

  Let us outline the forms of $C$ and $A$ in some more detail. We can be brief as much of the procedure has already been outlined in Steps~1--7. 
  Corresponding to the pair of basis elements with a one-part composition component, namely $[n]\times(1,0)$ and $[n]\times(0,1)$,
  the first pair of coefficients in $C$ are $c_n$ and $c_{0(n-2)0}$, corresponding to the generators $[\bs n]$ and $[\bs0]\,[\bs{n-2}]\,[\bs0]$.
  The corresponding $2\times2$ top-left block in $A$ is the matrix $A_0$. All the remaining entries in the first two rows of $A$ to the right
  of this block are zero. We move onto basis elements whose composition components have $k=2$ parts. In descent order the first set of
  basis elements consists of $[(n-1)1]\times(1,0,0,0)$ and $[(n-1)1]\times(0,0,0,1)$. We know from Example~\ref{ex:twopartcomps} in Step~5,
  these are the only two relevant basis elements for corresponding to compositions ending in `$1$'. The corresponding coefficients in $C$
  are $c_{(n-2)00}$ and $c_{0(n-3)1}$ associated with the generators $[\bs{n-2}]\,[\bs0]\,[\bs0]$ and $[\bs0]\,[\bs{n-3}]\,[\bs1]$.
  The corresponding coefficient matrix is the block $A_1$ taking up rows and columns $3$ and $4$ in $A$, with the remaining entries
  in rows $3$ and $4$ to the right of this block being zero. This first set of blocks is important once we start addressing the issue
  of the uniqueness of the solution $C$ to the linear system, or equivalently the consistency of the overall linear system.
  Next we consider the blocks of basis elements parametrised by two-part composition components of the form $a_1a_2$
  in descent order, with neither $a_1$ nor $a_2$ equal to zero. Associated with any such composition, there are four
  basis elements $[a_1a_2]\times\bs\beta_i$, $i=1,2,3,4$, and the corresponding coefficients in $C$
  are $c_{(a_1-1)(a_2-1)0}$, $c_{(a_1-1)0(a_2-1)}$, $c_{0(a_1-2)a_2}$ and $c_{0(a_1-2)0(a_2-2)0}$,
  corresponding to the monomial generators outlined in Example~\ref{ex:twopartcomps}.
  For each block of four basis elements with such a composition component, the corresponding signature coefficient matrix is $A_2$.
  As we run through the compositions $a_1a_2$ which avoid `$1$' in descent order,
  the coefficient matrix $A_2$ occupies rows and columns $4(n-a_1-1)+i$ for $i=1,2,3,4$ in the signature coefficient matrix $A$.
  All the entries in $A$ in the these rows to the right of these blocks are zero.
  We know from Step~5, that the basis elements with the two-part composition component $1(n-1)$ does not occupy any new
  columns in $A$, but instead, the two rows corresponding to the two basis elements concerned only contain
  non-zero entries in the columns parametrised by $4(n-a_1-1)+i$ for $i=1,2,3,4$ for $a_1\neq1$.

  We move onto blocks of basis elements corresponding to composition components with $k=3$ parts. We know from
  our arguments in Step~5 that we can focus on composition components of the form $a_1a_2a_3$ which avoid `$1$'
  or end in `$1$' and avoid `$1$' elsewhere. Basis elements with composition components lying in the complement of
  this set do not generate ``new'' columns, or equivalently, the non-zero entries in the rows in $A$ corresponding
  to these basis elements only occupy columns we have already encountered/parametrised. Let us examine the
  blocks we generate as we run through the compositions $a_1a_2a_3$ which avoid `$1$' or end in `$1$' and avoid `$1$' elsewhere.
  We know from Step~5, that for each of the former compositions we generate the block matrix $A_3$,
  modified to include the column factors mentioned in Step~5, associated with the $8$ corresponding basis elements
  and the $8$ ``new'' monomial generator columns. For each of the latter compositions we generate a block matrix $A_3^\prime$ 
  associated with the $4$ corresponding basis elements and the $4$ ``new'' monomial generator columns. The block matrices
  $A_3$ and $A_3^\prime$ are diagonal blocks in $A$, with all entries in the corresponding rows they occupy to the right
  of these blocks equal to zero.

  We know from our analysis in Step~5 and in particular Lemma~\ref{lemma:uniquegenerators} and the discussion immediately following this
  lemma, that we have the following. For every composition $w$ of $n$ that avoids `$1$', or ends in `$1$' but avoids `$1$' elsewhere, we can construct
  a $2^{|w|}\times2^{|w|}$ block matrix $A_{w}$ in the former case, or a $2^{|w|-1}\times2^{|w|-1}$ block matrix $A_w^\prime$ in the latter case,
  which occupies a distinct diagonal block in $A$. Here by ``distinct'', we mean that its rows and columns do not coincide with the rows
  and columns of any of the analogous block matrices corresponding to basis elements with such composition components. All the entries
  in $A$ in the rows occupied by these blocks and to the right of them, are zero. Thus for general $n$, the signature coefficient matrix $A$
  is indeed a lower block triangular matrix. Further, from our results in Step~6, we know that each of the block matrices $A_w$ and $A_w^\prime$
  has full rank. Even further, from our results in Step~7, we know that the total number of basis elements corresponding to such
  composition components and generating the rows of such diagonal blocks, exactly equals the total number of monomial generators
  generating the columns of such diagonal blocks. This means that if we ignore the basis elements with composition components
  in the complementary set for the moment, then we can proceed block by block, starting with $A_0$, and solve for the corresponding
  set of P\"oppe polynomial coefficients in $C$, until we precisely exhaust the blocks and uniquely recover $C$.

  \begin{table}
  \caption{The top left block entries in the signature coefficient matrix $A$
    at any order $n$, depending on whether $n$ is odd (top) or $n$ is even (bottom). 
    The coefficients are the $\chi$-images of the signature entries shown.
    The forms shown for the case of when $n$ is odd or even, are used to prove the
    consistency of the overdetermined linear system of algebraic equations for
    the P\"oppe polynomial coefficients. The first rows are not shown.}
    \label{table:consistency}
  \begin{center}
  \begin{tabular}{|l|cccc|}
  \hline\hline
   $n$ odd $\phantom{\biggl|}$  \hspace{-0.35cm}                        & $[\bs n]$ & $[\bs{n-2}]\,[\bs0]^2$               & $[\bs0]\,[\bs{n-3}]\,[\bs1]$     & $\cdots$ \\ 
  \hline
  $[\bs0(n-1)\bs0 1\bs0]$ $\phantom{\Bigl|}$ \hspace{-0.35cm}           & $(n-1)1$  & $2\cdot\bigl((n-2)\ot0\ot0\bigr)$   & $2\cdot\bigl(0\ot(n-3)\ot1\bigr)$ & \\
  $[\bs0(n-1)\bs0^\dag 1\bs0^\dag]$ $\phantom{\Bigl|}$ \hspace{-0.35cm}  &           & $-2\cdot\bigl((n-2)\ot0\ot0\bigr)$  & $2\cdot\bigl(0\ot(n-3)\ot1\bigr)$ & \\  
  \hline\hline
  $~~\vdots$ $\phantom{\Bigl|}$  \hspace{-0.35cm} &&&&\\
  \hline\hline
  $n$ even  $\phantom{\biggl|}$  \hspace{-0.35cm}  & $[\bs n]$ & $[\bs0]\,[\bs{n-2}]\,[\bs0]$    & $[\bs{n-2}]\,[\bs0]^2$ & $[\bs0]\,[\bs{n-3}]\,[\bs1]$ \\
  \hline
  $[\bs0 n\bs0]$  $\phantom{\Bigl|}$  \hspace{-0.35cm}  & $n$       & $2\cdot\bigl(0\ot(n-2)\ot0\bigr)$ & & \\
  $[\bs0(n-1)\bs0 1\bs0]$  $\phantom{\Bigl|}$ \hspace{-0.35cm} & $(n-1)1$ & $2\cdot\bigl(0\ot(n-2)\ot0\bigr)$ & $2\cdot\bigl((n-2)\ot0\ot0\bigr)$ & $2\cdot\bigl(0\ot(n-3)\ot1\bigr)$\\
  $[\bs0(n-1)\bs0^\dag 1\bs0^\dag]$ $\phantom{\Bigl|}$ \hspace{-0.35cm} && $2\cdot\bigl(0\ot(n-2)\ot0\bigr)$ & $-2\cdot\bigl((n-2)\ot0\ot0\bigr)$& $2\cdot\bigl(0\ot(n-3)\ot1\bigr)$\\  
  \hline\hline
  \end{tabular}
  \end{center}
  \end{table}

  The rest of the proof is now concerned with demonstrating the consistency of the remaining rows/linear equations for the coefficients
  in $C$ associated with those basis elements with composition components that contain a `$1$', not including the instances of
  compositions ending in `$1$', but avoiding it elsewhere. We heavily rely on the fact that the linear system of algebraic equations for $C$
  is almost homogeneous apart from the single unit entry in the first position when $n$ is odd and in the second position when $n$ is even.
  The first phase of this section of the proof focuses on the blocks associated with basis elements with $1$ and $2$-part composition components.
  Assume for the moment that $n$ is odd. In this instance the first two linear equations for the coefficients $c_n$ and $c_{0(n-2)0}$
  are $c_n+2\cdot c_{0(n-2)0}=1$ and $-2\cdot c_{0(n-2)0}=0$. Thus when $n$ is odd, we always have $c_{0(n-2)0}=0$. This means that all
  the entries in the second column of the signature coefficient matrix $A$ are not relevant to the linear system of equations
  for $C$ and we thus eliminate this column in $A$ when $n$ is odd. Implementing this, and ignoring the first row in $A$ corresponding
  to the basis element $[\bs0 n\bs0]$, the top left corner of $A$ has the form shown in the top part of Table~\ref{table:consistency}.
  Now assume that $n$ is even. The first two linear for the coefficients $c_n$ and $c_{0(n-2)0}$ in this instance are
  are $c_n+2\cdot c_{0(n-2)0}=0$ and $-2\cdot c_{0(n-2)0}=1$. When $n$ is even we swap over the first two rows in the 
  signature coefficient matrix $A$ which is equivalent to swapping the order of the first two linear equations shown. 
  If we ignore the \emph{new} first row of $A$ corresponding to the nonhomogeneous equation $-2\cdot c_{0(n-2)0}=1$,
  then the top left corner of $A$ has the form shown in the bottom part of Table~\ref{table:consistency}. 
  All the remaining rows and columns in the signature coefficient matrix $A$, whether $n$ is odd or even, remain
  the same. However, with the first rows ignored, the system of linear equations that remains is homogeneous.
  And, of course, the top left corners in the case that $n$ is odd or even have the forms shown in Table~\ref{table:consistency}. 
  In both cases in Table~\ref{table:consistency} we trace a diagonal starting from the top left non-zero  
  coefficient, which is $\chi\bigl((n-1)1\bigr)=n$ when $n$ is odd, and $\chi(n)=1$ when $n$ is even.
  When $n$ is even the next two entries along this diagonal are $\chi\bigl(2\cdot(0\ot(n-2)\ot0)\bigr)=2$ and 
  $\chi\bigl(-2\cdot((n-2)\ot0\ot0)\bigr)=-2$. When $n$ is odd the next diagonal entry is $\chi\bigl(-2\cdot((n-2)\ot0\ot0)\bigr)=-2$.
  Thereafter the diagonal entries for the cases of $n$ being even or odd are the same.
  In Table~\ref{table:NLS4}, when $n$ is even, after swapping the first two rows, though not ignoring the new first row yet,
  we can view the diagonal we have identified as the diagonal just below the leading diagonal. Similarly in Tables~\ref{table:NLS5a}
  and \ref{table:NLS5b}, when $n$ is odd, after eliminating the second column but still retaining the first row, we can again view
  the diagonal we have identified as that just below the leading diagonal in those Tables.
  In either Table~\ref{table:NLS4} or in Tables~\ref{table:NLS5a} and \ref{table:NLS5b}, let us call the diagonal just
  below the leading diagonal the `sub-diagonal'. Consider for example Tables~\ref{table:NLS5a} and \ref{table:NLS5b}.
  If we follow the sub-diagonal with the view of retaining non-zero terms along it, we observe we meet a obstruction
  in the first $4\times4$ block with matrix $A_2$ characterised by the composition component $32$. The problem is
  that while the sub-diagonal of $A_2$ has non-zero entries, the next term along the diagonal that lies in the last column
  of that $A_2$ block, but beneath the entire block, and in fact the row corresponding to the basis element $[23]\times(1,0,0,0)$
  is zero. However there is a quick fix to this obstruction. That is to simply swap the two columns in the signature coefficient matrix
  $A$ corresponding to the final two columns of the $A_2$ block characterised by the composition component $23$.
  Such a swap simply corresponds to changing the order of the monomial generators.
  We can see from Tables~\ref{table:NLS5a} and \ref{table:NLS5b} that this column-swap procedure would guarantee
  that the next entry in the sub-diagonal would be non-zero. We can then continue to consider the sub-diagonal
  in the $A_2$ block corresponding to the basis elements $[23]\times\bs\beta_i$ for $i=1,2,3,4$. However we observe
  a similar obstruction necessitating an anologous swap of the columns of $A$ corresponding to the final two
  columns of this second $A_4$ block. This is again enacted to ensure that the term in the final column
  immediately below this $A_4$ block is non-zero. The term in question corresponds to the signature coefficient
  $\chi(0\ob0\ot3)$ in the row corresponding to $[14]\times(1,0,0,0)$. We have thus established, for all the
  basis elements with $1$ and $2$-part composition components, a complete diagonal with all entries non-zero, and which can act as pivots.
  Since all the linear equations corresponding to the rows we are considering (we are ignoring the top row) are homogeneous
  we can use Gaussian elimination to render all the entries in the columns below the sub-diagonal to be zero.
  
  The procedure for the case of general $n$ proceeds in exactly the same manner as the $n=5$ case we have just
  outlined, except that now we need to establish that when we swap the columns over as just outlined, we are
  guaranteed a non-zero entry in the corresponding sub-diagonal entry. We also need to guarantee that the
  entry immediately below the final column of the $2\times2$ block $A_1$ corresponding to the rows $[(n-1)1]\times(1,0,0,0)$
  and $[(n-1)1]\times(0,0,0,1)$ is also non-zero. This case corresponds to the column given by $[\bs0]\,[\bs{n-3}]\,[\bs{1}]$.  
  In the other cases of the $A_4$ blocks corresponding to the basis elements $[(n-m)m]\times\bs\beta_i$ for $i=1,2,3,4$,
  the column in question is the third column in the $A_4$ block corresponding to the column given by $[\bs0]\,[\bs{n-m-2}]\,[\bs{m}]$.
  In particular this means we can treat the $m=1$ and $m=2,\ldots,n-2$ cases simultaneously. Indeed, using Corollary~\ref{cor:specialactions}
  in Step~4, we observe that at leading order we have,
  \begin{equation*}
    [\bs0]\,[\bs{n-m-2}]\,[\bs{m}]=[(n-m)m]\times(2,0,0,2)+[(n-m-1)(m+1)]\times(1,1,1,1)+\cdots,
  \end{equation*}
  where we have used that at leading order $[\bs{m}]=[m]\times(1,0)+\cdots$. The first term on the
  right is the term we expect at leading order for this generator, while the second term on the right 
  is the column corresponding to the rows in the next block down---we see that $(n-m-1)(m+1)$ is obtained
  from $(n-m)m$ by replacing $m$ by $m+1$, which is one composition further down in descent order.
  Thus indeed we are guaranteed that the next entry in the sub-diagonal is non-zero when we enact the column swap. 
  Further we observe that in the case of the final $A_4$ block corresponding to the composition $2(n-2)$
  for which $m=n-2$, the corresponding generator is $[\bs0]\,[\bs0]\,[\bs{n-2}]$ while the corresponding
  row containing the sub-diagonal entry of interest is $[1(n-1)]\times(1,0,0,0)$. At leading order we have,
  \begin{equation*}
    [\bs0]\,[\bs0]\,[\bs{n-2}]=[2(n-2)]\times(2,0,0,2)+[1(n-1)]\times(2,2,0,0)+\cdots,
  \end{equation*}
  which thus guarantees a final sub-diagonal non-zero entry. We are thus in the exact same situation
  as described for the $n=5$ case just above, and we can use the sub-diagonal entries as pivots
  to render all the entries in $A$, in all the sub-diagonal columns, below the sub-diagonal to be zero.
   
  The second phase of this section of the proof now focuses on all the blocks associated with 
  basis elements with composition components with $k$-parts with $k\geqslant3$. This phase is
  more straightforward. Let us focus on the $3$-part composition cases to begin with.
  The first $3$-part composition in descent order is $(n-2)11$, and we know from Step~5 that there
  are no ``new'' generators associated with with any such composition that lies in the set of
  compositions complementary to those avoiding `$1$' or ending in `$1$', but avoiding it elsehwere.
  Hence we can use Gaussian elimination, using the pivots from the
  sub-diagonal outlined for the $1$ and $2$-part composition cases just outlined,
  to render the entries in for the two rows/basis elements concerned here, 
  $[(n-2)11]\times\bs\beta_1$ and $[(n-2)11]\times\bs\beta_8$, equal to zero.
  Next we consider the block of rows/basis elements corresponding to the composition $(n-3)21$.
  As outlined in Step~5 this block is associated with $4$ generators and the block matrix $A_3^\prime$. 
  We know this has full rank and we can thus use the leading diagonal as pivots to render
  all entries in the corresponding columns of $A$ below this diagonal to be zero.
  The next block is associated with the composition $(n-3)12$, which with a `$1$' in the
  middle is not associated with any new generators, and from our Gaussian elimination processes
  thus far has all row entries rendered zero. The next blocks are associated with the compositions $(n-4)31$,
  $(n-4)22$, $(n-4)14$. The first two of these compositions are associated with separate copies of the $8\times8$ matrix $A_3$ (with
  the columns mentioned in Step~5 suitably scaled) and a total of $16$ generators (one set of $8$ each). The leading diagonals 
  of both copies of $A_3$ can again be used as pivots to render all the entries, below this diagonal in the columns of $A$ associated
  with these two copies, equal to zero. The entries in the rows associated with the third composition $(n-4)13$ will have been
  rendered zero in the Gaussian elimination process just outlined for the other two compositions. And so forth, we can see
  that we can proceed in descent order through the blocks associated with $3$-part compositions, either in the case of
  compositions that avoid `$1$' or end in `$1$' but avoid it elsewhere, using the diagonals of the blocks associated
  with $A_3$ or $A_3^\prime$, to render the corresponding entries in $A$ below these diagonals to be zero, or recognising
  for the blocks associated with the complementary set of compositions, the entries in the rows of those block will
  already be rendered zero. The procedure for all further blocks associated with compositions of $4$ or more parts
  proceeds exactly analogously. Naturally that the corresponding diagonals with non-zero entries exist for all
  blocks associated with compositions that avoid `$1$' or end in `$1$' but avoid it elsewhere, is guaranteed
  by the results in Step~6, in particular Proposition~\ref{prop:fullrank}.

  Hence we have thus rendered all the entries in all the rows corresponding to basis elements with
  composition components which lie in the set complementary to those that avoid `$1$' or end in `$1$' but avoid it elsewhere, 
  equal to zero. Briefly returning to the rows/blocks associated with the $1$ and $2$ part compositions, still
  ignoring the top row as indicated in Table~\ref{table:consistency}. A quick count reveals that when $n$ is even,
  we have $3+4(n-1)$ rows, i.e.\/ homogneous linear equations, in $4+4(n-1)$ unknowns, while when $n$ is odd,
  we have $2+4(n-1)$ homogeneous linear equations, in $3+4(n-1)$ unknowns. In either case when $n$ is even or odd,
  proceeding through all the other blocks associated with compositions of three or more parts, the remaining
  number of homogeneous linear equations equals the remaining number of unknowns (as outlined in the first section of this proof).
  Hence, in either case when $n$ is even or odd, we can solve the entire system of linear homogeneous equations,
  with one less equation then the total number of unknowns, to find expressions for all the unknowns in terms of only one of them.
  In the case that $n$ is odd, we solve for all of them in terms of $c_n$. In the case that $n$ is even, we solve for
  all of them in terms of $c_{0(n-2)0}$.
  We now re-introduce the very first row we ignored at the beginning of this second, ``consistency'', section of the proof.
  When $n$ is odd, that first equation is $c_n+2\cdot c_{0(n-2)0}=1$. Since we have an expression for $c_{0(n-2)0}$ in terms of $c_n$
  from the homogeneous set of linear equations, we can substitute that expression into this non-homogeneous linear equation
  and determine $c_n$. When $n$ is even, the first equation is $-2\cdot c_{0(n-2)0}=1$ or equivalently $c_{0(n-2)0}=-1/2$.
  Since in this case we have expressions for all the other unknowns in terms of $c_{0(n-2)0}$, this fixes the values of
  all the other unknowns. In either case, whether $n$ is odd or even, we have established a unique solution $C$, and
  the proof is complete.
\qed
\end{proof}
\begin{remark}
  As mentioned, the overall proof in Step~8 just above is analogous to that for
  the non-commutative potential Korteweg--de Vries hierarchy in Malham~\cite{Malham:KdVhierarchy}.
  Therein we proceed by considering compositions with $k=1$, $k=2$, and so forth, parts.
  In that case there are no blocks as there are no $\mathbb R\la\mathbb B\ra$ components.
  The basis elements are just compositions of $n$, there are no skew forms.
  Also, as it is the potential equation, the generators are just monomials of
  the signature expansions $\bs n$, with $n\in\mathbb N$---in particular there are
  no generators corresponding to $[\bs0]$. Further, since there are no skew forms,
  the generators can be of even or odd degree. See in particular Section~6 in Malham~\cite{Malham:KdVhierarchy}.
  We observed at the end of Step~7 above that the total number of generators equalled
  the total number of basis elements with composition components that avoided `$1$' together
  with those that ended in `$1$' but avoided it elsewhere. It would be natural to
  wonder whether a similar situation occurs for the case of the non-commutative potential Korteweg--de Vries hierarchy,
  and indeed retrospectively, we can establish the exact same result for that hierarchy.
  In fact we can show, since there are no blocks and no generators akin to `$[\bs0]$', that
  for each set of compositions with $k$ parts, the number of monomial generators with $k$ factors
  equals the number of compositions (the basis elements here) that avoid `$1$' together
  with those that ended in `$1$' but avoided it elsewhere. Again, that we single out those
  compositions ending in `$1$' is just an artefact of the descent order we impose.
  To see this fact, we observe from Section~6 in Malham~\cite{Malham:KdVhierarchy},
  that the number of generators with $k$ factors, say of the form $(\bs n_1)(\bs n_2)\cdots(\bs n_k)$
  with the P\"oppe product, is given by $n-k$ choose $k-1$. This is because each P\"oppe product
  adds a `$1$' to one of the composition parts in the eventual expansion in compositions.
  The complete set of such monomial generators is exhausted by those with $k=1,2,\ldots,\frac12(n+1)$ parts.
  We already know from Step~7 above, that the number of compositions of $n$ with $k$ parts that
  avoid `$1$' equals $n-k-1$ choose $k-1$ for $k=1,2,\ldots,\frac12(n-1)$, and the number of
  compositions that end in `$1$' but avoid it elsewhere equals $n-k-1$ choose $k-2$ for $k=2,\ldots,\frac12(n+1)$.
  If we restrict ourselves to $k=2,\ldots,\frac12(n-1)$, we observe that the number of compositions
  satisfying either property is given by,
  \begin{equation*}
     \begin{pmatrix} n-k-1\\ k-1\end{pmatrix}+\begin{pmatrix} n-k-1\\ k-2\end{pmatrix}=\begin{pmatrix} n-k\\ k-1\end{pmatrix},
  \end{equation*}
  with equality following by simply adding the two relevant fractions on the left.
  The cases $k=1$ and $k=\frac12(n+1)$, for which the number of such compositions and generators is singular,
  just follows by inspection.
\end{remark}

\section{Conclusion}\label{sec:conclu}
There are many open directions of research we intend to pursue based on the combinatorial algebraic approach we introduced herein.
One direction we have not directly addressed herein is that of alternative formulations of the modified Korteweg--de Vries hierarchy members
of orders $3$, $5$ and higher. See for example Liu and Athorne~\cite{LA}, Olver and Sokolov~\cite{OS}, Oevel and Rogers~\cite{OR} and Gerdjikov~\cite{Gerdjikov}. 
For example, the alternative non-commutative modified Korteweg--De Vries equation has the form,
\begin{equation*}
\pa_t g=\pa^3g+3\bigl(g(\pa^2g)-(\pa^2g)g\bigr)-6g(\pa g)g.
\end{equation*}
Note that the polynomial partial differential field includes even degree terms. Such alternative forms can be obtained
from non-commutative modified Korteweg--De Vries equation via a suitable gauge transformation as, for example,
outlined in detail in Carillo and Schiebold~\cite{CSIII}. The combinatorial algebraic structure we have developed
would seem a natural context to investigate such alternative hierarchy forms further. Closely related is the
\emph{Miura transformation}. This is particularly simple in our context. Assuming the order $n=2m+1$ with $m\in\mathbb N$
is odd, then since $\mathcal I^2=\id$, the base dispersion equation for $P$ is,
\begin{equation*}
\pa_t P=(-1)^{m+1}\pa^{2m+1}P. 
\end{equation*}
We can assume this to be the base equation for the non-commutative potential Korteweg--de Vries hierarchy considered in Malham~\cite{Malham:KdVhierarchy}. 
In that case the solution $G^{\mathrm{pKdV}}$ is given by $G^{\mathrm{pKdV}}=P(\id-P)^{-1}$.
For the non-commutative modified Korteweg--de Vries hierarchy, we observe that when $n$ is odd we can
assume the solution $G^{\mathrm{mKdV}}$ to have the form $G^{\mathrm{mKdV}}=2P(\id+P)^{-1}(\id-P)^{-1}$,
i.e.\/ replacing the `$\mathrm{i}P$' everywhere simply by $P$, and all our results in Section~\ref{sec:ncNLS} and thereafter follow through. 
This is because we carried through the quantity `$\mathrm{i}P$' throughout our computations in Section~\ref{sec:ncNLS}
and, in particular, into our abstract encoding. For example, our computation for $\pa_t[V]$ preceding Definition~\ref{def:timederiv}
carries through with this replacement with $V\coloneqq(\id-P)^{-1}$, $P^\dag=-P$ and $V^\dag=(\id+P)^{-1}$.
For convenience we set $U^{\mathrm{pKdV}}\coloneqq(\id-P)^{-1}$ and $U^{\mathrm{mKdV}}\coloneqq(\id+P)^{-1}(\id-P)^{-1}$.
Note that by operator partial fractions we have $U^{\mathrm{pKdV}}=\id+PU^{\mathrm{pKdV}}$ so $\pa G^{\mathrm{pKdV}}=\pa U^{\mathrm{pKdV}}$.
Then as in Doikou \textit{et al.} \cite[Cor.~3.15]{DMS} we observe that since $U^{\mathrm{pKdV}}=(\id+P)U^{\mathrm{mKdV}}$ we have,
\begin{align*}
               &&\pa U^{\mathrm{pKdV}}&=\pa\bigl(PU^{\mathrm{mKdV}}\bigr)+\pa U^{\mathrm{mKdV}}\\
\Leftrightarrow&&\pa U^{\mathrm{pKdV}}&=\pa\bigl(PU^{\mathrm{mKdV}}\bigr)+U^{\mathrm{mKdV}}\pa(P^2)U^{\mathrm{mKdV}}\\
\Rightarrow    &&\pa\lb G^{\mathrm{pKdV}}\rb&=\pa\lb G^{\mathrm{mKdV}}\rb+\lb G^{\mathrm{mKdV}}\rb^2.
\end{align*}
In the last step we used the P\"oppe product rule. This represents the Miura transformation giving the connection
between the non-commutative potential, and modified, Korteweg--de Vries hierarchies. A natural question is what
the translation (likely non-trivial) of this result is at even orders?

A natural formulation for Hankel and Toeplitz operators is the $L^2$ Hardy spaces $\mathbb H_{\pm}$, corresponding to the
upper and lower half complex plane; see for example Peller~\cite{Peller}. This can be thought of as the Fourier transform
representation of the formulation we gave in Section~\ref{sec:Hankelandlinear}. Recently this context has been used to
prove interesting integrability results/connections for the cubic Szeg\"o equation, see Pocovnicu~\cite{Pocovnicu},
Grellier and Gerard~\cite{Gerard} and Gerard and Pushnitski~\cite{GP},
and to extend regularity results for the Korteweg--de Vries equation, see Grudsky and Rybkin~\cite{GRI,GRII,GRIII}.
There is a natural decomposition $L^2(\R)=\mathbb H_+\oplus \mathbb H_-$ and thus an immediate direction to pursue
would be to consider our Marchenko equation and Fredholm Grassmannian flow in this context and establish a
connection to the results of, for example, Grellier and Gerard~\cite{Gerard} and Grudsky and Rybkin~\cite{GRIII}.

At the abstract algebra level, for the skew-P\"oppe algebra
$\mathbb C[\mathbb Z_{\bs0}]\cong\mathbb C[\mathcal C]\times\mathbb R\la\mathbb B\ra$,          
there are many open questions as follows:
(i) The skew-P\"oppe algebra $\mathbb C[\mathbb Z_{\bs0}]$, endowed with the triple product in Lemma~\ref{lemma:tripleproductaction},
constitutes a \emph{triple system} or \emph{ternary algebra}. See for example Meyberg~\cite[p.~21]{Meyberg} or Ricciardo~\cite[p.~23]{Ricciardo}.
Exploring this context is very much of interest.
(ii) The P\"oppe products in Lemma~\ref{lemma:skewandsymmPoppeproducts} are quasi-Leibniz products in which the `quasi' label refers to the term additional to the
two expected Leibniz terms which essentially involves inserting a `$1$' between the two terms in the product (as well as a factor `$2$').
A natural question is, is it possible to establish an isomorphism between between this skew-P\"oppe algebra and the corresponding
skew-P\"oppe algebra endowed with the triple product based on the P\"oppe products in Lemma~\ref{lemma:skewandsymmPoppeproducts}
without the `quasi' terms? This will necessarily require a fix of the non-quasi product for low order terms, for example those
involving products with `$[\bs0]$' and so forth.
The analogy is the isomorphism between the shuffle algebra and the quasi-shuffle algebra proved by Hoffman~\cite{Hoffman}.
Establishing such an isomorphism would significantly simplify the proofs of the results herein and would help to
establish (iii) and (iv) just below. 
(iii) We observe that in our main result we sought P\"oppe polynomial expansions $\pi_n=\pi_n([\bs0],[\bs1],\ldots,[\bs n])$
for the endomorphisms $[\bs 0n\bs0]$ when $n$ is odd, and $[\bs0 n\bs0^\dag]$ when $n$ is even. 
However more generally we might ask the question of whether there exists P\"oppe polynomial expansions for any of the basis elements
in $\mathbb C[\mathbb Z_{\bs0}]$? In other words can we express any basis element in $\mathbb C[\mathbb Z_{\bs0}]$ as a linear combination
of monomials of the form $[\bs{n_1}]\,[\bs{n_2}]\,\cdots\,[\bs{n_k}]$?
(iv) A directly related broader question then is, does there exist an isomorphism between the algebra of
odd-degree monomial forms $[\bs{n_1}]\,[\bs{n_2}]\,\cdots\,[\bs{n_k}]$ with $n_i\in\mathbb N\cup\{0\}$ endowed with the concatenation product,
and the skew-P\"oppe algebra? The connection is provided by the signature expansions. The odd-degree monomial form parametrising factors $n_1n_2\cdots n_k$
are actually weak compositions.
(v) Can we establish a \emph{co-algebra} associated with the skew-P\"oppe algebra $\mathbb C[\mathcal C]\times\mathbb R\la\mathbb B\ra$?
This was achieved for P\"oppe algebra in Malham~\cite[Sec.~5]{Malham:KdVhierarchy}.
Here we have to deal with the $\mathbb R\la\mathbb B\ra$-components. Indeed, we have already started in this direction.
(vi) Establishing such a co-algebra, or at least a refined de-P\"oppe co-product $\Delta_n$ associated with $\mathbb C[\mathcal C]\times\mathbb R\la\mathbb B\ra$,
would be useful. Consider the odd-degree monomial $[\bs{n_1}]\,[\bs{n_2}]\,\cdots\,[\bs{n_k}]$ with $n_1n_2\cdots n_k$ a weak composition of $n$.
Using the signature expansions for each of the factors, $[\bs{n_1}]\,[\bs{n_2}]\,\cdots\,[\bs{n_k}]$ can be expressed in the form,
\begin{equation*}
  \sum \chi(w_1\ot w_2\ot\cdots\ot w_k)\cdot\bigl([w_1]\times\bs\beta_1(|w_1|)\bigr)\,\bigl([w_2]\times\bs\beta_1(|w_2|)\bigr)\,\cdots\,\bigl([w_k]\times\bs\beta_1(|w_k|)\bigr),
\end{equation*}
where the sum is over all basis elements $[w_i]\times\bs\beta_1(|w_i|)$ for $i=1,\ldots,k$, with $w_i\in\mathcal C(n_i)$ and
$\bs\beta_1(|w_i|)\in\mathbb R\la\mathbb B\ra$ of length $2^{|w_i|}$ with the first component equal to `$1$' as the only non-zero component. 
If we compute all the odd-degree P\"oppe products on the right, we generate the following form, 
\begin{equation*}
  \sum_{w\in\mathcal C(n)}\sum_{i=1}^{2^{|w|}}\chi_{\bs\beta_i}\Bigl(\bigl(\Delta_k([w]\times\bs\beta_i\bigr)\Bigr)\cdot\bigl([w]\times\bs\beta_i\bigr).
\end{equation*}
In this expression, the $\bs\beta_i$ are the natural basis elements of $\mathbb R\la\mathbb B\ra$ of length $2^{|w|}$,
containing a `$1$' in the $i$th position and zeros elsewhere. The combined pair of sums correspond to a sum over all possible basis elements,
for example over all the left-most column elemnts in Tables~\ref{table:NLS4} and \ref{table:NLS5a}.
The co-product $\Delta_k$ generates all forms $w_1\ot w_2\ot\cdots\ot w_k$ such that the odd-degree P\"oppe product
$\bigl([w_1]\times\bs\beta_1(|w_1|)\bigr)\,\cdots\,\bigl([w_k]\times\bs\beta_1(|w_k|)\bigr)$ generates $[w]\times\bs\beta_i$.
The homomorphic map $\chi_{\bs\beta_i}$ records the signature coefficient $\chi(w_1\ot w_2\ot\cdots\ot w_k)$ together
with the factor in the $2^{|w|}\times2^{|w|}$ block associated with the composition $[w]$ as outlined in Section~\ref{sec:hierarchycoding}.
See for example the matrix $A_3$ from that section. It records the factor associated with the $\bs\beta_i$ row and the
$[\bs{\sigma(w_1)}]\,[\bs{\sigma(w_1)}]\,\cdots\,[\bs{\sigma(w_1)}]$ column, where $\sigma(w_i)$ represents the sum of
all the factors in the composition $w_i$. Any terms resulting from the `quasi' term in the P\"oppe product are included.
Recall we can systematically generate all the blocks and all such factors
using the appropriate three standard actions---see Step~5 in Section~\ref{sec:hierarchycoding}.
Let $\mathcal C^{\ast}(n)$ denote the set of all odd-length weak compositions $v$ of $n$ such that $\sigma(v)+|v|-1=n$.
Assuming we have established such a co-product $\Delta_k$, then we observe, we can express any P\"oppe polynomial,
or even an arbitrary sum of P\"oppe polynomials, in the form,
\begin{align*}
  \sum_{n\geqslant1}\pi_n&=\sum_{w\in\mathcal C}\Pi\bigl([w]\bigr)\cdot\bigl([w]\times\bs\beta_i\bigr),
  \intertext{where}
  \Pi\bigl([w]\bigr)&=\sum_{i=1}^{2^{|w|}}\sum_{v\in\mathcal C^\ast(\sigma(w))} c_v\,\chi_{\bs\beta_i}\bigl(\Delta_{|v|}([w]\times\bs\beta_i)\bigr).
\end{align*}
Thus, in principle, we can express the whole hierarchy as the co-algebra sum $\Pi$.

\begin{acknowledgement}
  SJAM would like to thank the EPSRC for the Mathematical Sciences Small Grant EP/X018784/1. 
  It is also a pleasure to acknowledge very interesting discussions with Alexander Pushnitski and Alexei Rybkin in connection with our work herein.
\end{acknowledgement}

\section{Declarations}

\subsection{Funding and/or Conflicts of interests/Competing interests}
SJAM received funding from the EPSRC for the Mathematical Sciences Small Grant EP/X018784/1.
There are no conflicts of interests or competing interests. 

\subsection{Data availability statement}
No data was used in this work.

\end{document}